\documentclass[12pt]{article} 
\usepackage{graphicx,graphics,times,amsmath,amsfonts,amssymb,latexsym,bbm,textcomp,dsfont,hyperref}
\usepackage{comment}
\usepackage{subfigure}
\usepackage{color}
\usepackage{bm}
\newtheorem{theorem}{Theorem}
\newtheorem{definition}{Definition}
\newtheorem{proposition}{Proposition}
\newtheorem{lemma}{Lemma}
\newtheorem{corollary}{Corollary}
\newenvironment{proof}[1][Proof]{\noindent\textbf{#1.} }{\ \rule{0.5em}{0.5em}}
\newcounter{remark}
\newenvironment{remark}[1][]{\refstepcounter{remark}\par\medskip\noindent%
\textbf{Remark~\arabic{remark}.#1} \rmfamily}{$\blacksquare$\medskip}

\newcommand{\ignore}[1]{}
\renewcommand{\>}{\rangle}
\newcommand{\<}{\langle}
\newcommand{\diag}{\mathrm{diag}}
\newcommand{\rank}{\mbox{rank }}
\newcommand{\Herm}{\mbox{Herm}}
\newcommand{\Kerm}{\mbox{Kerm}}
\newcommand{\PT}{$\mathcal{PT}$}
\newcommand{\ket}[1]{\left|#1\right\rangle}

\newcommand{\Tr}{\mbox{Trace}}
\begin{document}
\title{Topological Discrimination of Steep to Supersteep Gap 
as Evidence of Tunneling in Adiabatic Quantum Processes}
\author{E. A. Jonckheere\\
Departments of Electrical and Computer Engineering \\
and Department of Mathematics, \\
University of Southern California, \\
Los Angeles, California 90089, USA}
\maketitle

\begin{abstract}
It is shown that the gap that limits the speed of a quantum annealing process can take three salient, 
stable morphologies: the {\it super-steep}, {\it steep},  and {\it mild} gaps. 
The difference is the number of pairs of inflection points (2, 1, 0, resp.) near the gap. 
Classification of the various singularities betrayed by the inflection points relies on 
the critical value curves of the quadratic numerical range mapping of the matrix $H_0+\imath H_1$,  
where $H_0$, $H_1$ are the transverse field and problem Hamiltonians, resp.  
In this representation, the ground level becomes the generically smooth boundary of the numerical range, 
while the first excited level becomes an interior non-smooth critical value curve 
generically exhibiting swallow tails. 
A super-steep gap, or ``anti-crossing," is characterized by a swallow tail about to coalesce on the boundary, 
whereas the two other cases do not critically involve swallow tails.  
Moreover, the positioning of the swallow tail relative to the boundary provides a ``magnifying lens" on the anti-crossing otherwise difficult to visualize on the energy plots. 
Global properties of the ground versus first excited levels are revealed 
by the Legendrian approach where the energy level curves become Legendrian knots in the contact space. 
As an application important in many aspects, it is shown that 
tunneling leaves its differential topological signature in 
the swallow tail associated with a steep gap. 
More importantly, the stability of the singularities under perturbations calls into question the topologically unstable Grover search, with the consequence of invalidating the gap scaling estimates computed around the unstable singularity 
when uncertain parameters are taken into consideration. 
\end{abstract}

\section*{Introduction}


Adiabatic Quantum Computations (AQC) endeavor to find the ground state of a problem Hamiltonian $H_1$ 
proceeding from the known, easily prepared  
ground state of an initial Hamiltonian $H_0$. 
The passage from the ground state of $H_0$ to the ground state of $H_1$ is accomplished through a variant of the so-called continuation methods~\cite{Alexander1978,Jonckheere1997,Wacker1978}, that is, 
the algorithm---in its ideal implementation---would track the ground state of $H_0p_0(s)+H_1p_1(s)$  
from $s=0$ until $s=1$, 
subject to the initial/terminal conditions $p_0(0)=p_1(1)=1$ and $p_0(1)=p_1(0)=0$.  
If $(p_0(s),p_1(s))$ is differentiable, the map $s \mapsto H_0p_0(s)+H_1p_1(s)$ is a curve from $H_0$ to $H_1$ 
in the subspace $\mbox{span}_{\mathbb{R}}\{H_0,H_1\}$ spanned by $H_0$ and $H_1$ in the space of Hermitian matrices. 
The specific feature of the {\it quantum adiabatic} continuation is that it tracks the solution  
to the Schr\"odinger equation with 
time-varying Hamiltonian $H_0p_0(s(t))+H_1p_1(s(t))$     
from the ground state of $H_0$ to the ground state of $H_1$ via a \textit{\textbf{scheduling}} $s(t)$, 
$s(0)=0$, $s(t_{\mathrm{final}})=1$, 
slow enough $\left(\left|\tfrac{ds(t)}{dt}\right|\right.$ small enough) so that the solution remains close enough to the ground eigenstate all along the path. 
Except for some exceptional cases as the one of Sec.~\ref{s:constant_gap}, 
the \textit{\textbf{gap}} $\lambda_2(H_0p_0(s)+H_1p_1(s))- \lambda_1(H_0p_0(s)+H_1p_1(s))$  
between the first excited eigenstate $\lambda_2$ and the ground eigenstate $\lambda_1$ 
reaches a  minimum, assumed to be nonvanishing, for some $s\in [0,1]$. It is especially around that minimum 
that the integration of the Schr\"odinger equation has to be slowed down,  
in a manner quantified by the adiabatic theorem~\cite{reichardt-adiabatic}. 
Clearly, the scheduling $s(t)$  
depends on the shape of the 
$\lambda_1(s)$ and $\lambda_2(s)$ curves---{\it how fast} and {\it how close} they come together along the path. 
This motivates the study of the \emph{morphology} of the $\lambda_1$, $\lambda_2$ curves, 
how they are interrelated to produce the gap, 
and more importantly how the inflections points that characterize their shapes 
should be \emph{classified}  in the differential category 
to produce a topological discrimination of supersteep, steep and mild gaps.  

The restriction that the strict inequality $\lambda_2(s)-\lambda_1(s)>0$ remains in force $\forall s$ 
can be justified on the ground that 
this property is {\it generic}, that is, it is preserved under sufficiently small parameter variation,  
whereas an \emph{exact} crossing $\lambda_2(s_\times)-\lambda_1(s_\times)=0$ for some $s_\times \in [0,1]$  
indicates 
a \emph{nongeneric} phenomenon~\cite{von_neumann_wigner} that disappears under arbitrarily small perturbation.

An outline of the paper follows.

\begin{figure}[t]
	\centering
	\mbox{
\subfigure[Super-steep gap where the two pairs of inflection point can be related to the swallow tail creating the minimum of $\lambda_2$ (Th.\ref{t:lambda1}(2))]
{\scalebox{0.6}{\includegraphics{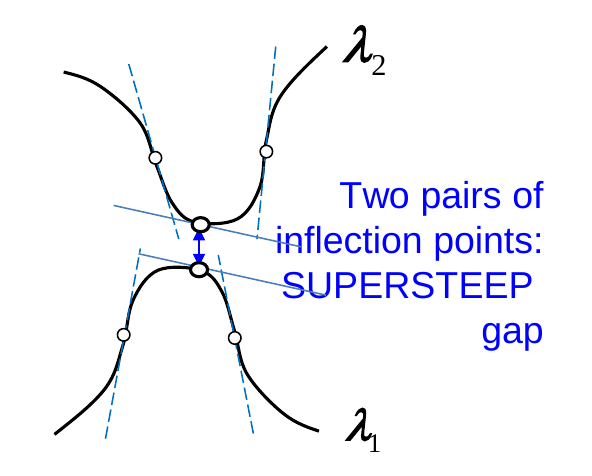}}}\;\;
		\subfigure[Steep gap where the inflection points of the ground level $\lambda_1$ are unrelated to the swallow tail creating the inflections points of the first excited level $\lambda_1$ (Th.\ref{t:lambda1}(1))]
{\scalebox{0.6}{\includegraphics{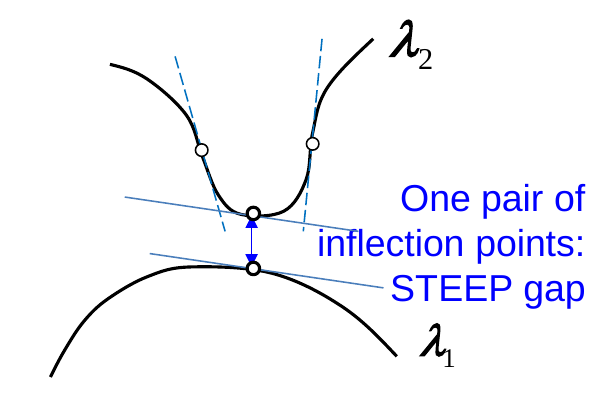}}}
}
	\caption{Steep versus supersteep gap. 
The maximum of the ground level and the minimum of the first excited level are in general offset. 
Hence the gap occurs where the contact points between the parallel tangents (solid lines) 
and the respective curves are vertically aligned. 
}
	\label{f:gaps}
\end{figure}

\subsection{Conventional approach: energy level plots (Sec.~\ref{s:energy_levels})}

The genesis of this paper is the observation  
that the $\lambda_{1}(s)$,  $\lambda_{2}(s)$ 
eigen-energy level curves could take two salient yet stable shapes:  
(i) the ``supersteep" gap where $\lambda_2$ has a steep descent and $\lambda_1$ a steep ascent, 
with the two curves curves nearly coming together in the ``anti-crossing" phenomenon~\cite{JonckheereAhmadGutkin}
and (ii) the ``steep gap" where $\lambda_2(s)-\lambda_1(s)$ goes to a minimum steeply because of the steep descent of $\lambda_2$ while $\lambda_1$ increases moderately.  
This is illustrated in Fig.~\ref{f:gaps} 
and justified on a benchmark problem in Sec.~\ref{s:constant_to_steep}. 
Such configurations are {\it stable} 
in the sense that they are topologically unaffected by a reasonably small change of parameters in $H_0$ and $H_1$. 
The case where the $\lambda_1, \lambda_2$ curves actually intersect is an {\it unstable} phenomenon 
investigated in Sec.~\ref{s:genericity}.    
From the differential viewpoint, 
the anti-crossing or supersteep case means that both curves have pairs of neighboring inflection points 
 occurring nearly simultaneously along the adiabatic path. 
Less striking is the ``steep" case where $\lambda_2(s)$ has a pair of neighboring inflection points, 
but not $\lambda_1(s)$. 
Finally, there is also a ``mild" gap where neither $\lambda_1(s)$ nor $\lambda_2(s)$ have inflections points in a neighborhood of $\arg \min_{s \in [0,1]}(\lambda_2(s)-\lambda_1(s))$, as shown in Fig.~\ref{f:vonNeumanngap}.

\begin{figure}[t]
\centering
\scalebox{0.7}{\includegraphics{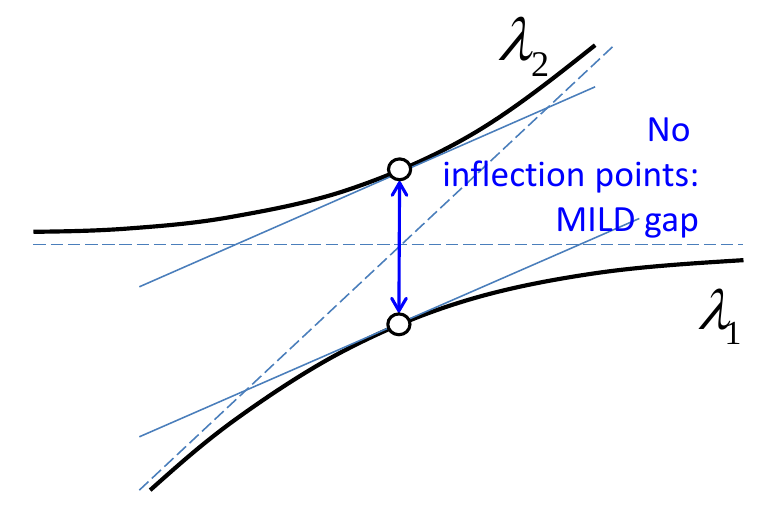}}
\caption{The gap as developed by von Neumann and Wigner~\cite{von_neumann_wigner} as an example of a mild gap. 
The gap occurs where the tangents to the energy levels are parallel. }
\label{f:vonNeumanngap}
\end{figure}

The problem is that, no matter how visually obvious the supersteep versus steep versus mild distinction is, 
it relies on such quantitative features as how close the inflection points are 
and how simultaneously they occur on $\lambda_1(s)$ and $\lambda_2(s)$.  
Clearly, we need a \emph{qualitative} feature---a topological invariant---that would classify the adiabatic problems in a manner consistent with how steep the gap could be.

Since the classification is driven by inflection points, 
the analysis needs to be set within the differentiable category (Sec.~\ref{s:diff_category}). 
Next, in order to get a broad view as to where the inflection points could be and how close they could be, 
the path from $H_0$ to $H_1$ is closed by a return path from $H_1$ to $H_0$. 
Since the classical affine interpolation $p_0(s)=1-s$, $p_1(s)=s$, $0\leq s \leq 1$, cannot be differentiably extended to 
a cyclic path, we fix a circular path visiting $H_0$ and $H_1$ before returning to $H_0$ in a periodic pattern. 
Extending the adiabatic path to the unit circle $\mathbb{S}^1$ allows us to develop invariants of the homotopy type,   
more specifically, the Arnold Legendrian invariants~\cite{Arnold1994} 
and a ``modified" Thurston-Bennequin invariant~\cite{Thurston_Bennequin_Maslov}, 
as shown in Table~\ref{t:classification}.   
Next to the topological invariants, a new invariant, 
the number of roots of $\lambda_2(s)-\lambda_2''(s)$, is developed in Sec.~\ref{s:new_invariant}.  

Besides homotopy considerations, closing the path brings other adiabatic problems in the picture: 
for example, the minimum energy level from $\theta=\pi$ to $\theta=3\pi/2$, that is, from $-H_0$ to $-H_1$, 
ends at the opposite of the \emph{maximum} energy level of $H_1$.  
Likewise, the path from $\theta=\pi$ to $\theta=\pi/2$ terminates at the minimum level of $H_1$, 
starting from the transverse field $-H_0$.  

\begin{remark}\label{r:EPD}
Recently, there has be much interest in the so-called {\it Exceptional Point Degeneracy (EPD)}~\cite{exceptional_points_close}, 
mainly due to the enhanced sensing capability collateral to the emergence of  
such topological structure~\cite{exceptional_enhanced,exceptional_tomography,exceptional_sensing}. 
In a certain sense, an exceptional point is an {\it exact} crossing~\cite[Fig. 2(a)]{exceptional_points_close} 
of energy levels for some parameter values~\cite{exceptional_Mexican05,exceptional_Mexican11}, 
with coalescence of eigenvalues, as here, {\it and} coalescence of eigenvectors, contrary to what happens here.  
Coalescence of eigenvectors happens because EPD appears in nonHermitian quantum systems, 
where Parity-Time (\PT)-symmetry is substituted for the Hermitian property~\cite{exceptional_Czech}.  
A $\mathcal{PT}$-symmetric Hamiltonian $H$, defined by $\mathcal{PT}H-H\mathcal{PT}=0$, retains real eigenvalues, 
yet allowing the modeling of nonequilibrium processes and the encoding of dissipative dynamics~\cite{exceptional_nonreciprocal}.  
\PT-symmetry is said to be {\it broken} when the eigenvalues become complex 
as analytic continuations of the degenerate real eigenvalues of some Hamiltonian~\cite{exceptional_Czech}.  
Despite their elusive interpretations, 
complex energy levels are betrayed by state flipping after an adiabatic path around such exceptional points  (EPs)~\cite{exceptional_points_close,exceptional_nonreciprocal}.  
The {\it steep gap}, the central object of investigation here, can be interpreted in the quantum sensing context as 
extreme sensitivity to the time parameter, in the same manner as EPs are extremely sensitive to Hamiltonian parameters. 
Despite this possible cross-interpretation, a fundamental road block between the two problems is the {\it stability} of the complex eigenvalues of EPD versus the {\it instability} of the exact crossing in the present context of real eigenvalues. 
Certainly, nonHermitian quantum systems could be analyzed in the same differential-topological spirit as the one  developed here, 
but the two contexts are different enough to warrant a further publication. 
\end{remark}

\begin{remark}
Maintaining adiabaticity sometimes requires slowing down the process so much as to open the floodgates to decoherence. The concepts of \emph{diabatic continuation}~\cite{diabatic}  
and \emph{Shortcut To Adiabaticity (STA)}~\cite{shortcut_to_adiabaticity} proceed from the premise that it is futile to enforce adiabaticity in situations where jumping to higher excitation states entails only a small terminal error. 
Moreover, the survey developed in the same Ref.~\cite{shortcut_to_adiabaticity} brings adiabaticity (or the lack thereof) to the much broader abstraction where adiabatic invariance is ubiquitous, 
e.g., state transfer, counter-diabatic driving, quantum dynamical cycles, etc. 
\end{remark}

\begin{remark}
Among the exotic AQC applications with explicit gap considerations, 
one will mention the wireless networking scheduling  problem~\cite{adiabatic_nature,quantum_wireless_II} 
and the prime factorization problem~\cite{adiabatic_integer_factorization}. 
The former introduces a \emph{gap enlargement} 
that applies to the terminal gap, in a rather diabatic approach. 
The latter developed a STA digitized version of AQC where the number of gates increases as the AQC gap decreases. 
\end{remark}

\begin{remark}
Along a totally different line of applications, the same anti-crossing phenomenon manifests itself in 
the Bode singular value plots of transfer matrices~\cite{Zhou}. 
In particular, \cite[Fig. 4.3]{Zhou} shows a supersteep anti-crossing while~\cite[Fig. 4.4]{Zhou} shows a mild 
gap. 
\end{remark}

\subsection{Novel approach: critical value plots (Sec.~\ref{s:diff_category})}

Having set the problem within the differentiable category, we show that 
the eigenenergy levels $\lambda_k(s)$, $k=1,2,...,N$, can be viewed as   
plots of critical values~\cite{Cerf1970} of the quadratic energy function 
$\<z|H_0p_0(s)+H_1p_1(s)|z\>$ defined over the unit sphere $\mathbb{S}^{2N-1}$. 
Recall that a \textit{\textbf{critical value}} is the value of a function at one of its critical point, 
that is, a point where its differential is rank deficient. 
 
A visualization of the critical values is given by the \textbf{\textit{field of values}} or 
\textbf{\textit{numerical range}} of 
$H:=H_0+\imath H_1$ defined as $\mathcal{F}(H):=\{\<z|H|z\>:|z\> \in \mathbb{S}^{2N-1}\}$,
where $N$ is the size of the Hamiltonian.   
usually of the form $2^n$ where $n$ is the number of spins. 
By the Toeplitz-Hausdorff theorem~\cite{Hausdorff,Toeplitz}, $\mathcal{F}(H)$ is compact and convex. 
It turns out that the (convex) boundary curve somehow represents the ground state along 
the circular path visiting $H_0$ and $H_1$. 
Foundational in the differential topological viewpoint is that   
this boundary is a critical value curve of the mapping $|z\> \to \<z|H_0+\imath H_1|z\>$ 
defined over the $(2N-1)$-sphere. 
In addition, there are other critical value curves in the interior; in particular,  
the critical value curve closest to the boundary  
represents the first excited state along the 
circular path (see Fig.~\ref{f:geometry}). 
In this setting, the energy gap $\lambda_2(s)-\lambda_1(s)$ is the ``distance" between the two critical value curves, 
in a sense illustrated by Fig.~\ref{f:unfolding}(b). 

\begin{remark}
Following in the footsteps of~\cite{GutkinJonckheereKarow,adiabatic,JonckheereAhmadGutkin}, 
Spitkovsky and Weiss~\cite{newest_from_Weiss} developed a quantum phase transition interpretation of the boundary of $\mathcal{F}(H)$.  
\end{remark}

\begin{remark}
The plots of critical values of smooth functions defined over a compact differentiable 
 manifold (e.g., the unit sphere) was demonstrated to be a powerful differential topological tool 
in both the stratification of the space of differentiable functions~\cite{Cerf1970} and robust control~\cite{Jonckheere1997}. 
\end{remark}

\subsection{Swallow tails and new invariant (Sec.~\ref{s:swallow_tail})}

Besides providing a new graphical representation of the energy levels, 
the {\it new feature} revealed by the numerical range is the presence of {\it cusps} forming {\it swallow tails} 
on the first excited critical value curve 
(see Fig.~\ref{f:geometry}). 
A \textit{\textbf{cusp}} is a stable singular point where two branches of a curve converge to a common tangency point, 
the two branches being on either side of the tangent (Sec.~\ref{s:swallow_tails}).   
A \textbf{\textit{swallow tail}} is a generic singularity phenomenon consisting of two cusps 
that are connected by an arching edge abutting the tangents at the cusps 
and a crossing of the two branches abutting the tangents on the other side of the connecting arc (see Fig.~\ref{f:unfolding}(b)). 

Most importantly, in the class of problems classically characterized by a (super)steep energy anti-crossing, 
the gap in this new formulation occurs between the arching edge connecting the two cusps of the swallow tail of the first excited state and the boundary curve of the ground state 
(see Fig.~\ref{f:unfolding_details} for a theoretical illustration and Fig.~\ref{f:high_barrier_small_y}for a real adiabatic problem). 

Regarding the ``mild" gap, it is topologically different in that the local minimum between the boundary 
and the first excited critical value curve does not occur at a swallow tail,   
as shown by Fig.~\ref{f:barrier_high_Hamming}, top left-hand panel. 

In the $s$-parameterization of the critical curve of the first excited level, 
the cusps are located at the $s$-solutions of $\lambda_2(s)+\lambda_2''(s)=0$,   
which may or may not exist. 
Existence of a pair of solutions is precisely the topological invariant necessary for a steep gap 
and is the gateway to the supersteep gap.  
The boundary would require inspection of $\lambda_1(s)+\lambda_1''(s)$, which generically does not go to zero, 
but could take small values creating points of extreme curvature in the boundary. 
It is precisely the closeness of those extreme curvature points on the boundary and the cusps that 
determines the morphology of gap.  
For this reason, it is convenient to plot the  $\lambda_k(s)+\lambda_k''(s)$, $k=1,2$, curves as done in  
Figs.~\ref{f:first_case_no_barrier_no_y}-~\ref{f:barrier_high_Hamming}.

Regarding inflection points, it is futile to look for \emph{existence} of inflection points, 
because a theorem by Tabachnikov~\cite{Arnold1994},~\cite{Tabachnikov_original} implies that 
the equation $\lambda_k''(s)=0$ with periodic left-hand side always has at least 2 solutions.  
This implies that characterizing a (super)steep gap by existence of inflection points is inadequate, 
unless the distance between inflection points, their number,  
and their locations along the adiabatic path are taken into consideration.

\begin{remark}
A swallow tail typically occurs when tomographically projecting a 3-dimensional object 
(e.g., a glass torus) to a 2-dimensional space;  
see~\cite{Arnold1994,Arnold1993,ArnoldGusein-ZadeVarchenko1985,arnold_birthday} for precise exposition. 
\end{remark}

\subsection{Legendrian classification (Sec.~\ref{s:Legendrian})}

The critical value curves, all closing on themselves with crossings and other singularities, 
are classified up to Reidemeister moves in accordance with (i) Arnold's winding number and  
Maslov index and (ii) a ``modification" of the Thurston-Bennequin number (see Table~\ref{t:classification}). 
This requires an \textbf{\textit{orientation,}} that is, 
a traveling direction along the critical value curve and a \textbf{\textit{co-orientation}} vector orthogonal to the curve  
(see Fig.~\ref{f:Legendrian}).  
The co-orientation vector is continuous across the cusps and defines a local argument 
$\theta$ of the curve in the complex plane.  
With the contact element, the critical value curves in the complex plane 
are lifted to the 3-dimensional contact space with coordinates $(\Re, \Im, \theta)$,    
to become Legendrian curves allowing for a global analysis to reveal self-knotting of a single curve and linking of two curves. 

This part entails a theoretical new concept: 
the modification of the critical value curve $\gamma_2$ to comply with the Thurston-Bennequin number 
\emph{that does not allow vertical tangents.} 
This modification where vertical tangents are replaced by cusps 
is illustrated in Fig.~\ref{f:breaking_vertical_tangents}.

\subsection{Hamming weight plus barrier: emergence of swallow tail (Sec.~\ref{s:Hamming_plus_weight})}

We introduce a specialized Quadratic Binary Optimization (QBO) problem  
that minimizes a ``Hamming weight plus barrier" function.  
The height and the position of the ``barrier" are manipulated to illustrate 
the various gap topologies and how they relate to swallow tails (Sec.~\ref{s:constant_to_steep}). 

Probably the ultimate quantum adiabatic interpretation of a swallow tail, especially when it  coalesces on the boundary, is that, 
for the ``Hamming weight plus barrier" Hamiltonian~\cite{reichardt-adiabatic} 
and conjecturally other problems~\cite{Understanding_Quantum_Tunneling_through_Quantum_Mo},  
it is a differential-topological signature of tunneling.  
The parameters of the barrier can indeed be adjusted so that the 
adiabatic process tracking the ground state of $H_0p_o(s)+H_1p_1(s)$ is forced to tunnel through the barrier to reach the ground state, 
leaving a swallow tail signature on the critical value curves (Fig.~\ref{f:Hamming_plus_barrier_geometry}).

\begin{figure}[t]
\begin{center}
\scalebox{0.5}{\includegraphics{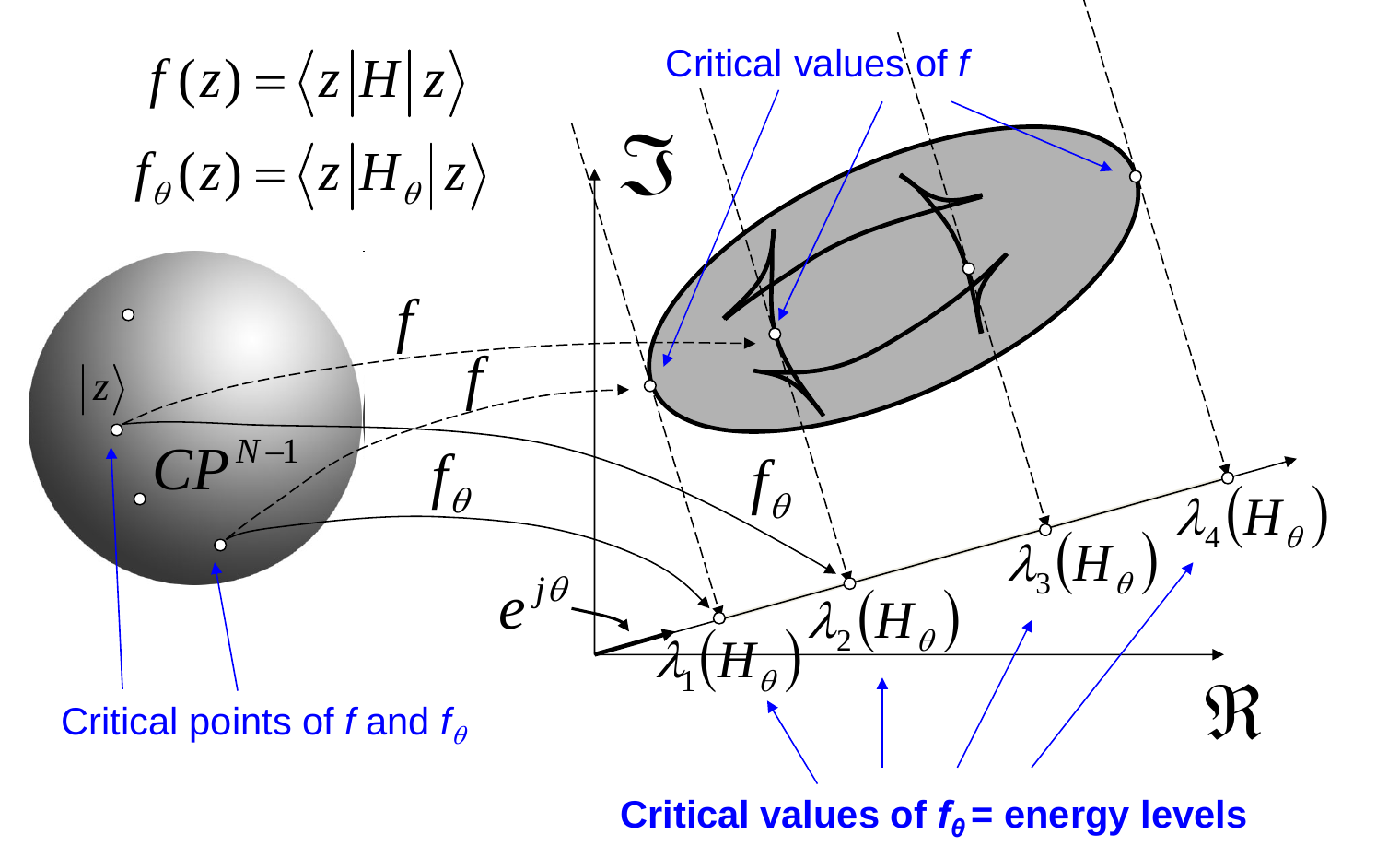}}
\end{center}
\caption{Geometry of the critical values. 
The representation of $\mathbb{C}\mathbb{P}^{N-1}$ as a sphere is only valid for $N=2$ as $\mathbb{C}\mathbb{P}^1\cong \mathbb{S}^2$. }
\label{f:geometry}
\end{figure}

\section{Energy level adiabatic gap}
\label{s:energy_levels}

Traditionally, the adiabatic path is defined as $p_0(s)=1-s, p_1(s)=s$, $s\in [0,1]$. To investigate the global homotopy properties, it can easily be extended to a rhombus passing through the fixed points, or vertices,  
$\mathcal{V}=\{(1,0),(0,1),(-1,0),(0,-1)\}\subset \mathbb{R}^2$. 
Call this path $p_\mathrm{rh}$. 
However, the lack of differentiability at the fixed points make this path inadequate for differential topological investigations. Here, as in~\cite{adiabatic}, 
we will take the following ``master" path, an embedding of $[0,1]$ into $\mathbb{R}^2\setminus (0,0)$,
\begin{equation}
\label{e:path}
H_{\frac{\pi s}{2}}:=H_0\cos\frac{\pi s}{2} +H_1\sin\frac{\pi s}{2}, \quad s \in [0,1],   
\end{equation}
easily extendable from $s\in[0,1]$ to $s\in [0,4]$ as a circle passing through the fixed points. 
Call this path $p_\mathrm{circ}$. The above is referred to as {\it master} path, 
because by various isotopies it could it be enlarged to yield other paths 
that would yield the same results regarding the morphology of the gap, 
hence removing the dependency of the analysis on a specific path.  

There is however a caveat that is easily seen by observing that~\eqref{e:path} is more than an embedding 
in the traditional sense of~\cite[Sec. 1.3]{Hirsch1976}; 
in addition to the traditional first order property of immersion, 
the Hessian~\cite[Sec. 6.1]{Hirsch1976} has constant inertia $(+,+,...,+)$. Therefore, if we define an isotopy 
$F_\tau: (0,1) \times [0,1]\hookrightarrow \mathbb{R}^2\setminus (0,0)$ 
from $\overset{\circ}{p}_\mathrm{circ}=p_\circ \setminus \{(1,0),(0,1)\}$ to 
$F_\tau \overset{\circ}{p}_\mathrm{circ}$ this new embedding should preserve the inertia as $(+,+,...,+)$. 
Moreover, the recover the fixed point, set $\lim_{s\downarrow 0}F_\tau \overset{\circ}{p}_\mathrm{circ}(s)=(1,0)$, 
with a similar restriction for the other fixed points. Note that with such an isotopy, one cannot quite get 
from $p_{\mathrm{circ}}$ to $p_{\mathrm{rh}}$ because the inertia of the latter is (0,0,...,0); however, 
one can get arbitrarily close to it. To summarize, one can argue with a variety of paths, all isotopic to 
$p_{\mathrm{circ}}\cong \mathbb{S}^1$, all giving the same differential topological morphology of the gap. 

Note that any path $p$ in the plane maps to a path in the space $\mathbb{R}^{N^2}$ of Hermitian matrices
with the same differential topological properties.

\subsection{Genericity of anti-crossing}\label{s:genericity}

Let us generalize the problem by bringing intermediate Hamiltonians, 
$H_{i_1},...,H_{i_{m-1}}$ in the adiabatic path in an attempt to open the gap. Such a path would evolve 
in $\mbox{span}\{H_0,H_{i_1},..., H_{i_{m-1}},H_1\}$ and be of the form 
\[ H_{p(s)}:=p_0(s)H_0+p_{i_1}(s)H_{i_1}+...+p_{i_{m-1}}(s)H_{i_{m-1}}+p_1(s)H_1, \] 
where, as for the $1$-dimensional case, 
it is convenient to restrict ourselves to those paths whose homotopy deformation to 
the $\mathbb{S}^m$ sphere is an isotopy.  Restricting the path to $\mathbb{S}^m$ brings the following often overlooked result: 
\begin{theorem}[von Neumann-Wigner adiabatic theorem~\cite{von_neumann_wigner}]
\label{t:von_neumann-wigner}
For $m\leq 2$, the subset of Hamiltonians $\{H_p: p \in \mathbb{S}^m, \lambda_i(H_p)\ne \lambda_j(H_p) \mbox{ for } i \ne j\}$ with single eigenvalues is open and dense (generic) 
in the set of all Hamiltonians $\{H_p: p \in \mathbb{S}^m\}$, but for $m \geq 3$ the genericity of the ``no multiple eigenvalue"  property fails. $\blacksquare$
\end{theorem}

The corollary is that the anti-crossing property is generic for $m=1$, 
but this genericity survives only up to $m=2$, that is, only one 
intermediate Hamiltonian preserves the genericity of the anti-crossing. 
As of $m=3$, the genericity of the anti-crossing is lost, 
and the $\lambda_1$ and $\lambda_2$ curves could intersect. 

The same threshold appears in the crucial convexity property of the $(m+1)$-block joint numerical range:
$$ \{\left(\<z|H_0|z\>,\<z|H_{i_1}|z\>,..., \<z|H_{i_{m-1}}|z\>,\<z|H_1|z\>\right):\<z|z\>=1,z\in \mathbb{C}^{N}\}. $$ 
While under the realistic assumption that $N \geq 2$ the joint range is convex for $m=2$, 
the situation becomes drastically different for $m \geq 3$ 
(see~\cite{GutkinJonckheereKarow}).

\section{Energy levels as critical values}
\label{s:diff_category}

Defining $\theta:=\frac{\pi s}{2}$, clearly, the various energy levels along the extended master path ($0 \leq s \leq 4$) 
are the eigenvalues of 
$H_0 \cos \theta + H_1 \sin \theta=:H_\theta$. This motivates the definition of the function
\begin{eqnarray*}
f_\theta: \mathbb{S}^{2N-1}/\mathbb{S}^1 \cong \mathbb{C}\mathbb{P}^{N-1} &\to& \mathbb{R},\\
z & \mapsto&  \<z|H_0 \cos \theta + H_1 \sin \theta|z\> .
\end{eqnarray*}
Note that the domain of definition of the function $f_\theta$ has been redefined as  
the $(2N-1)$-dimensional sphere quotiented out by the unit circle phase factor,  
which has no effect on the quadratic form.  
The various energy levels are recovered by taking $|z\>$ to be the various eigenvectors. 

Recall that the \textit{\textbf{differential}} $d_z f_\theta: T_z \mathbb{CP}^{N-1} \to \mathbb{R}$ 
of $f_\theta$ at $z$ is the unique linear form 
such that 
\[ \frac{d_z f_\theta(\delta)-(\<z+\delta|H_\theta|z+\delta\>-\<z|H_\theta|z\>)}{\|\delta\|} \to 0, \quad 
\mbox{as } \delta \in T_z \mathbb{C}\mathbb{P}^{N-1} \to 0,\] 
with explicit expression $\<z|H_\theta|\delta\>+\<\delta|H_\theta|z\>$. 
A \textit{\textbf{critical point}} $z^0$ is defined as 
$d_{z^0}f_\theta(\delta)=0$, $\forall \delta \in T_z \mathbb{CP}^{N-1}$, and $f_\theta(z^0)$ is the 
\textit{\textbf{critical value}}.  
In the following, we can safely identify $\mathbb{C}\mathbb{P}^{N-1}$ and $\mathbb{S}^{2N-1}$, 
so that $T_z \mathbb{CP}^{N-1}=z^\perp$.

\begin{lemma}
The critical points (values) of $f_\theta$ are eigenvectors (eigenvalues) of $H_\theta$ and vice versa. 
\end{lemma}
\begin{proof}
If $\ket{z}$ is a critical point, that is, 
$d_z f_\theta(\delta)=2\<\delta|H_\theta|z\>=0$, $\forall \delta \perp z$, 
the decomposition $H_\theta\ket{z}=\alpha \ket{z}+\beta \delta$ yields  $\<\delta|H_\theta|z\>=\beta\|\delta\|^2=0$. Hence $\beta =0$ and $\ket{z}$ is an eigenvector with eigenvalue $\alpha=\lambda$. Conversely, if $H_\theta\ket{z}=\lambda \ket{z}$, 
we have $d_zf_\theta(\delta)=2\<\delta|H_\theta|z\>= \lambda\<\delta|z\>=0$ .
\end{proof}

Clearly, 
$H_\theta= \begin{pmatrix} \cos\theta & \sin \theta\end{pmatrix}\begin{pmatrix}H_0 \\ H_1\end{pmatrix}$, 
but to remain in the analytic function framework, we prefer to rewrite it as
\[ f_\theta(z)=\Re\left(e^{-\imath \theta} \<z|H_0+\imath H_1|z\> \right). \]
Geometrically (see Fig.~\ref{f:geometry}), this means that $f_\theta(z)$ is 
the orthogonal projection of $\<z|H_0+\imath H_1|z\>$ on $\mathbb{R}e^{\imath \theta}$, 
the half-line with argument $\theta$ drawn from the origin of $\mathbb{C}$.  
Taking $|z\>$ to be the $k$th eigenvector, $\<z|H_0+\imath H_1|z\>$ projects onto the corresponding eigenvalue, 
that is, the energy level $\lambda_k(\theta)$. 

More formally, defining
\begin{eqnarray*}
f: \mathbb{S}^{2N-1}/\mathbb{S}^1 \cong \mathbb{C}\mathbb{P}^{N-1} &\to& \mathbb{C},\\
z & \mapsto&  \<z|H_0+\imath H_1|z\> ,
\end{eqnarray*}
the function $f_\theta$ decomposes as 
$f_\theta=p_{\mathbb{R}e^{\imath \theta}}\circ f$. 
Defining the \textit{\textbf{numerical range}} or \textit{\textbf{field of values}}~\cite{JonckheereAhmadGutkin}   
of $H:=H_0+\imath H_1$ as $\mathcal{F}(H):=f(\mathbb{C}\mathbb{P}^{N-1})$,  
the decomposition of $f_\theta$ can be depicted, more formally than by Fig~\ref{f:geometry}, 
by the following diagram:
\begin{equation}\label{e:triangular_diagram}
\begin{array}{ccl}
\mathbb{C}\mathbb{P}^{N-1} & \stackrel{f}{\longrightarrow} & \mathcal{F}(H) \\
         & \stackrel{f_\theta}{\searrow} & 
                         \downarrow  p_{\mathbb{R}e^{\imath \theta}}\\
&&\\
& &        \mathbb{R}e^{\imath \theta}
\end{array}
\end{equation}
By the Toeplitz-Hausdorff theorem~\cite{Toeplitz_Hausdorff}, 
$\mathcal{F}$ is compact and convex.

The numerical range mapping has critical points, that is, points $z^0$ where $\mathrm{rank} (d_{z^0}f)<2$,
where  $d_{z^0}f : T_{z^0}\mathbb{C}\mathbb{P}^{N-1} \to \mathbb{C}$ denotes the differential of $f$ at $z^0$.
The differentials of $f$ and $f_\theta$ are related as 
$d_{z}f_\theta=\Re \left(e^{-\imath \theta} d_z f  \right)$. 
This immediately leads to 
\begin{corollary}\label{eigenHvsHtheta}
A critical point of $f_\theta$ is a critical point of $f$. 
Conversely, if $z$ is a critical point of $f$  
such that $\mbox{rank}(d_z f)=1$, there exists a unique $\theta$ (up to a multiple of $\pi$) such that 
$z$ is a critical point of  $f_\theta$. 
If  $\mbox{rank}(d_z f)=0$, $z$ is a critical point of $f_\theta$ 
for all $\theta$'s.
\end{corollary}

\begin{proof}
Observe that 
$df_\theta(\delta)=2\begin{pmatrix}\cos \theta & \sin \theta\end{pmatrix}
\begin{pmatrix}\Re\<z|H_0|\delta\> \\ \Re\<z|H_1|\delta\>\end{pmatrix}$
is linear (but not complex analytic) in $\delta$. 
Likewise, $d_zf(\delta)=2\begin{pmatrix}1 & \imath \end{pmatrix}
\begin{pmatrix}\Re\<z|H_0|\delta\> \\ \Re\<z|H_1|\delta\> \end{pmatrix}$ 
is also linear (but not complex analytic) in $\delta$. 
From there, the result should be obvious.
\end{proof}

The locus of the $f(z^0)$'s for the various critical points $z^0$ splits into several critical values curves 
embedded in $\mathcal{F}$.  
One such critical value curve is the boundary curve, $\partial \mathcal{F}$, 
while the other curves are in the interior of $\mathcal{F}$. 
The main distinction between the boundary curve and the others is that the former is generically smooth, 
while the others are singular, with cups combining to form swallow tails. 
These topics are addressed in the forthcoming section 
in much more details than in~\cite{adiabatic,JonckheereAhmadGutkin}.

The critical value curves of Fig.~\ref{f:geometry} were constructed as the envelopes of lines orthogonal to $\mathbb{R}e^{\imath \theta}$. Now we prove the converse:

\begin{proposition}
\label{p:projection}
The line orthogonally projecting 
the critical value $f(z_\theta^0)$ onto $\mathbb{R}e^{\imath \theta}$  
is tangent to the critical value curve 
passing through $f(z_\theta^0)$. 
\end{proposition}

\begin{proof}
Consider the critical point $z_\theta^0 \in \mathbb{CP}^{N-1}$,  
associated with some eigenvalue of $H_\theta$, along with the critical curve 
$\gamma$ passing through it.  To prove the theorem, it suffices to show that 
\[ d_{z^0_\theta} \left( f_\theta | \gamma\right) = 0. \]
But since $z_\theta^0$ is a critical point of $f_\theta$, $d_{z^0_\theta} f_\theta=0$. Hence the proof. 
\end{proof}

Finally, we look at the rank 0 critical values. 

\begin{lemma}\label{l:eigenH}
The $\mathrm{rank} (d_{z^0}f)=0$ critical points of $f$ are the eigenvectors of $H$.
\end{lemma}
\begin{proof}
Let $[z]$ denote the equivalent class $\{ze^{\imath \theta}: \theta \in [0,2\pi)\}$.  
Observe that $T_z\mathbb{C}\mathbb{P}^{N-1}$ is the space $[z^\perp]$ orthogonal to $z$ in the sense that 
 $\<[z]|[z^\perp]\>_{\mathbb{C}\mathbb{P}^{N-1}}=\left[\<z|z^\perp\>_{\mathbb{S}^{2N-1}}\right]=0$. 
Hence, if $z^0$ is an eigenvector and $[\delta] \in [z^{0\perp}]$,  it follows that 
$d_{z^0}f(\delta)=\<[z^{0}]|H|[\delta]\> + \< [\delta] |H|[z^{0}] \>=0$ since each term vanishes.  
Conversely, if $d_{z^0}f(\delta)=0$, 
pulling the arbitrary global factors of $[\delta]$ and $[z^0]$ 
from the terms of the sum $\<[z^{0}]|H|[\delta]\> + \< [\delta] |H|[z^{0}]\>$, 
which becomes $\Re(\<z^0|H|\delta\>\exp(\imath \phi))=0$, 
where $\phi$ is the consolidated phases of $\delta$ and $z^0$, 
it follows that the latter sum equals zero only if each term vanishes. Hence $z^0$ is an eigenvector. 
\end{proof}

\begin{remark}
The proof of Lemma~\ref{l:eigenH} is already in~\cite{JonckheereAhmadGutkin}, 
but is here considerably simplified by arguing in the complex projective space as opposed to the sphere. 
\end{remark}

\section{Critical value curves and their singularities}
\label{s:swallow_tail}

\begin{figure}[t]
	\centering
	\mbox{
		\subfigure[Fairly generic $\gamma_1$, $\gamma_2$ critical value curves. 
Observe the swallow tail at a 45$^\circ$ deg angle where the minimum gap occurs. 
Also observe the tendency of the edge of the $\gamma_2$ swallow tail to align itself 
with an area of minimum curvature of $\gamma_1$.]
{\scalebox{0.6}{\includegraphics{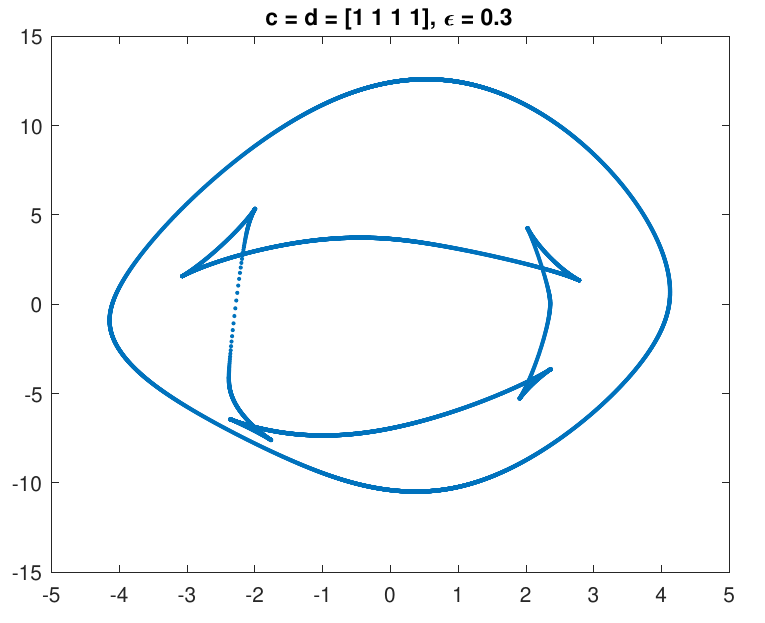}}}
\quad		
\subfigure[Less generic case revealing the unexpected configuration of double swallow tail.  
This pairing is caused by the four sides of the boundary quadrangle each consisting of two branches of small curvature at a shallow angle with the edges of the swallow tails attempting to connect to the low curvature branches.]{\scalebox{0.6}{\includegraphics{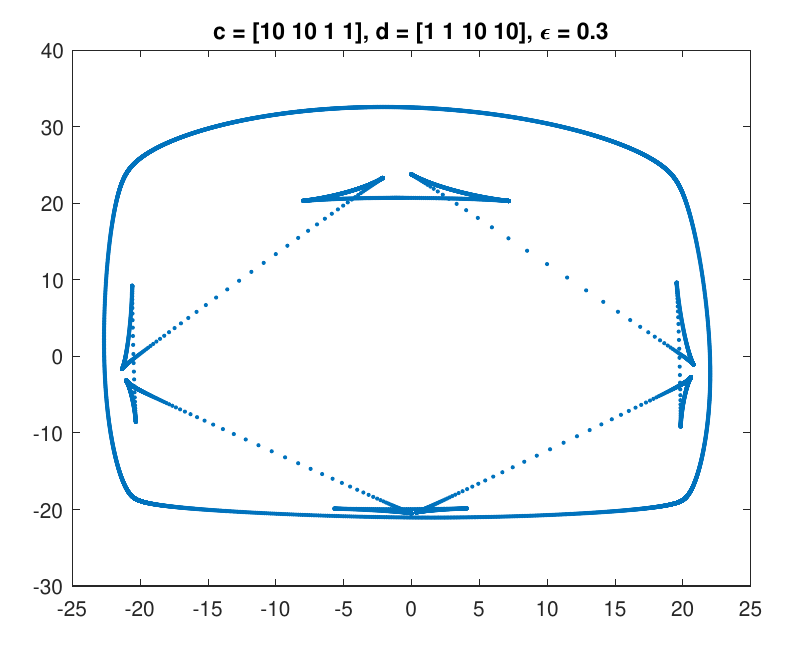}}}
	}
	\caption{Critical value curve of 4-qubit Ising chain $H_1$ with $y$-transverse field $H_0$, 
onsite potentials $c$, $d$ on $H_1$, $H_0$, resp., 
and random perturbations $[-\epsilon,+\epsilon]$ on $H_1$ and $H_0$}
	\label{f:realistic}
\end{figure}

The overarching assumption while crafting the singularity landscape is that the eigenvalues of $H_\theta$ are pairwise distinct and listed as 
\begin{equation}\label{e:overarching}
\lambda_1(\theta) < \lambda_2(\theta) < \lambda_3(\theta) < ... < \lambda_N(\theta), 
\quad \forall \theta \in [0,2\pi].
\end{equation} 
This does not mean that $\lambda_2$ stays at a safe numerical distance from $\lambda_1$; 
quite on the contrary, in general, there will be some angle at which $\lambda_2-\lambda_1$ 
will be numerically indistinguishable from $0$. This is illustrated in Fig~\ref{f:realistic} 
where some swallow tails get very close to the boundary.   
\subsection{Analytic category}

Since $\cos(\theta)=(e^{\imath \theta}+e^{-\imath \theta})/2$ and 
$\sin(\theta)=(e^{\imath \theta}-e^{-\imath \theta})/2\imath$, the operator 
$H_\theta$ can be redefined as $H(e^{\imath \theta})=H_0(e^{\imath \theta}+e^{-\imath \theta})/2 + H_1 (e^{\imath \theta}-e^{-\imath \theta})/2\imath$,  
and 
the eigenvalues $\lambda_k(e^{\imath \theta})$ of $H(e^{\imath \theta})$ can be viewed as functions of $z=e^{\imath \theta}$ defined on the unit circle $\mathbb{T}$ of the complex plane, 
soon to be extended to a (nonsimply connected) neighborhood $\mathcal{N}(\mathbb{T})$ of the unit circle. 

\begin{lemma}
Under Conditions~\eqref{e:overarching} on the eigenvalues, 
$\lambda_k(z)$ is complex analytic (but real valued) in a neighborhood $\mathcal{N}(\mathbb{T})$ of the unit circle. Moreover, 
$\lambda_k(e^{\imath \theta})$ viewed as a function of $\theta \in [0,2\pi]$ is real analytic. 
\end{lemma}
\begin{proof}
Following Kato~\cite[Chap 2, Sec. 1]{Kato1995}, 
the $\lambda$-roots of $\det(\lambda I - H(e^{\imath \theta}))=0$ 
are branches of complex analytic functions of 
$e^{\imath \theta}$ restricted to a simply connected domain, say $\theta \in [0,2\pi)$. 
Since the eigenvalues are pairwise distinct, all such complex analytic functions are 
single-valued.
Moreover, since $H(e^{\imath \theta})$ is periodic, 
so are its eigenvalues, all of which are \emph{real-valued.}  
Hence the domain can be extended to the whole $\mathbb{T}$, 
even to a neighborhood $\mathcal{N}(\mathbb{T})$ sufficiently small 
so that Conditions~\eqref{e:overarching} still hold off $\mathbb{T}$. 
As for the real analyticity, consider the following composition of real analytic functions:
$\theta \to (\cos(\theta),\sin(\theta)) \to (\Re(H_\theta),\Im(H_\theta))\to \lambda_k(\theta)$. 
By the Fa\`a di Bruno formula~\cite[Proposition 1.4.2]{primer_analytic}, this composition of real analytic functions is real analytic.  
\end{proof}

\begin{lemma}\label{l:analytic}
Under the same conditions~\eqref{e:overarching} on the eigenvalues, 
the eigenprojections $P_k(e^{\imath \theta})$ of the eigenvalues 
$\lambda_k(e^{\imath \theta})$ are complex analytic in a neighborhood $\mathcal{N}(\mathbb{T})$ of the unit circle. 
Moreover, the eigenvectors $z_k: \mathcal{N}(\mathbb{T}) \to \mathbb{C}\mathbb{P}^{N-1}$ are complex analytic. 
\end{lemma}
\begin{proof}
The first part is in Kato~\cite[Chap. 2, Sec. 4]{Kato1995}. 
For the second part, observe that $P_kv$ is complex analytic, $\forall v\in \mathbb{C}^N$. 
Since $P_k=v_kv_k^\dagger$, it follows that $P_kv=z_k (z_k^\dagger v)$ is complex analytic. 
It remains to show that $z_k$ is complex analytic. 
Assume by contradiction that one of its component, say $z_{k1}$, is not. 
Pick $v=(1,0,...,0)^\dagger$. It follows that $v_{k,1}v_{k,1}^\dagger$ is not complex analytic. 
But this contradicts the analyticity of $(P_k)_{11}$. 
\end{proof}

While the preceding lemma guarantees \emph{existence} of complex analytic eigenvectors, 
we now develop an explicit \emph{construction} method:

\begin{corollary}\label{c:the_theorem}
Under the same pairwise distinct condition~\eqref{e:overarching} on the eigenvalues of $H_\theta$, 
for any $\ell \in \{1,...,N\}$ 
and $\forall \theta \in [0,2\pi)$, 
all $(N-1)\times(N-1)$ principal minors  of $H_\theta-\lambda_\ell I$ are nonvanishing. 
\end{corollary}
\begin{proof}
Since $\lambda_1(\theta) <...< \lambda_N(\theta)$ are pairwise distinct, so are 
\begin{equation}\label{e:inequalities}
\lambda_1(\theta)-\lambda_\ell(\theta) < ... < \lambda_{\ell-1}(\theta)-\lambda_\ell(\theta) < 0 
< \lambda_{\ell +1}(\theta)-\lambda_\ell(\theta) <...< \lambda_N(\theta)-\lambda_\ell(\theta),\; \forall \theta.
\end{equation} 
Assume, by contradicting hypothesis, that a principal minor vanishes at $\theta^*$. 
Without loss of generality, assume that this minor is in the top left-hand corner position in  $\lambda_\ell(\theta^*) I-H_{\theta^*}$. 
By local unitary transformations (they need not be analytic), 
the last row and columns of the $N \times N$ matrix 
can be brought to vanishing form to reflect the cancellation of the determinant at $\theta^*$. 
Since this top left-hand corner submatrix is Hermitian and singular, 
by another series of local unitary transformations (they need not be analytic), 
the $(N-1)^{th}$ row and column of the 
top left-hand corner block can be manipulated to vanish. 
But then it is evident that the rank of $\lambda(\theta^*)I-H_{\theta^*}$ drops to $N-2$, 
which can only occur when 
among $\{\lambda_k(\theta^*)-\lambda_\ell(\theta^*)\}_{k=1,\ne \ell}^N$ a difference of eigenvalues  
vanishes. But this contradicts~\eqref{e:inequalities}. 
\end{proof}

To put it simply, if we were to plot \emph{all} $(N-1) \times (N-1)$ minors of 
$\lambda_\ell(\theta) I- H_\theta$ versus $\theta$, \emph{none} of them would cross the 0-level. 

Choosing a particular $(N-1)\times (N-1)$ minor specified by the  position it occupies in $\lambda_\ell(\theta) I - H_\theta$ 
and guaranteed to be nonsingular across the unit circle $\mathbb{T}$,  
an eigenvector 
$\ket{z_k(e^{\imath \theta})}$ can be constructed using elementary linear algebra. 
Temporarily ignoring the normalization, construction of $\ket{z_k(e^{\imath \theta})}$ only entails rational operations and $\ket{z_k(e^{\imath \theta})}$ is analytic. 
The normalization requires the nonanalytic operation of conjugation; 
however, it can be made analytic 
by setting $e^{\imath \theta}=z$ and  $e^{-\imath \theta}=z^{-1}$. 
(For a formalization of this \emph{para-conjugation} operation, 
$(e^{\imath \theta})^\ddagger=z^{-1}$), see~\cite{Belevitch}.) 
Hence $z_k^\ddagger z_k$ is analytic and by proper scaling its square root can also be made analytic. 
Therefore, the normalized eigenvector is analytic 
in $z$ in a neighborhood of the unit circle $\mathbb{T}$.

\begin{remark}
Lemma~\ref{l:analytic} has been proved in the simplest possible setting. It is noted that this lemma is just a manifestation of a broader set of results 
along the lines of the analytic Dole\v{z}al theorem~\cite{SilvermanBucyDolezal} 
and the spectral factorization~\cite{analytic_spectral_factorization}. 
Moreover, the latter reveals that the eigenvector $z_k$ is defined only up to an inner factor. 
\end{remark}

\begin{remark}
There is an overlap between between Lemma~\ref{l:analytic} and Corollary~\ref{c:the_theorem}, 
in the sense that the analyticity result of the former could be used to prove existence of a nonsingular 
principal minor. Indeed, if such a minor would not exist, the construction of the eigenvector $z_k(\theta)$ that is analytical over the entire unit circle would require a partition of unity, which does not exist in the analytical category. 
\end{remark}

\subsection{Curvature}
\label{s:theta_parameterization}

As per Proposition~\ref{p:projection}, and as illustrated in Fig.~\ref{f:geometry}, 
the  critical $\lambda_k$-value curve is the envelope  
of the lines with complex argument $\theta \pm \pi/2$ drawn from $\lambda_k(\theta)e^{\imath \theta}$. 
The equations of the  critical $\lambda_k$-value curve $\gamma_k(\theta)=x_k(\theta)+\imath y_k(\theta)$ are given by
\[\begin{array}{cccc}
  & y_k \sin \theta + x_k \cos \theta -\lambda_k(\theta)&   & =0,\\
( & y_k \sin \theta + x_k \cos \theta -\lambda_k(\theta)&)' & =0,
\end{array}\]
where $(\cdot)'$ denotes the derivative relative to $\theta$. 
After some manipulation, this yields the parametric equations
\begin{eqnarray*}
x_k(\theta)&=&\lambda_k(\theta) \cos \theta - \lambda'_k(\theta) \sin \theta, \\
y_k(\theta)&=& \lambda_k(\theta) \sin \theta + \lambda'_k(\theta) \cos \theta .
\end{eqnarray*}
A \textit{\textbf{singular point}} of the critical $k$-value curve is a point such that $x_k'(\theta)=y_k'(\theta)=0$. Some further manipulation on the above yields
\begin{eqnarray}
x_k'(\theta)&=& -(\lambda_k(\theta)+\lambda''_k(\theta))\sin \theta,\label{e:singular_points_1}\\
y_k'(\theta)&=& (\lambda_k(\theta)+\lambda''_k(\theta))\cos \theta. \label{e:singular_points_2}
\end{eqnarray}
Therefore, the singular points of the $k$-critical value curve are given by $\lambda_k(\theta)+\lambda''_k(\theta)=0$.

The arc length of the $k$-curve is given by 
\begin{equation}
\label{e:ds2dtheta2} 
ds_k^2 = (\lambda_k(\theta)+\lambda''_k(\theta))^2 d\theta^2. 
\end{equation} 
The curvature of the $k$-curve is given by
\begin{equation}\label{e:curvature} 
\kappa_k:=\frac{d\theta}{ds_k}=\pm \frac{1}{\lambda_k(\theta)+\lambda''_k(\theta)}. 
\end{equation}

The quantity $\lambda_k(\theta)+\lambda_k''(\theta)$ therefore has a two-fold importance:
First, by Eq.~\eqref{e:singular_points_1}-~\eqref{e:singular_points_2}, $\lambda_k(\theta)+\lambda_k''(\theta)=0$ 
means that $(x_k(\theta),y_k(\theta))$ is a singular point of the critical $k$-value curve and, secondly, by the above, $\lambda_k(\theta)+\lambda_k''(\theta)$
dictates the curvature of the critical $k$-value curve. To determine the sign of the curvature, first observe that $\theta$ is measured trigonometrically, but in addition an \textbf{\textit{orientation}}, that is, a direction to circulate along the curve, 
more precisely, an assignment of a sign to $ds_k$, is necessary. It is common practice to take the orientation trigonometric as well. To solve the $\pm$ ambiguity, the following technical result, which relates the curvature to the eigenstructure of the Hamiltonian, is needed:
\begin{lemma}\label{l:Karow}
\label{l:karow}
\begin{subequations} 
	\begin{align}
\lambda_k(\theta)+\lambda_k''(\theta)
&=2\<z_k(\theta)|H'_\theta(\lambda_k(\theta)I-H_\theta)^+H_\theta'|z_k(\theta)\>,\label{e:fromKarow}\\
&=2\sum_{\ell \ne k} |q_{k,\ell}(\theta)|^2
\frac{1}{\lambda_k(\theta)-\lambda_\ell(\theta)},\label{e:frommine}
\end{align}
\end{subequations} 
where $z_k(\theta)$ is the (normalized) eigenvector associated with $\lambda_k(\theta)$, 
$(\cdot)'$ denotes the derivative relative to $\theta$, $(\cdot)^+$ denotes the Moore-Penrose pseudo-inverse, and 
\begin{equation}\label{e:theQs}
q_{k,\ell}(\theta):=\<z_\ell(\theta) |H'_\theta| z_k(\theta)\>.
\end{equation} 
\end{lemma}
\begin{proof}
Eq.~\eqref{e:fromKarow} is a corollary of a general result~\cite[Th. 3.7(.3)]{GutkinJonckheereKarow} about numerical ranges; 
for the specific details related to adiabatic computation, see~\cite[Sec. 4.2.2]{adiabatic}. 
To keep the exposition self-contained, the relevant material is reviewed in Appendix~\ref{a:differentiability}. Eq.~\eqref{e:frommine} stems from the expansion $H'_\theta \ket{z_k}=\sum_\ell q_{k,\ell}(\theta)z_\ell(\theta)$  in terms of the eigenbasis of $H_\theta$ and the definition of the Moore-Penrose pseudo-inverse as 
\[\left(\lambda_k I - H_\theta\right)^+=U\mathrm{diag}
\left\{\frac{1}{\lambda_k-\lambda_1},\cdots,
\frac{1}{\lambda_{k}-\lambda_{k-1}},0,
\frac{1}{\lambda_{k}-\lambda_{k+1}},\cdots,
\frac{1}{\lambda_k-\lambda_N}\right\}U^\dagger,\] 
where $U$ is the matrix of eigenvectors of $H_\theta$. 
\end{proof}

To resolve the $\pm$ ambiguity in Eq.~\eqref{e:ds2dtheta2}, start with the boundary $\gamma_1$ in the generic case where it is smooth. 
An elementary geometric argument reveals that the curvature is positive, 
but by Lemma~\ref{l:karow} $\lambda_1(\theta)+\lambda_1'' (\theta)<0$.  Therefore, the negative sign prevails for $k=1$. 

Besides the $k=1$ case for which the curvature is obviously nonnegative, generically positive, 
Lemma~\ref{l:karow} as it stands does not shed light on the curvature of the $\gamma_{k>1}$ curves. 
However, averaging over $k>1$ and $\theta$ allows us to identify those segments of the first and higher excitation curves that have negative curvature. 
We need the following lemma:
\begin{lemma}\label{l:correspondence}
The critical value curves $\gamma_k$ and $\gamma_{N-k+1}$  have the same images; precisely, 
$\gamma_k([0,2\pi))=\gamma_{N-k+1}([0,2\pi))$.  Moreover, there exists an involution 
$\gamma_k([0,2\pi)) \to \gamma_k([0,2\pi))$ given by 
$\gamma_k(\theta)\mapsto \gamma_k(\theta+\pi)$. 
\end{lemma}
\begin{proof}
On the $\lambda$-coordinate axes with arguments $\theta$ and $\theta+\pi$, 
it is easily seen that $\lambda_k(\theta)=-\lambda_{N-k+1}(\theta+\pi)$, as \emph{coordinates.} 
However, if $\bm \lambda_k(\theta)$ denotes the \emph{geometrical} point with corrdinate $\lambda_k(\theta)$, 
we have $\bm\lambda_k(\theta)=\bm\lambda_{N-k+1}(\theta+\pi)$. 
If $\gamma[\bm\lambda(\theta)]$ denotes the point of contact of a critical value curve with 
the tangent raised from $\bm\lambda(\theta)$ perpendicular to the \emph{line} of argument $\theta$ or $\theta+\pi$, 
we have $\gamma[\bm\lambda_k(\theta)]=\gamma[\bm\lambda_{N-k+1}(\theta+\pi)]$. 
The left-hand side can be associated with the $k$-curve 
while the righ-hand side is associated with $(N-k+1)$-curve. Hence 
$\gamma_k([0,2\pi))=\gamma_{N-k+1}([0,2\pi))$. From there we define the involution 
$\gamma_k(\theta)\mapsto\gamma_k(\theta+\pi)$. 
\end{proof}

In simple terms, the involution takes a point with normal $\theta$ on a critical value curve and maps it to the point on the same curve with normal $\theta+\pi$. 

To proceed to the averaging procedure, observe that the Fourier expansion of $\Tr(H_\theta)$ has only the fundamental 
without constant term. Therefore, $\Tr(H_\theta)+\frac{d^2}{d\theta^2} \Tr(H_\theta)=0$, and in further, 
\[\sum_{k=1}^N \left({\lambda}_k(\theta)+ {\lambda}_k''(\theta)\right)=0. \]
Assuming that $N$ is even as is the case in adiabatic quantum computations, we group the various terms as follows:
\[  \sum_{k=1}^{N/2} \left(\left({\lambda}_k(\theta)+ {\lambda}_k''(\theta)\right)
+ \left({\lambda}_{N-k+1}(\theta)+ {\lambda}_{N-k+1}''(\theta)\right)\right)=0, \]
the idea being that, via the involution, every (closed) curve has two points of contact with the tangent orthogonal to $e^{\imath \theta}$,  
as see in Lemma~\ref{l:correspondence}. Clearly, $k=1$ corresponds to the boundary. 
Since it is generically smooth and convex, its curvature is positive, viz., 
\[  \left(\left({\lambda}_1(\theta)+ {\lambda}_1''(\theta)\right)
+ \left({\lambda}_{N}(\theta)+ {\lambda}_{N}''(\theta)\right)\right) < 0, \]
But then
\[  \sum_{k=2}^{N/2} \left(\left({\lambda}_k(\theta)+ {\lambda}_k''(\theta)\right)
+ \left({\lambda}_{N-k+1}(\theta)+ {\lambda}_{N-k+1}''(\theta)\right)\right) > 0, \]
which implies that the average curvature is negative---if the $\pm$ ambiguities can be resolved.

Using Lemma~\ref{l:correspondence}, the inequality for the ground curve can be rewritten
\[ \left( {\lambda}_1(\theta)- {\lambda}_{1}(\theta+\pi) \right)
 + \left( {\lambda}_1''(\theta)
- {\lambda}_{1}''(\theta+\pi)\right) < 0. \]
Regarding the higher excitation curves, we get
\[\sum_{k=2}^{N/2} \left(\left( {\lambda}_k(\theta)- {\lambda}_{k}(\theta+\pi)\right)
+\left( {\lambda}_k''(\theta)
- {\lambda}_{k}''(\theta+\pi)\right)\right) > 0. \]
Noting that $\lambda_1(\theta)$ and $(-\lambda_1(\theta+\pi))$ have the same sign, 
with a similar observation for the other terms, we get
\begin{proposition}
\begin{equation}\label{eq:average_int}
\int\limits_0^{2\pi} \left( \lambda_1(\theta)+\lambda_1''(\theta)\right) d\theta =
-\int\limits_0^{2\pi}~\sum_{k=2}^{N/2} \left(\lambda_k(\theta)+\lambda_k''(\theta) \right)d\theta .
\end{equation}
\end{proposition}

To translate the above into a \emph{curvature} statement, an orientation and a \textit{\textbf{co-orientation}}, that is, a direction normal to the curve, must be agreed upon 
\emph{consistently} for both the ground and the higher excitation curves. For the ground level, the orientation is trigonometric and the co-orientation $e^{\imath \theta}$  points outward from the closed curve. The co-orientation $e^{\imath \theta}$ of the higher excited levels points outward as for the ground level. The orientation of the higher excitation levels requires some deeper consideration to make it consistent with 
that of the ground level. It is a phenomenon, justified later, that around tight gaps, the edges of the swallow tails get close to the boundary, making it imperative to make the orientation of the higher excitation levels at the swallow tail edges consistent with the ground level. With these considerations, with $\lambda_1+\lambda_1''<0$ yet the curvature positive, it follows that the left-hand side of Eq.~\eqref{eq:average_int} is $-\kappa_1^{-1}$. For the higher excitation curves, since $\lambda_k+\lambda_k''$ changes sign, yet the sign of the curvature remains constant, there is a need to change the $\pm$ sign  across the cusps. Let the cusps be labelled with an odd number if the orientation departs from the cusp of a swallow tail along its edge and with an even number of the orientation of the edge terminates at the cusp. Hence the integral of the right-hand side of Eq.~\eqref{eq:average_int} should be broken as 
\[\sum_{m~\mathrm{odd}}^{} \left(
 \int_{ \theta_{m \mathrm{~odd}} }^{ \theta_{ m+1 \mathrm{~even}} } (\cdot)d\theta
-\int_{ \theta_{m+1 \mathrm{~even}} }^{ \theta_{m+2 \mathrm{~odd}} }(\cdot)d\theta\right),\]
where the sum is extended over one cycle around every $k$-curve. 
Hence,
\begin{proposition}
\[ -\int\limits_0^{2\pi} \kappa_1^{-1}(\theta)d\theta =
-\sum_{m~\mathrm{odd}}^{} \left(
\pm \int_{ \theta_{m \mathrm{~odd}} }^{ \theta_{ m+1 \mathrm{~even}} } \sum_{k=1}^{N/2}\kappa_k^{-1}(\theta)d\theta
\mp\int_{ \theta_{m+1 \mathrm{~even}} }^{ \theta_{m+2 \mathrm{~odd}} }\sum_{k=1}^{N/2}\kappa_k^{-1}(\theta)d\theta \right),\]
where the $\pm,\mp$ ambiguities must be resolved 
from Eq.~\ref{e:curvature} on a case-to-case basis. 
\end{proposition}

\subsection{Boundary curve}
\label{s:smoothness_boundary}

The smoothness---or the lack thereof---of the boundary of  $\mathcal{F}(H)$ was investigated in~\cite{JonckheereAhmadGutkin} and the relevant results are summarized in Proposition~\ref{p:smoothness_boundary_review} of Appendix~\ref{a:review}. 
Here we address essentially the same problem, but from the novel point of view of 
$\lambda_1(\theta)+\lambda''_1(\theta)$.

\begin{theorem}
\label{t:dramatic_commutator}
If the matrix commutator $[H_0,H_1]$ has no eigenvalues at $0$, $\lambda_1(\theta)+\lambda_1''(\theta)\ne 0$, $\forall \theta$.  
Conversely, 
if $[H_0,H_1]|z_1(\theta)\>=0$ for $|z_1(\theta)\>$ the eigenvector of $H_\theta$ associated with $\lambda_1(\theta)$, then $\lambda_1(\theta)+\lambda_1''(\theta)= 0$ for that particular $\theta$. 
\end{theorem}
\begin{proof}
We prove the first statement by contradiction. 
By the Lemma~\eqref{l:Karow}, 
$\lambda_1(\theta)+\lambda_1''(\theta)= 0$ for some $\theta$  
together with $\lambda_1(\theta)I-H_\theta \leq 0$ implies that 
$(\lambda_1(\theta)I-H_\theta)^\dagger H_\theta'|z_1(\theta)\>=0$.
The preceding in turn implies that 
$H_\theta'|z_1(\theta)\>$ must be the eigenvector of $(\lambda_1(\theta)I-H_\theta)$ associated with the $0$ eigenvalue, that is, the eigenvector $z_1(\theta)$ of $H_\theta$ associated with $\lambda_1(\theta)$. 
Thus $H_\theta'|z_1(\theta)\>=\nu_1(\theta)|z_1(\theta)\>$ for some, 
possibly vanishing,  
$\nu_1(\theta)$. Premultiplying 
$H_\theta'|z_1(\theta)\>=\nu_1(\theta)|z_1(\theta)\>$ by $H_\theta$, premultiplying 
$H_\theta|z_1(\theta)\>=\lambda_1(\theta)|z_1(\theta)\>$ by $H'_\theta$ and subtracting the resulting equalities yields
$[H_0,H_1]|z_1(\theta)\>=0$. \\
Conversely, proceeding from $H_0H_1|z_1(\theta)\>=H_1H_0|z_1(\theta)\>$, multiplying both sides of the equality by $\cos^2\theta+\sin^2\theta$, 
and rearranging terms yields $H_\theta H'_\theta |z_1(\theta)\>=H'_\theta H_\theta|z_1(\theta)\>$. 
Recalling that $H_\theta |z_1(\theta)\>=\lambda_1(\theta) |z_1(\theta)\>$,  
the preceding becomes 
$H_\theta H'_\theta |z_1(\theta)\>=\lambda_1(\theta)H'_\theta |z_1(\theta)\>$ and in turn 
$\left(\lambda_1(\theta)I-H_\theta \right)H'_\theta)|z_1(\theta)\>=0$. 
By the definition of the Moore-Penrose pseudo-inverse, this implies 
$\left(\lambda_1(\theta)-H_\theta \right)^\dagger H'_\theta |z_1(\theta)\>=0$ and in turn $\lambda_1(\theta)+\lambda''_1(\theta)=0$. 
\end{proof}
\\
~

Th.~\ref{t:dramatic_commutator} linked $\lambda_1+\lambda_1''=0$ to the the singularity of the commutator $[H_0,H_1]$. In yet another interpretation, $\lambda_1+\lambda_1''=0$ can be reformulated in terms of a common eigenvector of $H_0$ and $H_1$ or alternatively a common eigenvector of $H$ and $H^\dagger$.

\begin{theorem}
\label{t:dramatic_eigenvectors}
$\lambda_1(\theta)+\lambda''_1(\theta)=0$ for some $\theta$ implies that $H_0$ and $H_1$ have a common eigenvector (or alternatively, $H$ and $H^\dagger$ have a common  eigenvector, $H\ket{z}=\mu \ket{z}$, $H^\dagger \ket{z}= \mu^* \ket{z}$). 
Conversely, if $H_0$ and $H_1$ have $|z_1(\theta)\>$ as common eigenvector, then $\lambda_1(\theta)+\lambda''_1(\theta)=0$. 
\end{theorem}

\begin{proof}
$\lambda_1(\theta)+\lambda''_1(\theta)=0$ for some $\theta$ implies that 
$\<z_1(\theta)|H_\theta'(\lambda_1(\theta) I - H_\theta)^\dagger H_\theta'|z_1(\theta)\>=0$, 
which in turn implies that $H_\theta'|z_1(\theta)\>$ is the eigenvector of $H_\theta$ corresponding to $\lambda_1(\theta)$. 
Thus $|z_1(\theta)\>$ is an eigenvector of $H'_\theta$ and 
$H_\theta'\ket{z_1(\theta)}=\nu_1(\theta) \ket{z_1(\theta)}$ 
for some possibly vanishing $\nu_1(\theta)$.  
$H_\theta$ and $H_\theta'$ sharing the eigenvector $|z_1(\theta)\>$ can be rewritten as 
\begin{eqnarray*}
(H_0 \cos \theta + H_1 \sin \theta ) |z_1(\theta)\> &=& \lambda_1(\theta) |z_1(\theta)\>, \\
(-H_0 \sin \theta + H_1 \cos \theta ) |z_1(\theta) \>&=& \mu_1 |z_1(\theta)\>.
\end{eqnarray*}
Multiplying the first equality by $\cos \theta$, the second by $\sin \theta$, and subtracting the resulting inequalities yields
\[ H_0 |z_1(\theta)\>=(\lambda_1(\theta) \cos \theta -\mu_1 \sin \theta) |z_1(\theta) \>.\]
Likewise, multiplying the first equality by $\sin \theta$, the second by $\cos \theta$, and adding the resulting equalities yields
\[ H_1 |z_1(\theta)\>=(\lambda_1(\theta) \sin \theta + \mu_1 \cos \theta) |z_1(\theta)\>. \]
Clearly, $H_0$ and $H_1$ have a common eigenvector. 
\\
~
Conversely, if for some $\theta$ $|z_1(\theta)\>$ is an eigenvector common to $H_0$ and $H_1$ with eigenvalues 
$a$, $b$, resp., then $|z_1(\theta)\>$ is an eigenvector of $H_\theta$ 
with eigenvalue $a\cos \theta +b \sin \theta=\lambda_1(\theta)$. 
But it is also an eigenvector of $H_\theta'$ with eigenvalue
$-a \sin \theta + b \cos \theta$. 
Hence by Lemma~\ref{l:karow} $\lambda_1(\theta)+\lambda_1''(\theta)=0$. 
\end{proof}

The most striking consequence of the preceding material is the following:

\begin{corollary}\label{c:back_to_Gutkin}
A common eigenvector of $H$ and $H^\dagger$ implies that $[H_0,H_1]$ is singular with a sharp point in the boundary.
\end{corollary} 
\begin{proof}
The first statement is trivial. The second one is derived from~\cite{JonckheereAhmadGutkin}.
\end{proof}

\noindent{\bf Example:} As an extreme case of the preceding, consider $H=\diag\{\mu_\ell\}=\diag\{a_\ell+\imath b_\ell\}$, which is associated with the case $H_0=\diag\{a_\ell\}$ and $H_1=\diag\{b_\ell\}$, 
for which $[H_0,H_1]=0$. 
From the classical point of view of the field of values, 
it is well known~\cite{JonckheereAhmadGutkin} that $\mathcal{F}(H)$ is the convex hull of the $\mu_\ell$'s in the complex plane. 
Let $\{\mu_{\ell_m}\}_{m=1}^M$ be the subset of $\{\mu_\ell\}$ making the boundary of the convex hull. Let the $\mu_{\ell_m}$'s be ordered trigonometrically, in the sense that circulating along the boundary from $\mu_{\ell_m}$ to $\mu_{\ell_{m+1}}$ results in a trigonometric orientation. 
Except for a finite set of angles $\theta_{\ell_m}:=\tan^{-1}\frac{a_{\ell_{m+1}}-a_{\ell_m}}{b_{\ell_{m+1}}-b_{\ell_m}}$, 
the lines orthogonal to $\theta$ 
passing through $\lambda_{1,N}(\theta)e^{\imath \theta}$ 
are making contact with $\partial \mathcal{F}$ at the vertices of the convex hull, 
also referred to as \textit{\textbf{sharp points}}~\cite{JonckheereAhmadGutkin}, defined as points where the tangent to the curve is not continuous. 
At $\theta_{\ell_m}$, the projection line contains the whole $[\mu_{\ell_m},\mu_{\ell_{m+1}}]$ segment.  
Clearly, because there is no one-to-one correspondence between $\theta$ and the boundary, it cannot be parameterized by $\theta$. Nevertheless, 
since $\lambda_1(\theta)=a_{\ell_m} \cos \theta + b_{\ell_m} \sin \theta$ for  
$\theta\in (\theta_{\ell_{m-1}},\theta_{\ell_m})$, it follows that 
$\lambda_1(\theta)+\lambda_1''(\theta) =0$, 
 with the $\theta$-restriction to prevent the projection lines to make other contact with the field of values than the $\theta_{\ell_m}$ sharp point.  
The latter may be interpreted as {\it the boundary having infinite curvature at its sharp points.} 
$\square$

\subsection{Ground versus first excitation curves at awallow tail}
\label{s:swallow_tails}

We analyze the generic senario where the ``gap" occurs between the boundary curve and a nearby swallow tail of the first excitation curve. Quantitatively, we look at the situation where the dominant term in the right-hand side of Eq.~\ref{e:frommine} is  $|q_{k,\ell}|^2/(\lambda_k-\lambda_\ell)$ for $(k,\ell)=(1,2)$ and $(k,\ell)=(2,1)$, that is,
\begin{subequations}
\begin{align}
\lambda_1(\theta)+\lambda_1''(\theta)\approx 
            \frac{|q_{1,2}(\theta)|^2}{\lambda_1(\theta)-\lambda_2(\theta)}\label{e:1nominal},\\
\lambda_2(\theta)+\lambda_2''(\theta)\approx 
           \frac{|q_{2,1}(\theta)|^2}{\lambda_2(\theta)-\lambda_1(\theta)}\label{e:2nominal},
\end{align} 
\end{subequations}
respectively, for a range of $\theta$'s  that will be delineated soon. From Eq.~\ref{e:theQs}, it follows that $q_{1,2}=q^\dagger_{2,1}$, hence $|q_{1,2}|^2=|q_{2,1}|^2$, 
and the curvature of $\gamma_1$ and $\gamma_2$ have the same amplitude, but of opposite signs.  
Since $\gamma_1$ is positively curve, $\gamma_2$ is negatively curved. 

The next step is to identify a nominal $\theta$ around which the gap develops. This is a corollary of the following:
\begin{proposition}
Under the asymptotic conditions~\eqref{e:1nominal}-\eqref{e:2nominal}, 
the $\gamma_1$, $\gamma_2$ curves have a common normal at an angle
\[\theta_\perp= \arg \min_\theta d(\gamma_1(\theta),\gamma_2(\theta)).\]
Moreover, there exists a neighborhood $(\theta_-,\theta_+)\ni \theta_\perp$ such that 
\begin{equation}\label{e:cusps_exist}
\lambda_2(\theta_-)+\lambda_2''(\theta_-)=0, \quad \lambda_2(\theta_+)+\lambda_2''(\theta_+)=0.
\end{equation}
Finally,
\begin{equation}\label{e:argmin} 
\theta_\perp = \arg \min_{\theta \in(\theta_-,\theta_+)} |\lambda_1(\theta)-\lambda_2(\theta)|.
\end{equation}
\end{proposition}
\begin{proof}
Because of the curvature conditions on $\gamma_1$ and $\gamma_2$, 
$\min d(\gamma_1(\theta),\gamma_2(\theta))$ exists and defines the angle $\theta_\perp$.  
Moreover, 
the first order conditions require the vector $\overrightarrow{\gamma_1(\theta_\perp)\gamma_2(\theta_\perp)}$ to abut  
$\gamma_1$ and $\gamma_2$ perpendicularly. 
The second order conditions are consequential to the common normal property 
that implies concavity of $d(\gamma_1(\theta_\perp,\gamma_2(\theta))$  
and $d(\gamma_1(\theta),\gamma_2(\theta_\perp)$ viewed as functions of $\theta$. 
The singularities~\eqref{e:cusps_exist} on $\gamma_2$ are necessary to prevent $\gamma_2 \subset \mathcal{F}(H)$ 
to escape the field of values. 
To prove~\eqref{e:argmin}, 
let us draw the osculating circles $O\gamma_1, O\gamma_2$ of $\gamma_1, \gamma_2$  at  
$\gamma_1(\theta_\perp), \gamma_2(\theta_\perp)$, resp.~\cite[pg. 65]{oneill}.  
In this configuration, $\overrightarrow{\gamma_1(\theta_\perp)\gamma_2(\theta_\perp)}$ is aligned with the line joining the centers $C_1, C_2$ of the osculating circles $O\gamma_1, O\gamma_2$ 
at $\gamma_1(\theta_\perp), \gamma_2(\theta_\perp)$, resp.  
The projection lines of $\overrightarrow{\gamma_1(\theta_\perp)\gamma_2(\theta_\perp)}$ 
on $\lambda_1(\theta_\perp),\lambda_2(\theta_\perp)$ are tangent 
to $O\gamma_1$, $O\gamma_2$ at $\gamma_1(\theta_\perp), \gamma_2(\theta_\perp)$, resp.  
Consider a $d\theta$ rotation of the projection lines tangent to $O\gamma_1, O\gamma_2$ 
at $\gamma_1(\theta_\perp +d\theta), \gamma_2(\theta_\perp +d\theta)$, resp. 
The perturbed gap $\lambda_2(\theta_\perp +d\theta)-\lambda_1(\theta_\perp +d\theta)$ is measure as the distance between the rotated projection lines along their common normal 
passing through $\gamma_1(\theta_\perp +d\theta)$, $O_1$, and intercepting the tangent to $O\gamma_2$ at $x$. 
The crucial point is that $x$ lies outside the disk bounded by $O\gamma_1$. 
If $O_1$ denotes the center of $O\gamma_1$, we have 
\begin{align}
\lambda_2(\theta_\perp +d\theta)-\lambda_1(\theta_\perp +d\theta)
&=\|\overrightarrow{\gamma_1(\theta_\perp+d\theta)O_1}\|+\|\overrightarrow{O_1x}\|
=\|\overrightarrow{\gamma_1(\theta_\perp)O_1}\|+\|\overrightarrow{O_1x}\|\\
&>\|\overrightarrow{\gamma_1(\theta_\perp)\gamma_2(\theta_\perp)}\|
=\lambda_2(\theta_\perp)-\lambda_1(\theta_\perp),
\end{align}
and the result is proved. 
\end{proof}

Eqs.~\ref{e:cusps_exist} reveal existence of singularities at the distal ends of the curved edge 
$\gamma_2([\theta_-,\theta_+])$. It remains to classify those singularities.

\begin{definition}
	\label{d:cusp}
	A \textbf{3/2-cusp point} on a parameterized curve $\theta \mapsto (x(\theta),y(\theta))$ is a singular point $\theta^o$, 
	that is, a point characterized by $x'(\theta^o)=y'(\theta^o)=0$,  
	such that locally around $(x(\theta^o),y(\theta^o))$ there is a diffeomorphic change of coordinates $(x,y) \leftrightarrow (u,v)$  
	such that in the new coordinates the curve takes the canonic form $u^3=v^2$.  
\end{definition}
It is easily seen that a 3/2-cusp has 2 branches tangent to each other at the singular point 
with the two branches \emph{on either side} of the common tangent. 
There is also a 5/2-cusp or \emph{ramphoid} with canonical form $u^5=v^2$, with the difference that the two branches are \emph{on the same side} of the common tangent. 

\begin{theorem}
The generic singularities where $\lambda_k(\theta^0)+\lambda_k''(\theta^0)=0$ are 3/2-cusps. 
\end{theorem}

\begin{proof}
Following Arnold~\cite{Arnold1993}, there are two stable singularities: the cusp or 3/2 singularity and the  ramphoid or 5/2 singularity. We first show that the latter cannot happen. 
Indeed, if the singularity were a ramphoid, 
under the motion $\theta \uparrow \theta^0$ of the point of contact of the tangent to 
the singularity at $\theta^0$, it would follow that the tangent could not be prolonged for $\theta > \theta^0$, which contradicts the construction that the tangent exists for all $\theta$'s. 
Having ruled out the ramphoid, we now show that the only generic singularity is the 3/2-cusp. 
By the construction of the $k$-curve, its normal is $\theta$. 
Assume, contrary to the cusp hypothesis, that 
the normal vector has a discontinuity jumping from $\theta^-$ to $\theta^+$ at the singularity.  
For $\theta \in [\theta^-, \theta^+]$, the projection line normal to the direction $\theta$ erected from $\lambda_k(\theta)\exp(\imath \theta)$ passes through the same point, $(x(\theta^0),y(\theta^0))$, 
a degenerate subset of the envelope. This implies that $\lambda_k(\theta)$ has pure sinusoidal behavior for $\theta \in [\theta^-, \theta^+]$. 
We show that under such circumstances, 
the sinusoidal behavior \emph{persists $\forall \theta \in [0,2\pi)$} 
and the singularity reduces to a point, a nongeneric case. 
To develop an analyticity argument, the $\theta$-parameterization is replaced by the 
$e^{\imath \theta}$ parameterization, extendable to a neighborhood of the unit circle $z\in \mathbb{T}$. 
Restricting the argument of $z$ to $\theta \in [\theta^-,\theta^+]$, 
with $\cos\theta=(z+z^{-1})/2$ and $\sin\theta= (z-z^{-1})/2\imath$, we have the eigendecomposition
\begin{equation}\label{e:eigendecomp}
H_0\frac{z+z^{-1}}{2} + H_1\frac{z-z^{-1}}{2\imath}=U_z
\begin{pmatrix}
\lambda_k(z) & 0\\
0 & \Lambda_{\bar{k}}(z)
\end{pmatrix}U_z^{-1},
\end{equation}
where $\lambda_k(e^{\imath \theta})=a \cos\theta+b\sin \theta$ and $U$ is the matrix of eigenvectors, 
orthogonal but not normalized to preserve analyticity of all functions. 
Classically, this decomposition is constructed  
by identifying a nonsingular principal $(N-1)\times (N-1)$ submatrix. 
But by Corollary~\eqref{c:the_theorem}, the same $(N-1)\times (N-1)$ principal submatrix remains nonsingular  $\forall z=e^{\imath \theta}$.  
Therefore, Eq.~\eqref{e:eigendecomp} can be extended $\forall z \in \mathbb{T}$, 
and the sinusoidal behavior of $\lambda_1(z)$ persists along $\mathbb{S}^1=\mathbb{T}$; 
in other words, the tangent bundle $T\gamma_1 \to \mathbb{S}^1$ is trivial~\cite{Dolezal_vector_bundle},  
and the critical curve is reduced to a point, a nongeneric situation. 
\end{proof}

\begin{corollary}
	The cusps come in pairs, forming swallow tails.
\end{corollary}
\begin{proof}
Since $\lambda_k(\theta)+\lambda''(\theta)$ is periodic and, if it changes sign at a singular point, 
say $\theta^0_1$, it has to revert to the original sign at the next singular point $\theta_{2}^0 > \theta_1^0$. 
Hence the cusps appear in pairs. 
\end{proof}

\begin{remark}
Cusps have been found to be typical singularities of envelopes~\cite{envelope_and_cusp}. 
However, here, we have very peculiar $\lambda_{k\geq 2}$-envelopes that deserve special attention. 
Regarding the $\lambda_1$-envelope, it is the boundary of $\mathcal{F}(H)$, 
here defined as the envelope of a family of lines. In~\cite{Applications_of_envelopes}, $\mathcal{F}(H)$  
is shown to be the union of a family of circles, 
so that $\partial \mathcal{F}(H)$ is the envelope of a family of \emph{circles,} rather than \emph{lines}. 
\end{remark}

\subsection{Exact crossing}

On the eigenenergy plots, exact crossing is stated as $\lambda_2(\theta_\times)-\lambda_1(\theta_\times)=0$; 
on the critical value curves, it means that the $\gamma_2$ first excited curve 
tangentially contacts the boundary curve $\gamma_1$ in a nongeneric phenomenon.  
(This phenomenon can be reinterpreted in terms of Arnold's $J^+$-theory~\cite{Arnold1994}.)  
Fig.~\ref{f:barrier_low_Hamming}, top left-hand panel, illustrates such situation. 
Fig.~\ref{f:barrier_high_Hamming} shows that a change of parameters creates a swallow tail 
(at $\theta=\pi/2$) in the first excited curve and disconnects it from the boundary. 

\subsection{Classification using  
the $\mathrm{roots}(\lambda_2+\lambda_2'')$ invariant}\label{s:new_invariant}

The classification of all adiabatic paths from the conventional point of view of plotting the energy levels is driven by the inflection points in the $\lambda_{1}(\theta)$, $\lambda_2(\theta)$ plots. 
The problem is that a classification based on existence of solutions to $\lambda_{1}''(\theta)=0$ and $\lambda_2''(\theta)=0$ is inadequate because such equations \emph{always} have  roots 
as a consequence of the following theorem:
\begin{theorem}[Tabachnikov~\cite{Arnold1994}~\cite{Tabachnikov_original}]
Let $f(\theta)=\sum_{k=k_{\mathrm{min}}}^\infty (a_k \cos k \theta + b_k \sin k \theta)$ be a $2 \pi$ periodic function with its 
Fourier series starting at $k_{\mathrm{min}}$. Then $f(\theta)=0$ has at least $2k_{\mathrm{min}}$ roots in $[0,2\pi)$. $\blacksquare$
\end{theorem}
Clearly, the Fourier series of $\lambda_1''(\theta)$, $\lambda_2''(\theta)$ have no constant terms  
and therefore each function has at least two roots.   
While $\lambda_{1}(\theta)$, $\lambda_2(\theta)$ always have inflection points, 
they need not appear in pairs of closely spaced roots, nor is it guaranteed that the inflection points of $\lambda_1(\theta)$ and $\lambda_2(\theta)$ are aligned to create a steep gap. 
Clearly, these roots do not allow us to distinguish the various mild to steep cases. 

From the novel point of view of Section~\ref{s:theta_parameterization}, 
the classification is driven by cusp points and swallow tails. 
The latter are identified by pairs of roots of  
$\lambda_2(\theta)+\lambda_2''(\theta)=0$. 
Since in general the Fourier expansion of $\lambda_2$ has a constant term, 
Tabachnikov's theorem is of no help, and therefore 
$\lambda_2(\theta)+\lambda_2''(\theta)=0$ may or may not have roots, making 
$\lambda_2(\theta)+\lambda_2''(\theta)$ a stronger invariant than $\lambda_2''$.

As already proved in Sec.~\ref{s:smoothness_boundary}, $\lambda_1(\theta)+\lambda_1''(\theta)$ 
generically has no roots. 
Therefore, $\lambda_2(\theta)+\lambda_2''(\theta)$ is the only function left that may or may not 
  have roots (as illustrated in Section~\ref{s:constant_to_steep}). 
The existence and the number of such roots is a topological invariant.

We now strengthen the role of $\lambda_2+\lambda_2''$ in locating the inflection points of the  classical $\lambda$-energy level plots.  
We first consider the $\lambda_2$-case.

\begin{theorem}\label{t:lambda2}
Consider a swallow tail embedded in the numerical range 
together with the arc joining its two cups located at consecutive roots of $\lambda_2(\theta)+\lambda_2''(\theta)=0$. 
Assume that, along that arc, $\lambda_2>0$ and  
that the maximum of $\lambda_2$ is visited. 
Then along the arc joining the two consecutive cusps of the swallow tail, 
there are two $\lambda_2''=0$ inflections points. 
\end{theorem}

\begin{proof}
Since $\lambda_2+\lambda_2''=0$ at the cusps and since $\lambda_2>0$ 
along the arc joining the cusps, it follows that at the cusps $\lambda_2''<0$. 
But in a neighborhood of the local $\lambda_2$-minimum, $\lambda_2''>0$. 
Therefore, along the arc joining the cusps of the swallow tail, $\lambda_2''$ must cancel at at least 2 points, that is, 2 inflections points. 
\end{proof}

We now look at the $\lambda_1$ singular representation, 
that is, the boundary of the numerical range. 
Instead of relying on consecutive cusps points of \emph{infinite} curvature, 
here we rely on consecutive boundary ``vertices" of \emph{extremal} curvature.
A point of extremal curvature is referred to as \textit{\textbf{vertex}}.
Vertices are guaranteed to exist by the 4-vertex theorem:  
\begin{theorem}[4-vertex theorem~\cite{4_vertex_revisited}]\label{t:4-vertices}
	Any smooth, simple, closed curve in the plane has at least 4 vertices. $\blacksquare$
\end{theorem}

\begin{theorem}\label{t:lambda1}
Under the same conditions as Theorem~\ref{t:lambda2}, consider the boundary arc 
between two consecutive ``vertices" of maximum curvature characterized by $\lambda_1+\lambda_1''\approx 0$. 
Then one distinguishes the following two cases:
\begin{enumerate}
\item If $\lambda_1>0$, inflection points near those of $\lambda_2$ cannot be guaranteed, 
as shown by Fig.~\ref{f:gaps}(a).
\item If $\lambda_1<0$, then nearby $\lambda_1$ inflection points are guaranteed, 
as shown by Fig.~\ref{f:gaps}(b).
\end{enumerate}
\end{theorem}
\begin{proof}
Regarding the $\lambda_1>0$ case, at the two vertices,  
the curvature is maximum so that $\lambda_1+\lambda_1''\approx 0$;  
therefore, $\lambda_1''<0$ in a neighborhood of the vertices. 
But at the maximum of $\lambda_1$, $\lambda_1''<0$. Therefore, contrary to the $\lambda_2$-case, 
a change of sign in $\lambda_1''$ cannot be guaranteed. Regarding the $\lambda_1<0$ case, the proof follows the same pattern as that of Theorem~\ref{t:lambda2} (after reversing the signs) and is omitted. 
\end{proof}

\begin{remark}
Arnold gives credit to Tabachnikov for this result~\cite{Arnold1994}, 
although its trigonometric polynomial version was know to Sturm, and the general result had already been proved by Hurwitz, 
before being rediscovered by Kellogg and Tabachnikov~\cite{Arnold_on_Tabachnikov}.
\end{remark}

\begin{remark}
A detailed account of the 4-vertex theorem is available in~\cite{fronts_of_Legendrian_links}. 
The same reference also introduces a 4-cusp theorem for the equidistant curve. 
\end{remark}

\subsection{Stability of critical value curves}\label{s:stability}
\label{s:deformation}

\begin{figure}[t]
	\centering
	\mbox{
		\subfigure[Nongeneric $4\times 4$ $H_0+\imath H_1$  field of values consisting of the convex hull of the field of values of a generic $2\times 2$ matrix  and $\{\mu_1,\mu_2\}$ outside the elliptical field. 
The case of the Grover search is similar, but more singular as $\mu_1=\mu_2$.]{\scalebox{0.5}{\includegraphics{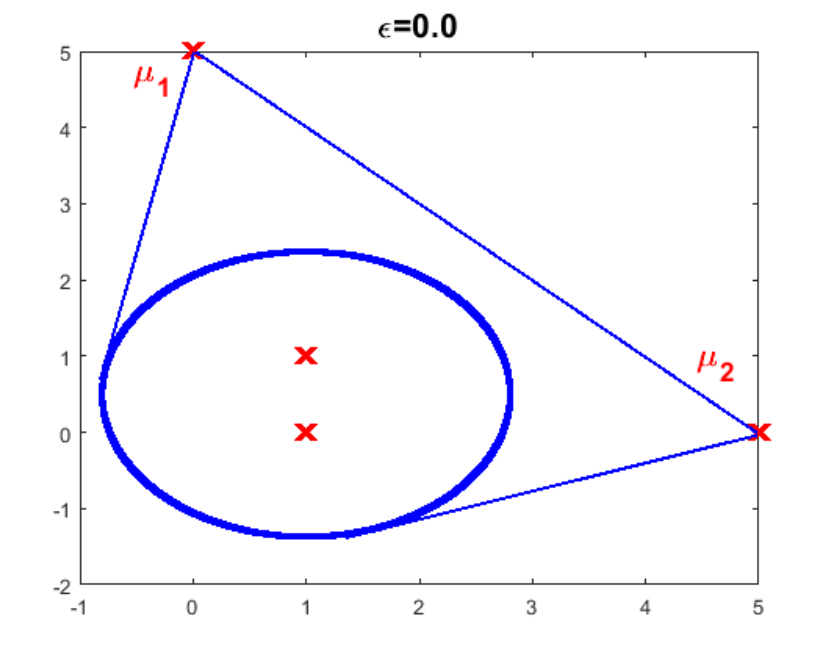}}}
\quad		\subfigure[Stable configuration of the field of values after perturbation. Increasing the resolution reveals that the boundary consists of a smooth arc passing nearly $\mu_1$ and $\mu_2$. Observe that four of the cusps of the three swallow tails are ``attracted" by  $\mu_1$ and $\mu_2$.]{\scalebox{0.5}{\includegraphics{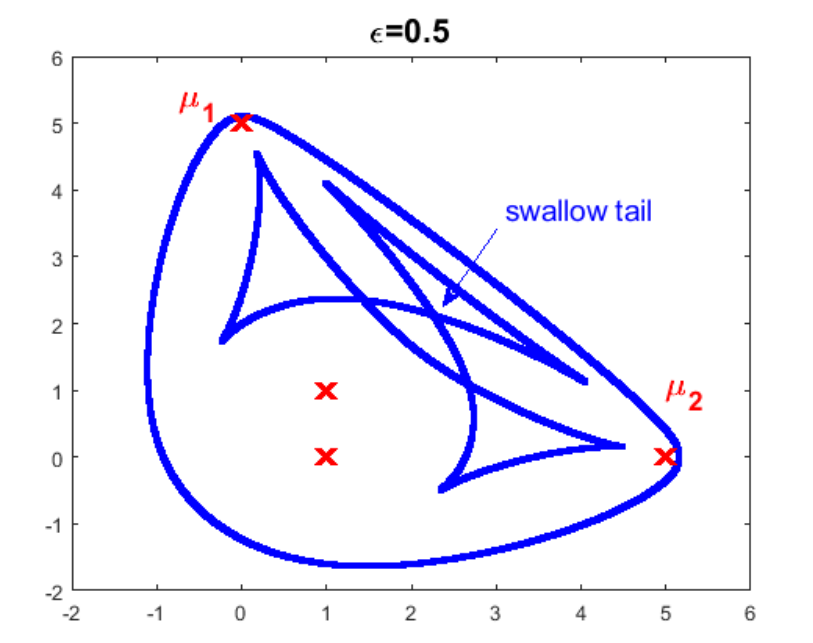}}}
	}
	\caption{Unfolding of an unstable numerical range}
	\label{f:unfolding}
\end{figure}

\begin{figure}[t]
	\centering
	\mbox{
		\subfigure[Emergence of the ideal hyperbolic triangle from $\mathcal{F}(S_{33}$)]{\scalebox{0.5}{\includegraphics{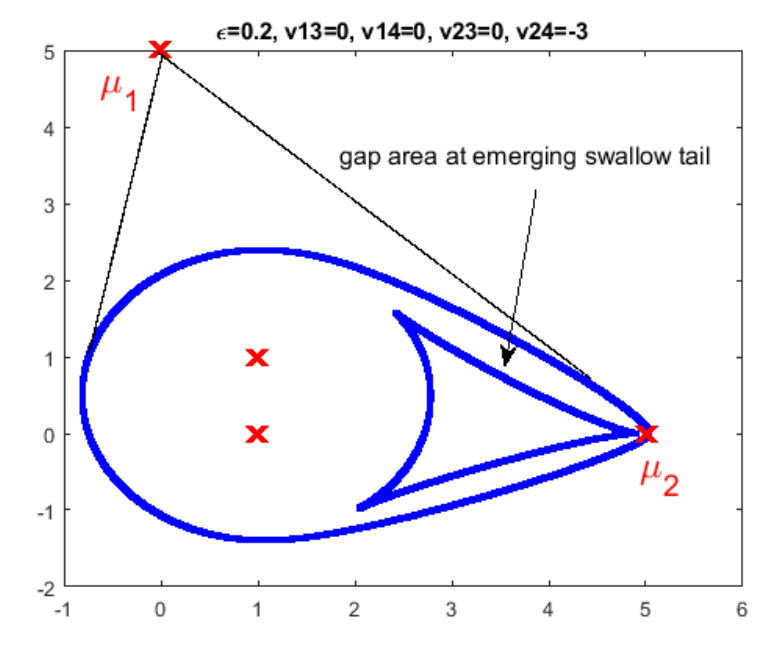}}}
\quad		\subfigure[Emergence of swallow tail from ideal hyperbolic trianngle]{\scalebox{0.5}{\includegraphics{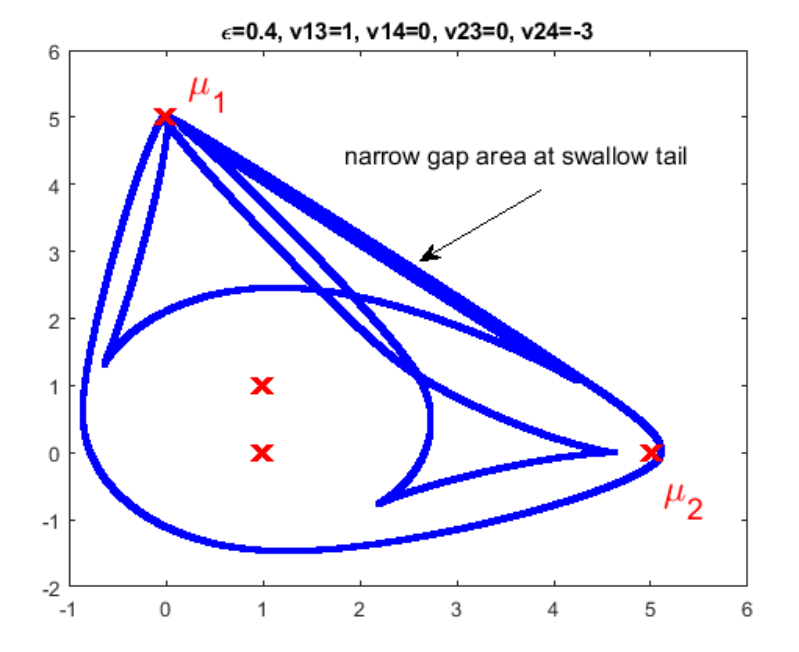}}}
	}
	\caption{Details of the unfolding of the unstable numerical range 
of Fig.~\ref{f:unfolding} between (a) and (b) of the same figure}
	\label{f:unfolding_details}
\end{figure}

\textbf{\textit{Stability}} of the critical value curves refers to their structure-preserving, continuous deformation under variation of some extraneous parameters, for example, the qubit couplings in the D-Wave~\cite{adiabatic_nature,quantum_wireless_II}. 
The corresponding (quadratic) map is then said to be \textbf{\textit{stable}}. 
Conversely, emergence of a swallow tail from an otherwise simple boundary curve is an \emph{unstable} phenomenon, 
illustrated by the passage from Fig.~\ref{f:unfolding}(a) to Fig.~\ref{f:unfolding}(b) 
with the details shown in Fig.~\ref{f:unfolding_details}.  
More technically speaking, this phenomenon is referred to as \emph{unfolding}~\cite{CastrigianoHayes1993} 
in the theory of stability of maps~\cite{GolubitskyGuillemin1973} (here $z \mapsto \<z|H|z\>$) 
and their singularities~\cite{ArnoldGusein-ZadeVarchenko1985,ArnoldGusein-ZadeVarchenko1988}. 

\subsubsection{Grover search catastrophe}

To show the potentially catastrophic consequences of Th.~\ref{t:dramatic_commutator} in quantum adiabatic computations, 
consider such benchmark examples as Grover search, inversion of Toeplitz matrices, or any problem 
that can be 
formulated in the present context as  
$H_0=I-|a\>\<a|$ and $H_1=I-|b\>\<b|$, where $\< a|a\>= \< b|b\>=1$. 
This simple structure makes the gap easy to compute~\cite{adiabatic}, 
leading us {\it to believe} that it scales as $1/\sqrt{N}$,  
but the big problem is that the map $|z\> \to \<z|H|z\>$ is unstable~\cite{GolubitskyGuillemin1973}, 
as evidenced by its nongeneric numerical range---the convex hull 
of a vertex point and an ellipse~\cite[Fig. 2]{adiabatic}. 
This has the ineluctable consequence that the unstable singularity at the vertex 
will break into \emph{stable} cusp and swallow tail singularities~\cite{GolubitskyGuillemin1973}.  
Now the gap, in its critical value curve interpretation, has to be re-computed around a swallow tail singularity, 
which is much more complicated, 
with the resulting gap scaling deviating from the nominal one.  
A similar, more recent observation was made in~\cite{Itay_Hen}. 

Precisely, 
if we take $|c\>$ of unit norm in the orthogonal complement of the span of $|a\>$ and $|b\>$, 
it is easily verified that 
$[H_0,H_1]|c\>=0$ and, moreover, $H_{\theta=5\pi/4}|c\>=-\sqrt{2}|c\>$. 
Thus the eigenvector $|c\>\in \mbox{span}\{|a\>,|b\>\}^\perp$ achieves the eigenvalue $-\sqrt{2}$ at $\theta=5\pi/4$. 
It remains to show that the other eigenvectors of $H_\theta$ at $\theta =5\pi/4$ achieve eigenvalues greater than $-\sqrt{2}$ so that 
$|c\>=|z_1(5\pi/4)\>$ is the ground state.  
The other eigenvectors of $H_\theta$ are orthogonal to $\ket{c}$, 
hence in the span of $\ket{a}, \ket{b}$, and  are of the form $\alpha |a\>+\beta|b\>$. 
Then it is easily verified that 
\[H_\theta (\alpha |a\>+\beta|b\>)=
|a\>(\alpha \sin \theta -\beta \<a|b\> \cos \theta)+
|b\>(\beta  \cos \theta -\alpha \<b|a\> \sin \theta). \]
Setting $\theta=5\pi/4$, the above yields $\alpha=\pm \beta$, and the eigenvalue associated with 
$\alpha |a\>+\beta|b\>$ becomes $\lambda=(1\pm \<a|b\>)(-\sqrt{2}/2)$. Clearly, 
$\lambda \geq -\sqrt{2}$, with equality achieved only in the unrealistic case $a=b$. 
It follows that any eigenvector $\alpha |a\>+\beta|b\>$ achieves   
an eigenvalue $>-\sqrt{2}$. 
Therefore, $-\sqrt{2}$ is the minimum eigenvalue $\lambda_1(5\pi/4)$ with eigenvector 
$|c\>=|z_1(\theta=5\pi/4)\>$.    
Moreover, $\<c|H_0 \cos(\theta)+H_1 \sin(\theta)|c\>=\cos \theta + \sin \theta$, so that 
$\lambda_1(\theta)+\lambda_1''(\theta)=0$, $\forall \theta$. 
In agreement with Theorems~\ref{t:dramatic_commutator},~\ref{t:dramatic_eigenvectors} 
and Corollary~\ref{c:back_to_Gutkin},  
the boundary has a sharp point, the vertex.  
This was already noted, using another approach, in~\cite{JonckheereAhmadGutkin}. 

As still noted in~\cite{JonckheereAhmadGutkin}, a perturbation of the $H_0,H_1$ data will smooth over the boundary, but at the expense of creating swallow tails, one of which in the gap area. 

\subsubsection{Singularity of convex hull}

Assume $H=\oplus_i H^{(i)}$ 
and $\mathcal{F}(H)=\mbox{conv}\{\cup_i \mathcal{F}(H^{(i)})\}$
to generalize the case of the Grover search where the numerical range is the convex hull of a point and an ellipse~\cite{adiabatic}. If $\cup_i \mathcal{F}(H^{(i)})$ is not convex, 
its convexification will add some edges to its boundary, 
either an edge connecting two isolated eigenvalues of $H$, 
or an edge made up with one isolated eigenvalue and the tangent to a convex, smooth $\partial\mathcal{F}(H^{(i)})$, 
or an edge made up of the common tangent to $\partial\mathcal{F}(H^{(i)})$ and $\partial\mathcal{F}(H^{(j)})$, 
assuming none of the fields of values contains the other. 
Write $[\nu_1,\nu_2]$ any such closed segment. 
This closed segment embedded in the boundary is nongeneric~\cite{JonckheereAhmadGutkin}. 
Generic or nongeneric, subsets of the boundary are still made up of critical values. 

\begin{theorem}
Regardless of whether the boundary $\partial \mathcal{F}(H)$ has nongeneric components, 
it is still composed of critical values of $f: \mathbb{CP}^{N-1} \to \mathbb{C}$.
\end{theorem}
\begin{proof}
The very general fact behind this result is the leitmotiv of~\cite{Jonckheere1997}, 
where many proofs in different contexts 
(e.g., the Brouwer domain invariance~\cite{Brouwer_domain_invariance}) are proposed. 
Closer to the present exposition, we could follow the argument depicted in Fig~\ref{f:geometry} 
taking $\theta$ to be the argument of the normal to $[\nu_1,\nu_2]$. 
We sketch a very general proof that relies exclusively on the smoothness of 
$f: \mathbb{S}^{2(N_1+N_2)-1} \to \mathbb{R}^2$, 
where $N_1$, $N_2$ denote the sizes of $H^{(1)}$, $H^{(2)}$, resp.  
Assume by contradiction that $\nu \in [\nu_1,\nu_2]$ is not critical. Hence, 
$d_{f^{-1}(\nu)}f: T_{f^{-1}(\nu)}\mathbb{S}^{2(N_1+N_2)-1}\to T_\nu \mathbb{R}^2$ has rank 2 and is hence submersive. 
Take a neighborhood $\mathcal{N}(\nu) \subset \mathbb{R}^2$; 
define $\mathcal{O}=f^{-1}(\mathcal{N}(\nu))$; 
by the implicit function theorem, 
$f(\mathcal{O})=\mathcal{N}(\nu)$. 
Hence some elements of the field of values would spill outside the convex hull, a contradiction. 
\end{proof}

In $\mbox{conv}\{\cup_i \mathcal{F}(H^{(i)})\}$, let $[\nu_1,\nu_2]$ be a convexifying line segment tangent to $\partial\mathcal{F}(H^{(1)})$ and $\partial\mathcal{F}(H^{(2)})$ 
at $\nu_1$, $\nu_2$, resp. 
Define $w_1\in \mathbb{S}^{2N_1-1}$, $w_2 \in \mathbb{S}^{2N_2-1}$ 
as the preimage points of $\nu_1,\nu_2$, that is, $f_1(w_1):=\<z_1|H^{(1)}|z_1\>=\nu_1$ 
and $f_2(w_2):=\<z_2|H^{(2)}|z_2\>=\nu_2$. In the bigger context of the direct sum of the two matrices, 
we have 
$\<w_1\oplus 0|H^{(1)}\oplus H^{(2)}|w_1\oplus 0\>=\nu_1$ 
and $\<0 \oplus w_2|H^{(1)}\oplus H^{(2)}|0 \oplus w_2\>=\nu_2$. 
Set $H:=H^{(1)}\oplus H^{(2)}$ to simplify the notation. 
Setting $z:=(z_1,z_2)\in\mathbb{S}^{2(N_1+N_2)-1}$, we define $f(z):=\<z|H|z\>$. 
To operate from $\mathbb{S}^{2N_1-1}$ and $\mathbb{S}^{2N_2-1}$ to $\mathbb{S}^{2(N_1+N_2)-1}$, we need to consider the embeddings
\[ \mathbb{S}^{2N_i-1} \hookrightarrow \mathbb{S}^{2(N_1+N_2)-1}, \quad i=1,2, \]
which is done as follows:
\[  \mathbb{S}^{2N_1-1} \subset \{z\in \mathbb{S}^{2(N_1+N_2)-1}: z_2=0\}, 
\quad \mathbb{S}^{2N_2-1} \subset \{z \in \mathbb{S}^{2(N_1+N_2)-1}: z_1=0\}. \]
With this embedding, observe that $\mathbb{S}^{2N_1-1}\cap \mathbb{S}^{2N_2-1}=\emptyset$. 
We are now in a position to compute $f^{-1}([\nu_1,\nu_2])$. 
Set $\nu_\lambda=\lambda \nu_1+(1-\lambda)\nu_2$. With $\|w_1\|=\|w_2\|=1$, it is easily seen that 
$f^{-1}(\nu_\lambda)=\sqrt{\lambda}(w_1,0)+\sqrt{1-\lambda}(0,w_2)$. 
This preimage follows a track between the two spheres 
$\mathbb{S}^{2N_1-1}$ and $\mathbb{S}^{2N_2-1}$.

\subsubsection{Unfolding}

Let $H_0+\imath H_1$ have two eigenvalues $\mu_1$, $\mu_2$  
with respective eigenvectors $z_1$, $z_2$ that are also eigenvectors of $H^\dagger$. 
Such eigenvectors are called \emph{normal eigenvectors} in~\cite[Def. 6]{JonckheereAhmadGutkin}. 
From~\cite[Th. 2]{JonckheereAhmadGutkin}, such eigenvectors are \emph{rank 0 critical points} 
in the sense that $\mbox{rank}(d_{z_i}f)=0$; moreover, from~\cite[Th. 2]{JonckheereAhmadGutkin}, 
$\mu_i\in\partial \mathcal{F}(H)$ and are \emph{sharp points}, that is, points where 
$T_{f(z_i)} \partial \mathcal{F}(H)$ is not defined. 
Such a situation is nongeneric~\cite[Ths. 10, 11; Cor. 3]{JonckheereAhmadGutkin}. 
Assuming further that $\mu_1-\mu_2^* \ne 0$, 
a little linear algebra shows that $\<z_1|z_2\>=0$. Via the unitary transformation $U=(z_1,z_2)\oplus I$, 
the matrix $H$ can be block-diagonalized as $\diag(\mu_1,\mu_2) \oplus H_{33}$. 
Moreover, $H_{33}$ can be put in upper-triangular Schur form $S_{33}$ by a further unitary transformation 
$I_{2 \times 2} \oplus U_{33}$. 
Next, since the numerical range is invariant under unitary transformation,  
the numerical range of $H$ is the numerical range of its upper-triangular Schur form~\cite[Lemma 1]{JonckheereAhmadGutkin},  
\begin{equation}\label{e:S} 
S=\left(\begin{array}{cc|cc}
\mu_1 & 0 & 0 \\
0 & \mu_2 & 0 \\\hline
0 & 0 & S_{33},  
\end{array}\right), \quad \left(S_{33}\right)_{i>j}=0.  
\end{equation}
It then follows that $\mathcal{F}(H)=\mathrm{conv}\{\mu_1,\mu_2,\mathcal{F}(S_{33})\}$.  
We further assume that $[\mu_1,\mu_2]\cap\mathcal{F}(S_{33})= \emptyset$, 
so as to expose the line segment $[\mu_1,\mu_2]$ as part of the boundary 
with $\mu_1, \mu_2 \in \partial \mathcal{F}(H)$ as sharp points. 
A line segment embedded in the boundary is nongeneric~\cite[Th. 12]{JonckheereAhmadGutkin}.

Next, we confirm the previous lack of genericity by the $\theta$-analysis. 
We need to formalize the passage from $S_0+\imath S_1$ to $S_0 \cos \theta + S_1\sin\theta$, 
where $S_0$ denotes the Hermitian part of $S$ and $\imath S_1$ its skew-Hermitian part. 
Proceeding from $H_0+\imath H_1=U(S_0+\imath S_1)U^\dagger$, we obtain
\begin{subequations}
\begin{align}
H_0\cos\theta+ H_1\sin\theta&=\mbox{Herm}\left((H_0+\imath H_1)e^{-\imath \theta}\right)\\
              &= \mbox{Herm}\left(U(S_0 + \imath S_1)e^{-\imath \theta}U^\dagger\right)\\
&= U\underbrace{(S_0\cos \theta + S_1 \sin\theta)}_{=:S_\theta}U^\dagger, 
\end{align} 
\end{subequations}
where $\mathrm{Herm}$ denotes the Hermitian part. 
Write $\mu_i=a_i+\imath b_i$ 
from which it follows that the matrix $H_0\cos \theta + H_1 \sin \theta$ is unitarily equivalent to 
\begin{align*}  
&S_\theta=\\
&\left(\begin{array}{cc|c}
a_1 \cos \theta + b_1 \sin \theta & 0 & 0 \\
0 & a_2 \cos \theta + b_2 \sin \theta & 0 \\\hline
0 & 0 & \mbox{Herm}(S_{33})\cos \theta +(\mbox{Kerm}(S_{33})/i)\sin \theta
\end{array}\right), 
\end{align*}
where $\mbox{Kerm}$ denotes the skew-Hermitian parts of a matrix. 
From the above, it follows that $\lambda_i(\theta)=a_i\cos\theta+b_i\sin\theta$, 
and further, $\lambda_i(\theta)+\lambda''_i(\theta)\equiv 0$, $i=1,2$.
 Clearly, the gap evolves as $\lambda_2(\theta)-\lambda_1(\theta)=(a_2-a_1)\cos \theta +(b_2-b_1)\sin \theta$. 
The Fourier expansion of 
$\lambda_2(\theta)-\lambda_1(\theta)$ therefore starts at $k_{\mathrm{min}}=1$, and 
by Tabachnikov's theorem 
$\lambda_2(\theta)-\lambda_1(\theta)$ has at least two roots in $[0,2\pi)$, two {\it exact} crossings, 
confirming the lack of genericity of the situation~\cite[Def. 9]{JonckheereAhmadGutkin}. 

Instead of constructing a perturbed 
$\tilde{S}_\theta=\Herm(\tilde{S})\cos \theta +(\Kerm(\tilde{S})/\imath)\sin \theta$ 
with the objective of observing that the exact crossing 
$\tilde{\lambda}_2(\theta_X)-\tilde{\lambda}_1(\theta_X)=0$ disappears,  
we argue on the eigenvectors $z_1$, $z_2$,  
removing the property that 
these eigenvectors are also eigenvectors of $H^\dagger$. 
This is equivalent to perturbing the Schur form to 
\[ 
\tilde{S}=\left(\begin{array}{cc|c}
\mu_1 & 0 & 0 \\
0 & \mu_2 & 0 \\\hline
0 & 0 & S_{33}
\end{array}\right)
+u_{1,2} E_{1,2}+\sum_{i=1}^2 \sum_{j=3}^4 u_{i,j}E_{i,j}, 
\]
where $E_{i,j}$ is the $N \times N$ elementary matrix made up of $0$s everywhere except for a $1$ in position $(i,j)$. 
The $u_{i,j}$'s are the so-called {\it unfolding or control parameters}~\cite[Chap. 7]{CastrigianoHayes1993}.  
These are perturbations to ``unfold" the unstable boundary of $\mathcal{F}(S)$, 
like the one depicted in Fig.~\ref{f:unfolding}(a),  
to the stable critical value curves of $\mathcal{F}(\tilde{S})$, 
as shown in Fig.~\ref{f:unfolding}(b). 
It is convenient to factor the control parameters as $u_{i,j}=\varepsilon v_{i,j}$,
where $\{v_{i,j}\}$ denotes a properly normalized \emph{structure} and $\varepsilon$ a global \emph{scaling}. 
An unfolding is said to be \textbf{\textit{versal}} 
if it exposes all possible stable structures. 
Existence of a versal unfolding is here guaranteed by the unitary equivalence of any matrix to its Schur form. 
The difficult part is to chose a \emph{minimal} set of versal unfolding parameters, 
in which case the unfolding is said to be \textbf{\textit{universal}}~\cite[Chap. 7]{CastrigianoHayes1993}. 

A universal unfolding manifests itself globally or locally. We start with the global picture. 
In this specific setting of singular value curves of quadratic forms, 
the aim is to find a minimum set of perturbations of the Schur form that unfold the nongeneric situation 
illustrated in Fig.~\ref{f:unfolding}(a) into stable situations as shown in Fig.~\ref{f:unfolding}(b). 
The process is combinatorially finite: in total there are $2(1+4)=10$ real parameters and their combinations, 
amounting to $\sum_{i=1}^{10} {10 \choose i}$ possibilities. 
Nevertheless, the search can be reduced by remarking that it suffices to remove the property 
that the eigenvectors of $H$ associated with $\mu_1,\mu_2$ are also eigenvectors of $H^\dagger$. 
Removal of the $\mu_1$ common eigenvector of $H, H^\dagger$ is accomplished with $v_{12}$, $v_{13}$, or $v_{14}$, but the resulting field of values still has a convexifying line segment in its boundary. 
To remove this latter situation, an additional $v_{23}$ or $v_{24}$ achieves a stable $\lambda_2$-singular value curve. Therefore, only one unfolding parameter is needed: 
any linear combination of one element of $\{v_{12}, v_{13}, v_{14}\}$  and one element of $\{v_{23}, v_{24}\}$. 
Except for some exceptional linear combinations of all 5 elements, the critical value curves remain stalbe. 
This situation is illustrated in Fig.~\ref{f:unfolding}, where
\[\tilde{S}=\left(\begin{array}{cc|cc}
 5\imath & 0 & \epsilon\imath & \epsilon 2 \\
0       & 5 & \epsilon 1      & -\epsilon 3 \\\hline
0       & 0 & 1      & 2+3\imath\\
0       & 0 & 0      & 1+\imath
\end{array}\right).\]

Locally, the universal unfolding can be computed as follows: 
Write $\<z|S|z\>$ in terms of real \emph{local} coordinates as $f(x)$, $x\in \mathbb{R}^{2N-1}$, 
and by a translation of coordinates $f(0)=0$.  
Let $\{e_1,e_2\}$ be the natural basis of $\mathbb{R}^2$.
Let $\mathbb{R}[[x]]$ be the ring of formal power series in the indeterminate $x$, 
and let $\left(\mathbb{R}[[x]] \right)^2:=\mathbb{R}[[x]]e_1+ \mathbb{R}[[x]]e_2$. 
Let $\mathfrak{J}=\<\frac{\partial f}{\partial x_1},\cdots,\frac{\partial f}{\partial x_{2N-1}}\>$ be the Jacobi ideal generated by the partial derivatives. Then
\begin{theorem}[Fundamental theorem on stability and universal unfolding]\label{t:stability_unfolding}
The map $f:(\mathbb{R}^{2N-1},0)\to (\mathbb{R}^2,0)$ is stable if the quotient module 
\begin{equation}\label{e:quotient_module}
\left(\mathbb{R}[[x]] \right)^2\left/ \left\{  \mathfrak{J}, \left\{f^i e_j\right\}_{j=1,...,2N-1}^{i=1,2}\right\}\right.
\end{equation}
is generated over $\mathbb{R}$ by $\{e_1,e_2\}$. 
Next, assuming that $0$ is critical in the sense that the partial derivatives of $f$ vanish.  
If the quotient module is finitely generated by  $\{e_1,e_2\}$ and $\{g_i(x)\}_{i=1}^r$,  
the map is unstable and a universal unfolding is given by 
\[f(x)+\sum_{i=1}^ru_i g_i(x), \quad u_i \in \mathbb{R}. \]
\end{theorem}
\begin{proof}
The first part related to stabilty is in Arnold~\cite[Section 6.6]{ArnoldGusein-ZadeVarchenko1985}. 
The second part related to unfolding is in Arnold~\cite[Section 9.3]{ArnoldGusein-ZadeVarchenko1985}. 
\end{proof}

\subsubsection{Steep gap near unfolding}
Here we develop a geometric picture 
of the paradigm of a swallow tail close to the boundary creating a steep gap. 
The forthcoming Figures~\ref{f:high_barrier_small_y},~\ref{f:barrier_low_Hamming} 
of Sec.~\ref{s:constant_to_steep} further illustrate this situation and probably more importantly 
interpret the swallow tails as symptomatic of the tunneling.  
More specifically here, we proceed from a nongeneric case and show how, under perturbation, a swallow tail unfolds from the boundary. 
This pattern is quite visible 
while perturbing the unstable singularity of the Grover search~\cite[Figs. 4, 5]{adiabatic}.  

\begin{remark}
Theorem~\ref{t:stability_unfolding} is pivotal   
in the context of robust stability when uncertainties cause the problem to traverse an unfolding~\cite[Th. 23.19]{Jonckheere1997},\cite{robust_performance_open} and in the context of   
the stability of the Riccati equation~\cite{Nanaz_SanDiego,Nanaz_NotreDame};
computationally speaking, Ref.~\cite{Nanaz_Alaska} promotes Gr\"obner bases as the polynomial algebra solution to the quotient module problem~\ref{e:quotient_module}. 
\end{remark}

\begin{remark}
Broadly speaking, the latter subsection has developed the theory of stability of maps in the restrictive context of stability of the critical value curves under perturbation of the entries of the quadratic form matrix, 
where by ``stability" it is meant that the topology (number of swallow tails, cusps, and crossings) of the critical value curves remains the same under perturbation.  
This viewpoint is justified by~\cite[Th. 23.12]{Jonckheere1997}.  
The exposition was not meant to cover the theory of stable maps and their singularities~\cite{GolubitskyGuillemin1973}. 
From a deeper standpoint, 
what has been developed in Sec.~\ref{s:stability} is the stability of the Legendrian front 
as developed in~\cite[Th. 5.3]{Legendrian_unfolding}, 
where a criterion similar to that of Theorem~\ref{t:stability_unfolding} is developed 
with a global extension springing out of the local theory. 
\end{remark}

\section{Legendrian classification} 
\label{s:Legendrian}

\begin{figure}[t]
	\centering
	\mbox{
		\subfigure[The notion of contact element at $\gamma(\theta_2)$, 
		a choice between the two half-planes bounded by the tangent, 
		and the co-orientation vector pointing to the chosen half-plane]{\scalebox{0.4}{\includegraphics{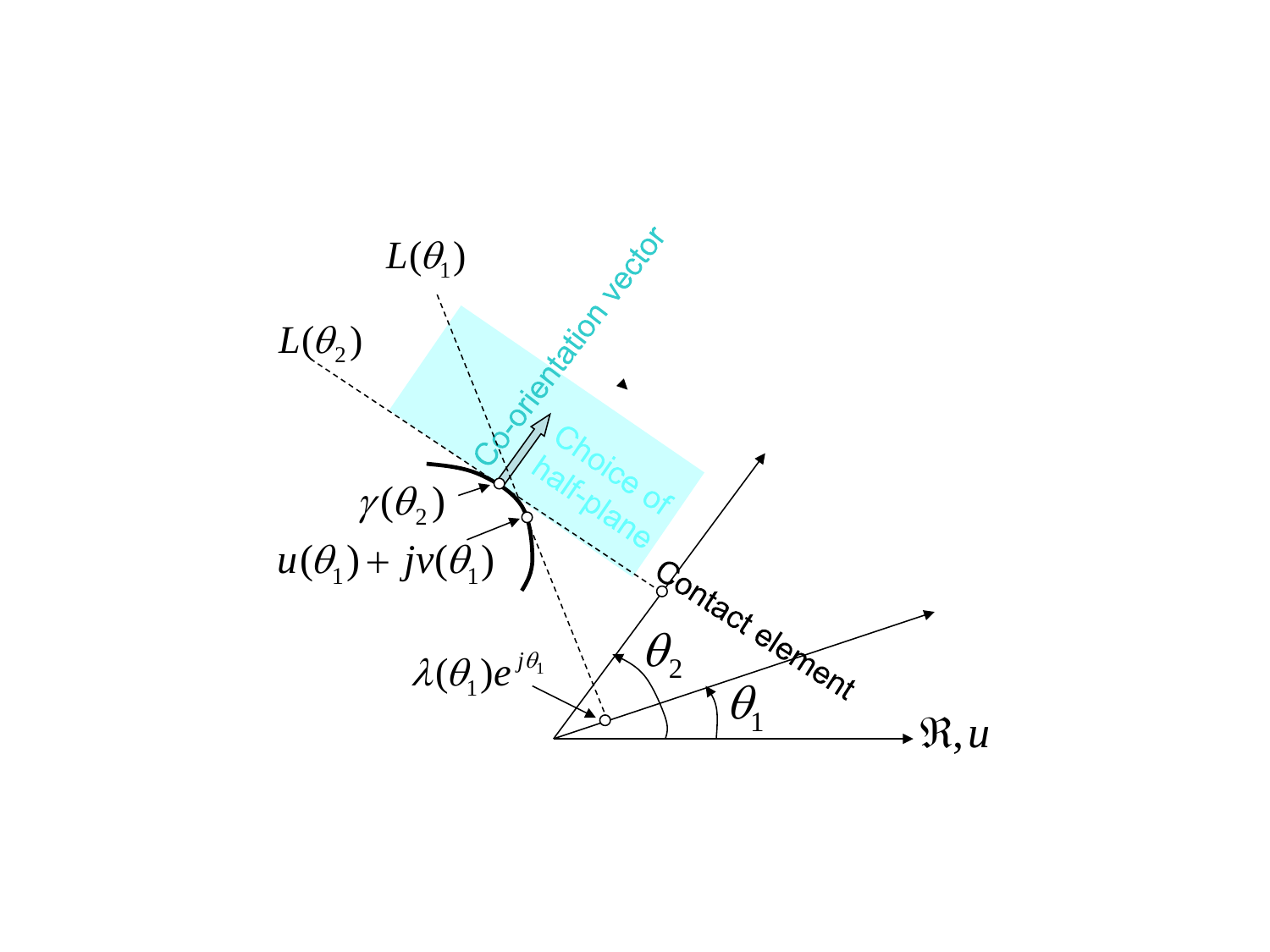}}}
		\quad		\subfigure[Orientation and co-orientation. Note that the co-orientation vector remains continuous at the cups. Also note that the co-orientation vector evolves trigonometrically along the orientation.]{\scalebox{0.4}{\includegraphics{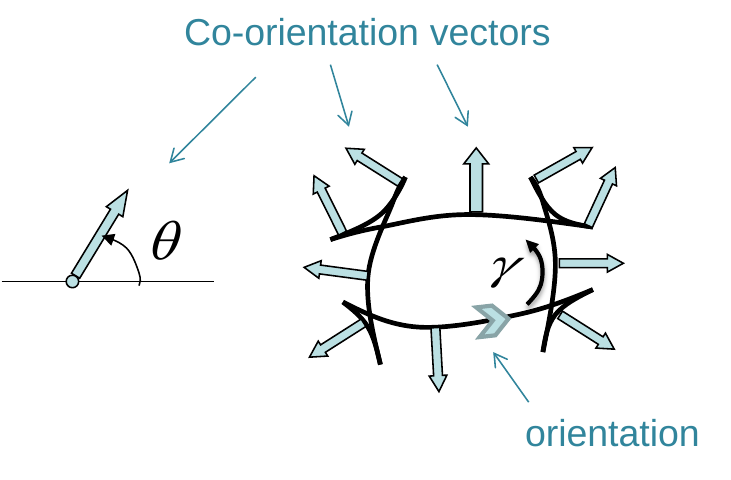}}}
	}
	\caption{Fundamental concepts of Legendrian approach}
	\label{f:Legendrian}
\end{figure}

So far, the main point is that the positioning of a swallow tail relative to the boundary of the numerical range dictates the steepness of the gap. But this approach is still local---involving \emph{one} swallow tail while there might be other swallow tails interconnected in some \emph{global} pattern. When  circulating along a singular path heading to a crossing, 
it becomes necessary to have a rule to decide whether the current branch lies above or below the one being intersected to determine whether 
the combination of such crossings could result in a \emph{knot}.  
Still more globally speaking, considering \emph{two} excited energy level critical curves, their above/below crossing patterns  could result in the two curves being \emph{linked}. This global viewpoint is precisely 
the Legendrian approach promoted by Arnold~\cite{Arnold1994},  
Thurston and Bennequin~\cite{Thurston_Bennequin_Maslov}, 
and Eliashberg~\cite{Eliashberg_original} and Chekanov~\cite{Chekanov_original}.

Central in the Legendrian approach is the concept of 
\textit{\textbf{contact element}} in $\mathbb{C}$, that is, a line $L_k(\theta)$ tangent to the curve 
$\theta \mapsto \gamma_k(\theta)=(x_k(\theta), y_k(\theta))$ of interest (typically $k=1,2$), 
as illustrated in Fig.~\ref{f:Legendrian}(a). Note that the contact element is still well-defined at the cusps. 
The contact element defines two half-planes bounded by such contact element, and 
the \textbf{\textit{co-orientation}} is a choice of one of the two half-planes bounded by the tangent $L_k(\theta)$ 
consistently along the critical values $(x_k(\theta),y_k(\theta))$ and across its cusps.  
A \textbf{\textit{co-orientation vector}} is a unit vector normal to the curve, 
pointing towards the half-plane determined by the co-orientation, 
as illustrated in Fig.~\ref{f:Legendrian}(a).   
Note that the co-orientation vector remains continuous across the cusps, 
as illustrated in Fig.~\ref{f:Legendrian}(b).   
The \textbf{\textit{orientation}} of the critical value curve is a choice of the direction of travel along 
the curve  
as illustrated in Fig.~\ref{f:Legendrian}(b).   
 
\subsection{Legendrian front: Arnold $J^+$-classification}

With a co-orientation, 
one can lift the critical value curves 
to the space $\mathbb{C} \times \mathbb{S}^1$ by the immersion  
$\gamma_k(\theta)=(x_k(\theta),y_k(\theta)) \mapsto (x_k(\theta),y_k(\theta),\theta \bmod 2\pi)=:\bar{\gamma}_k(\theta)$. 
An immediate result is that the lifted curve is dissingularized, 
in the sense that, if we take $\theta$ as the parameter of the lifted curve, 
the partial derivatives of $\bar{\gamma}_k$ relative to $x_k,y_k,\theta$ do not simultaneously vanish at any point.  
The manifold $M^3:=\mathbb{S}T^*\mathbb{C}$, the spherization of the cotangent bundle of $\mathbb{C}$ 
(diffeomorphic to $\mathbb{C} \times \mathbb{S}^1$), 
is referred to as $\textbf{\textit{contact space}}$ when it is endowed with 
a formal \textit{\textbf{contact structure,}} that is, 
an exterior differential form $\alpha: TM^3 \to \mathbb{R}$ 
that is maximally nonintegrable, that is,  
$\alpha_{\mathrm{nat}} \wedge d\alpha \ne 0 $. 
Here we take the \emph{natural} form in $\mathbb{S}T^*\mathbb{C}$, $\alpha_{\mathrm{nat}} = \cos \theta dx + \sin \theta dy$, 
for which it is easily verified that $\alpha_{\mathrm{nat}} \wedge d\alpha_{\mathrm{nat}} = -dx dy d\theta<0$.  
The kernel of this form defines a field of planes $\xi \in \mbox{ker}(\alpha_{\mathrm{nat}})$ in $\mathbb{C} \times \mathbb{S}^1$. 
The lifted curve is said to be \textbf{\textit{Legendrian}} as it is tangent to the field of planes, in other words, $\alpha(\gamma')=0$.  
The curve in the contact space is also referred as \textit{\textbf{Legendrian knot.}} 
The curve in $\mathbb{C}$ is referred to as the \textbf{\textit{Legendrian front}}.
 
It is important to observe that from their very construction the critical value curves have exactly 2 vertical tangents, 
that is, points where $dy/dx=\pm \infty$. This fact is also reflected by the natural from, 
from which it follows that the equation of the critical value curves 
is $dy/dx=-\cot \theta$ with diverging solution at $\theta=0,\pi$.

As already said, sorting out the knotting and linking requires a rule to determine which branch is above, 
which one is below, at a crossing in the front. 
By convention, at a self-crossing, $\gamma_k(\theta_{k,a})=\gamma_k(\theta_{k,b})$, 
the branch in the neighborhood of $\gamma_k(\theta_{k,a})$ will be said to be \textbf{\textit{above (below)}} 
the $\gamma(\theta_{k,b})$ branch if 
$\theta_{k,a}-\theta_{k,b}<0$ ($\theta_{k,a}-\theta_{k,b}>0$). 
At a crossing $\gamma_k(\theta_{k,a})=\gamma_\ell(\theta_{\ell,a})$ of different energy levels, 
the branch in the neighborhood of $\gamma_k(\theta_{k,a})$ will be said to be \textbf{\textit{above (below)}} 
the $\gamma(\theta_{\ell,a})$ branch if 
$\theta_{k,a}-\theta_{\ell,a}<0$ ($\theta_{k,a}-\theta_{\ell,a}>0$).

The nongeneric singularities of Legendrian fronts are triple crossings, self-tangencies, 
cusp birds, and cusp crossings~\cite[Chap. 2]{Arnold1994}. 
The self-tangencies are classified as positive if the two branches have consistent orientations 
and negative in case of opposite orientation. 
A self-tangency is referred to as \textbf{\textit{dangerous}} 
if both branches have consistent co-orientation at the self-tangency point~\cite[Sec. 10]{Arnold1994};
otherwise, it is said to be \textbf{\textit{safe}}~\cite[Fig. 3]{wave_front_Legendrian_knots}. 
The latter concept is illustrated in Fig.~\ref{f:J_minus}.
An invariant is said to be of  $J^+$-type if does not change under homotopies 
that involve no dangerous self-tangencies~\cite{wave_front_Legendrian_knots}. 
The $J^+$-theory classifies Legendrian fronts without self-tangencies of consistent orientation. 
The soon-to-be-developed winding number and Maslov index are invariants in the sense that they classify 
fronts up to isotopies that do not produce dangerous self-tangencies. 

While dangerous self-tangencies are ruled out, one could nevertheless consider self-tangencies 
in a limiting sense, but then the following restrictions emerge:

\begin{proposition}
On the $\gamma_1$ curve, a limiting move to self-tangency would reveal 
either a dangerous one with branches of opposite orientations 
or a safe one also with branches of opposite orientations (as shown in Fig.~\ref{f:J_minus}).
\end{proposition}
\begin{proof}
$\gamma_1$ has constant curvature, and it is easily sen that a self-tangency either safe or dangerous 
with branches of consistent orientations are not allowed 
as one of the branches would have positive curvature and the other negative curvature. 
\end{proof}

\subsubsection{Winding number}

\begin{definition}\label{d:index}
Given a co-orientation of an oriented generic critical value curve $\gamma_k$ in the plane,  
its \textbf{index}~\cite{Aicardi,Arnold1994}, $\mbox{ind}(\gamma_k)$, 
is the winding number of the co-orientation vector $e^{\imath \theta}(x_k(\theta),y_k(\theta))$ 
as the contact point travels around the curve from $\theta=0$ to $\theta=2\pi$. 
The winding number is positive if the contact point travels consistently with the orientation, 
negative otherwise. 
\end{definition}

Because of the geometry of Fig.~\ref{f:geometry}, 
one round trip around a critical value adiabatic curve is enough for $\theta$ to sustain a $2\pi$ change. 
 However, there are some cases, outside the realm of adiabatic processes, 
where it is necessary to circulate twice around the planar curve for $\theta$ 
to return to $2\pi$. 
This typically happens when $N$ is odd, say $N=3$. 
In this case, the singular curve in the field of values is 
an ideal hyperbolic triangle~\cite[Fig. 4]{JonckheereAhmadGutkin}, 
and upon one round trip $\oint_{\gamma_k} d\theta = \pi$.  
However, in the contact space, one travel around the curve yields $\oint_{\bar{\gamma}_k} d\theta = 2\pi$. 
Clearly, Definition~\ref{d:index} reworded for $\bar{\gamma}_k$ in the contact space becomes much more transparent. 

\begin{theorem}
The index of any critical value curve of an even-sized $H_0+\imath H_1$ is $\pm 1$. 
\end{theorem}

\begin{proof}
By its very construction, the curve $\gamma_k$ is parameterized 
by the argument $\theta \in [0,2\pi)$ of the co-orientation vector.  
By the geometry shown in Fig.~\ref{f:geometry}, 
as $\theta$ runs from $0$ to $2\pi$, the contact point on the curve circulates \emph{once} around the curve. 
Hence, the winding number is clearly $(1/2\pi)\oint_{\gamma_k} d\theta=\pm 1$,  
depending on whether the co-orientation vector 
circulates consistently with, or opposite to, the orientation, resp. 
\end{proof}

\subsubsection{Maslov index}

\begin{definition}\label{d:Maslov_index}
The \textbf{Maslov index of a cusp point} is said to be positive 
if the half-plane determined by the co-orientation vector at the cusp point contains the cusp branch going away from the cusp; 
otherwise, it is said to be negative~\cite{Aicardi,Arnold1994}. 
The number of positive (resp., negative) cusps of a generic Legendrian curve is written $\mu_+$ (resp., $\mu_-$). 
The \textbf{Maslov index
of a generic curve} $\gamma$ is defined as $\mu(\gamma)=\mu_+-\mu_-$.
\end{definition}

\begin{lemma}
Two consecutive cusps (that is, cusps occurring at $\theta_1 < \theta_2$ without cusps in $(\theta_1,\theta_2)$) have opposite indexes.
\end{lemma}

\begin{proof}
This is a corollary of the fact that the cusps are $3/2$, that is, the two branches abutting $\gamma_k(\theta_i)$ 
are on opposite sides of the common tangent at the cusp point. 
Assume that $\mbox{ind}(\gamma_k(\theta_1))=-1$, that is, the co-orientation half-plane contains the branch incoming to $\gamma_k(\theta_1)$, but does not contain the outgoing edge.  
Since $\theta$ increases from $\gamma_k(\theta_1)$ to $\gamma_k(\theta_2)$, 
the co-orientation half-plane does not contain the $\gamma_k(\theta_1)$ to $\gamma_k(\theta_2)$ cusped edge.  
But then at the $\theta_2$ cusp, the co-orientation half-plane does not contain the branch converging to $\gamma_k(\theta_2)$; 
hence the co-orientation half-plane will contain the branch going away from $\gamma_k(\theta_2)$, that is, 
$\mbox{ind}(\gamma_k(\theta_2))=+1$. 
\end{proof}

\begin{lemma}\label{l:even}
Let $f: [0,2\pi) \to \mathbb{R}$ be a $\mathcal{C}^2$-periodic function. If  $f(\theta)$ is not constant, 
then the equation $f(\theta)=0$ has an even number of solutions,  multiplicity counted. 
\end{lemma}
The proof is implicitly contained in the original proof of Tabachnikov~\cite{Tabachnikov_original}, except for 
the case of multiple solutions. Moreover, the following proof does not appeal to the starting order of the Fourier expansion. 

\begin{proof}
Assume $f(0)=f(2\pi) \ne 0$. (Or else, invoking the non-constancy hypothesis, 
shift the domain by $\theta_0$ so that $f(\theta_0)=f(\theta_0+2\pi)\ne 0$.)  
Without loss of generality, assume $f(0)=f(2\pi) >0$. The point $(0,f(0))$ has to connect to $(2\pi,f(2\pi))$. Follow the graph of the function from $(0,f(0))$; 
if it does not cross the $0$-level set, the lemma is proved. Assume now the graph crosses the $0$ level set  transversally at $\theta_1$; 
the crossing counter is set to 1. But the graph has to cross again the $0$-level to connect to  $(2\pi,f(2\pi))$. Assume it crosses this level at $\theta_2$;
the counter is set to $2$. From here on, either the graph connects to  $(2\pi,f(2\pi))$ without further crossing or it goes again to the sublevel set,
from where an inductive argument shows that the counter increments in case of transversal crossings are always even. 
Should a crossing by non-transversal, that is, tangential, this updates the counter by 2, again even. 
\end{proof}

\begin{theorem}
The Maslov index $\mu(\gamma_k)$ of any critical value curve $\gamma_k$ is $0$. 
\end{theorem}

\begin{proof}
It suffices to show that, generically, the number of cusp points on a critical value curve is even. 
The cusp points are clearly given by the transversal (non-multiple) solutions to $\lambda_k(\theta)+\lambda_k''(\theta)=0$. 
But this function is periodic with period $2\pi$. Hence by the lemma, 
the plot of $\lambda_k(\theta)+\lambda_k''(\theta)$ crosses the $\theta$-axis, the $0$ level set,  
at an even number of points.  
\end{proof}

In case $\lambda_k(\theta)+\lambda_k''(\theta)$ has a double solution, 
the two cusps of the swallow tail have merged and the swallow tail has disappeared. 

\subsubsection{Classification}

Using Arnold's notation~\cite[Chap. 2]{Arnold1994}, let $\Omega_{i,\mu}$ denote the set of $J^+=0$ Legendrian curves of index $i$ and Maslov index $\mu$. 
The following corollary classifies all adiabatic problems among Legendrian curves. 

\begin{corollary}
The critical value curves of any generic adiabatic problem belong to $\Omega_{1,0}$. $\blacksquare$
\end{corollary}

Unfortunately, $\Omega_{1,0}$ is too broad a class, 
because, as shown in Arnold~\cite[Fig. 26]{Arnold1994}, it contains curves never seen in adiabatic processes; 
nevertheless, it does contain swallow-tailed curves close to, but differebnt from,  
the swallow-tailed quadrilateral of Fig.~\ref{f:geometry} typical of adiabatic processes. 
But most importantly, as shown in Table~\ref{t:classification}, it does not discriminate among many adiabatic-relevant cases, 
nor does it discriminate among the three cases shown in Fig.~\ref{f:J_minus}. 

\begin{remark}
The terminology of {\it front} or {\it wave front} is justified by the swallow tail singularities 
of an implosive wavefront propagating from a smooth closed curve that is not a sphere, that has varying curvature. 
Let us chart the closed curve by its co-orientation $\theta$ and 
suppose the implosive phenomenon isotropically propagates at unit speed from $t=0$. 
It is shown in~\cite[Sec. 1.4]{Arnold1993} that the wavefront is smooth iff $t<R(\theta)$, 
where $R(\theta)$ is the radius of curvature of the closed curve. 
On the other hand, a swallow tail will develop in the $\theta$-region where $t>R(\theta)$. 
Thus by back-stepping along the co-orientation $\theta$ for a distance of $\lambda_k(\theta)+\lambda_k''(\theta)$ we recover the origin of the 
propagation phenomenon, which completely justifies the terminology of {\it front} for the $\gamma_k$ critical value curve.	
\end{remark}

\begin{remark}
Definition~\ref{d:Maslov_index} of the Maslov index is not the original one of Maslov's index for a loop of Lagrangian subspaces; 
it is Arnold's reinterpretation of it~\cite{Leray_Maslov,Maslov_index_for_paths}. 
The connection between Lagrangian subspaces and Legendrian knots is provided in~\cite{Lagrangian_skeleta}.
\end{remark}

\begin{figure}[t]
	\centering
	\mbox{
\subfigure[Ideal quadrilateral (sum of internal angles = 0) with orientation and co-orientation]{\scalebox{0.55}{\includegraphics{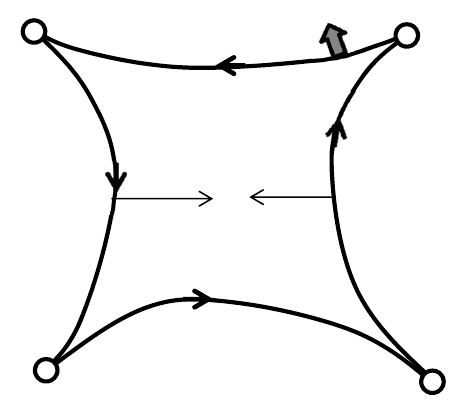}}}\quad
		\subfigure[Safe self-tangency with opposite orientation]{\scalebox{0.55}{\includegraphics{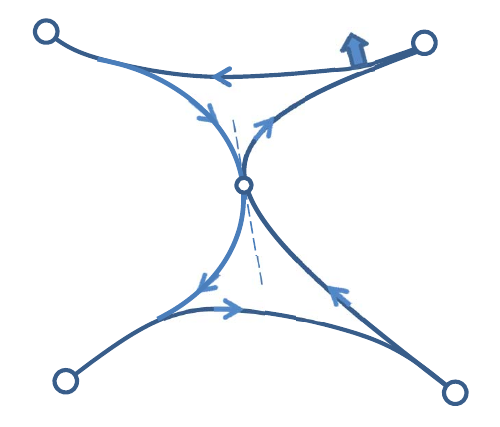}}}
\quad		\subfigure[Beyond the safe self-tangency with emergence of pair of swallow-tails. Observe the knotting,]{\scalebox{0.55}{\includegraphics{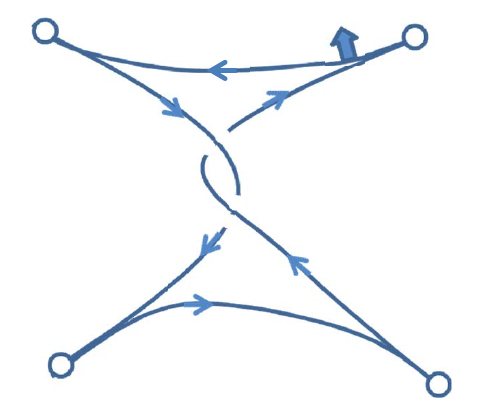}}}
	}
	\caption{Crossing a safe self-tangency with opposite orientation with emergence of swallow tails similar to those appearing in tunneling as shown in Fig.~\ref{f:high_barrier_small_y}}
	\label{f:J_minus}
\end{figure}

\subsection{Legendrian front: Rotation and Thurston-Bennequin invariants}

So far, we have used the \emph{natural} contact form $\alpha_{\mathrm{nat}}=dx \cos \theta + dy \sin \theta$  
in the contact space $\mathbb{S}T^*\mathbb{C}$ where lies the Legendrian knot that projects to the 
critical value curve. 
The Thurston-Bennequin and the so-called \emph{new invariants} of Chekanov and Eliashberg~\cite{computable_Legendrian_invariants}  
rely on the (Darboux) \emph{standard} contact form~\cite{contact_geometry_another_summary,Thurston_Bennequin} 
$\alpha_{\mathrm{std}}=du-pdq$ in the contact space identified with the space $J^1(\mathbb{R},\mathbb{R})$ 
of 1-jets of functions to the real line. 
Note that $\alpha_{\mathrm{std}}\wedge d\alpha_{\mathrm{std}}$ is also a everywhere nonvanishing volume form. 
The two forms can be related through the contactomorphism~\cite{fronts_of_Legendrian_links} 
$(x,y,\theta) \mapsto (p,q,u)$
defined by 
$p=-x \sin \theta + y \cos \theta$, $u=x \cos \theta + y \sin \theta$, and $q=\theta$.  
The standard contact form allows us to look at the curve in the $(p,\theta)$ plane, 
the so-called \textit{\textbf{Lagrangian projection}}, which is an immersion of the circle 
with crossings but no swallow tails. Here, however, we will work for the most part with the frontal projection. 

\subsubsection{Removal of vertical tangency points}

A significant difference between the natural and standard forms is that, with $\alpha_{\mathrm{std}}$, 
vertical tangents are not allowed because the equation of the front projection is $du/dq=p$ with $p$ {\it finite.}  
This forces us to replace the vertical tangency points by cusps, 
as suggested in~\cite[Fig. 19]{Invariants_Legendrian_transverse_knots}, 
and proceed with the classification on the resulting curve $\hat{\gamma}_2$.  
There are two vertical tangency points, hence two additional cusps, 
one on the left side of the front diagram and the other on the right side. 
Moreover, the cusps could be positioned in two different ways:  
either left-handed (cusp pointing to the left) or right-handed (cusp pointing to the right). 
This dichotomy is resolved by orienting the cusp in such a way 
that it does not add another vertical tangency point, 
the very situation the additional cusps are meant to remove. 
This process is illustrated in Fig.~\ref{f:breaking_vertical_tangents}.

\begin{lemma}
For a negatively curved front diagram $\gamma_2$, 
if, at the two vertical tangency points, 
one cusp is left-handed and the other right-handed, then the winding number is preserved, 
${\tt index}(\hat{\gamma})={\tt index}(\gamma)-1$; moreover, $\mu(\hat{\gamma})=\mu(\gamma)$. 
\end{lemma}

\begin{proof}
The guiding idea is that the positioning of the additional cusps respects the curvature. 
Along an edge of $\gamma$,  
$-\tfrac{\pi}{2}<\int_{\gamma_{\theta_i}^{\theta_o}} d\theta <0$. 
However, after placing the cusp,  
$\int_{\hat{\gamma}_{\theta_i}^{\theta_o}} d\theta =\int_{\gamma_{\theta_i}^{\theta_o}} d\theta - \pi. $ 
This takes care of the first claim. 
The second claim follows easily from the very definition of the Maslov index: 
The left-side cusp causes $\mu$ to increase by one and the right-side cusp causes it to decrease by 1. 
\end{proof}

The forthcoming \emph{rotation} and \emph{Turston-Bennequin} number are invariant under Legendrian isotopy in the contact space 
or Legendrian Reidemeister moves in the front projection. 
(A \emph{Legendrian} isotopy is an isotopy such that along its track the immersion of $\mathbb{S}^1$ is 
Legendrian.)
There are three \textit{\textbf{Legendrian Reidemeister moves:}} swallow tail removal, 
cusp crossing above or below a smooth branch, 
and triple crossing 
when a smooth branch crosses two already crossing branches 
above the bottom branch and below the top branch~\cite[Sec. 5]{contact_geometry_another_summary}.  

If $\gamma_2$ has swallow tails as its sole singularities, they can be removed by Reidemeister moves, 
and consequently the classical invariants identify $\gamma_2$ as a circle as evident from Table~\ref{t:classification}.   
If $\gamma_2$ has (an even number of) swallow tails with vertical tangents along their edges, 
removal of the vertical tangency points destroys such swallow tails 
leaving the Reidemeister moves with less freedom to trivialize the curve $\hat{\gamma}_2$ to a circle. 
After removal of the vertical tangencies, the classical invariants applied to $\hat{\gamma}_2$  
provide a stronger classification revealing the $\gamma_2$-swallow tails, 
as shown by Fig.~\ref{f:breaking_vertical_tangents} and Table~\ref{t:classification}.

\subsubsection{Rotation, Writhe and Thurston-Bennequin invariant}

Given an orientation, the concept of \textit{\textbf{upward and downward cusps}} is self-explanatory if the tangent 
to the cusp is horizontal~\cite[Fig. 3]{Chekanov_Eliashberg_invariants}. 
If some tangents are not horizontal, they can be made horizontal by an isotopy  
 that does not affect the crossings.

\begin{definition}[\cite{Thurston_Bennequin_Maslov}]
The \textbf{\textit{rotation}} of a Legendrian front with all its cusps horizontal is defined as $\mathrm{rot}(\gamma_k)=\frac{1}{2}(D-U)$, 
where $D$ is the number of downward cusps and $U$ the number of upward cusps of $\gamma_k$. 
\end{definition}

\begin{proposition}[~\cite{C-E_invariants_knots}]
$\tfrac{1}{2}(D-U)=\tfrac{1}{2}(\mu_+-\mu_-)$, that is, the rotation is 1/2 of the Maslov index. 
\end{proposition}

Since the ``preprocessing" of a front to make all its cusps horizontal is somewhat artificial, 
we will in the sequel discard $\tfrac{1}{2}(D-U)$ and fall back on $\tfrac{1}{2}(\mu_+-\mu_-)$, 
which is well defined no matter what the angles of the cusps are.

\begin{definition}\label{d:writhe} 
Given an orientation and a co-orientation of a critical value curve, 
a crossing of two branches is said to be positive (negative) 
if the smallest angle rotation of the ``above" path to align its orientation with that of the ``under" path is trigonometric (clockwise). 
The \textit{\textbf{writhe}} $w(\gamma)$ of a Legendrian fronts is the number of positive crossings minus the number of negative crossings. 
\end{definition}

(Note that the writhe does not depend on the orientation, nor on the co-orientation.) 

\begin{definition}[\cite{Thurston_Bennequin_Maslov}]
The \textbf{\textit{Thurston-Bennequin number}} is the writhe minus 1/2 the number of cusps:
\[ {\tt tb}(\gamma)=w(\gamma)-\frac{1}{2}|\{\mathrm{cusps}\}|. \]
\end{definition}

\subsubsection{Classical invariants}

\begin{figure}[t]
	\centering
	\mbox{
\subfigure[The vertical tangency points occur on edges of opposite swallow tails. The $0^\circ$ and the $90^\circ$ swallow tail gives some similarity between the beginning and the end of the adiabatic process.]{\scalebox{0.5}{\includegraphics{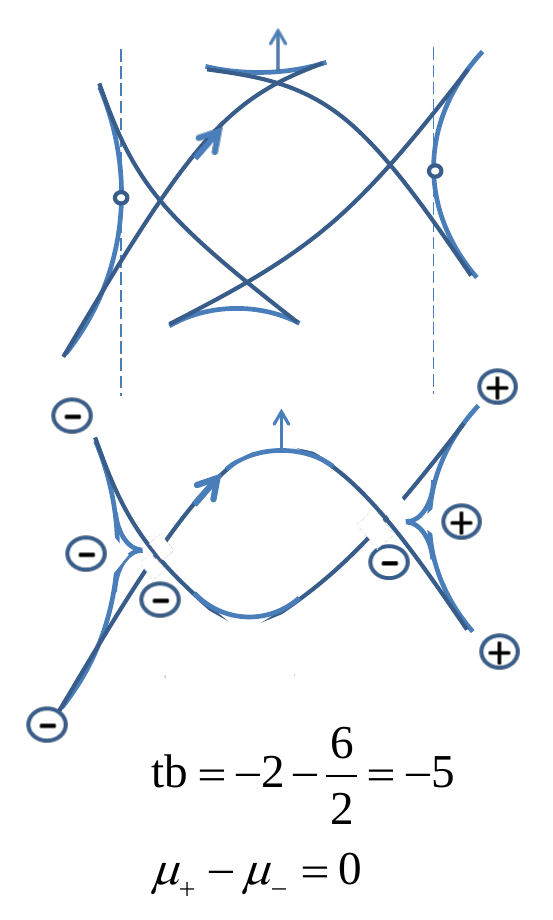}}}\quad
		\subfigure[The tangency points are on opposite branches joining two swallow tails. The swallow tail at 45$^\circ$ could be considered as ``dangerous," likely to create a steep gap,]{\scalebox{0.5}{\includegraphics{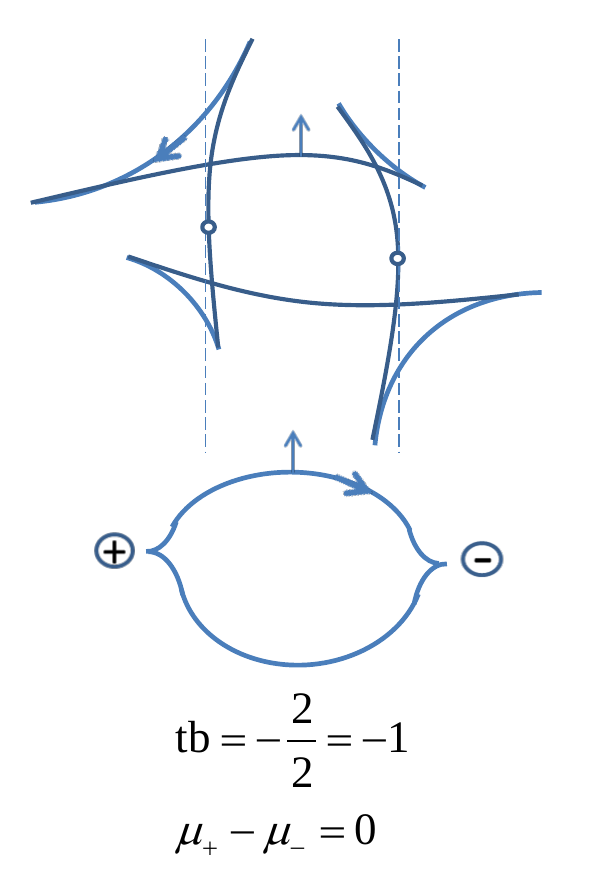}}}
\quad		\subfigure[Combination of the two preceding cases: the left tangency point is on a branch connecting two swallow tails, the other is on the edge of the opposite swallow tail.]{\scalebox{0.65}{\includegraphics{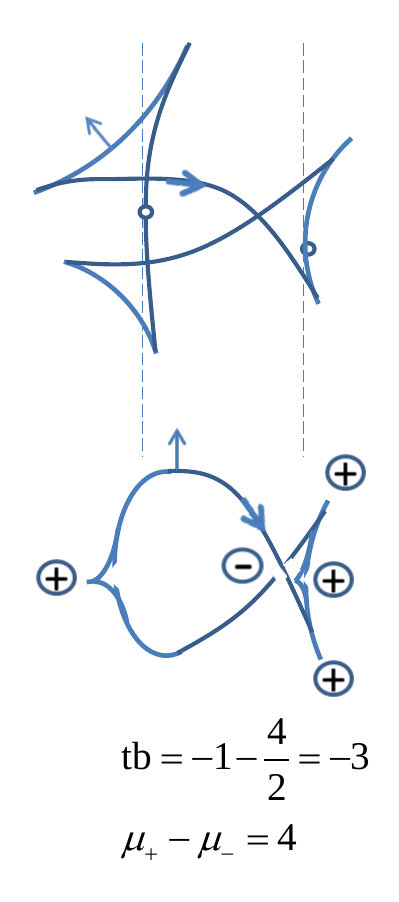}}}
	}
	\caption{Replacement vertical tangency points by cusps in 3 different configurations of 
critical value curves: 
top, $\gamma_2$ curves, bottom, $\widehat{\gamma}_2$ curves. 
The $\pm$ signs of crossings refer to the writhe and those of cusps refer to the Maslov index.}
	\label{f:breaking_vertical_tangents}
\end{figure}

\begin{theorem}
The Maslov index (or rotation) and the Thurston-Bennequin number, $\mu(\gamma)$ and ${\tt tb}(\gamma)$, are invariants of the Legendrian front $\gamma$ under Legendrian Reidemeister moves. 
Together they form the \textbf{\textit{classical invariants}}, $(\mu,{\tt tb})$.
\end{theorem}
\begin{proof}See~\cite[5.11]{contact_geometry_another_summary}.
\end{proof}

Applied to all cases of critical value curves having only swallow tails singularities---which constitute 
the most important class of critical value curves of $H_0+\imath H_1$ 
relevant to quantum adiabatic computations~\cite{adiabatic}---the removal of vertical tangents 
by replacement of their contact points by cusps 
followed by Reidemeister removal of the remaining swallow tails   
results in the 3 fundamental structures 
shown in Fig.~\ref{f:breaking_vertical_tangents} + the mirror image of Fig.~\ref{f:breaking_vertical_tangents}(c). 
Some details are compiled in Table~\ref{t:classification}. 
What distinguishes the 4 cases is whether a vertical tangency point is on an edge of a swallow tail or a branch connecting two successive swallow tails and more importantly the co-orientation position $\theta \in [0,2\pi)$ 
of the swallow tails.   
The Thurston-Bennequin number {\tt tb} of the $\hat{\gamma}_2$ critical value curves are computed 
after removal of the vertical tangents: 
\begin{enumerate}
\item For the swallow-tailed quadrangle, 
the vertical tangency points could occur either 
\begin{enumerate}
\item on opposite edges of swallow tails 
(Fig.~\ref{f:breaking_vertical_tangents}(a)) 
\item or on opposite branches connecting two opposite swallow tails 
(Fig.~\ref{f:breaking_vertical_tangents}(b)). 
\end{enumerate}
The two cases are discriminated by the ${\tt tb}$ but not by the $\mu_+-\mu_-$ invariant.  
\item For the swallow-tailed triangle, 
one vertical tangency point is on the edge of one of the swallow tails 
and the other on the branch joining the two other swallow tails. 
\begin{enumerate}
\item If the swallow tail broken by a cusp at its vertical tangency point is on the right side, 
(Fig.~\ref{f:breaking_vertical_tangents}(c)), somewhat as a surprise, the Maslov index $\mu_+-\mu_-=4$ discriminates this case against the preceding ones, as does the ${\tt tb}=-3$. 
\item However, under a vertical mirror symmetry when the broken swallow tail is on the left, 
the Maslov index changes to $\mu_+-\mu_-=-4$ 
while ${\tt tb}=-3$ remains the same. 
\end{enumerate}
\end{enumerate}

There remain 2 other cases: (i) the ideal hyperbolic triangle as shown in Fig.~\ref{f:} 
and the ideal hyperbolic rhombus~\cite[Fig. 9]{adiabatic}. 

\begin{enumerate}
\item The ideal hyperbolic triangle is not co-orientable  
making the applications of the {\tt tb} invariant impossible, unless we consider its double cover. 
From there on, the Arnold analysis goes through (see Table~\ref{t:classification}). 
For the Thurston-Bennequin number, 
removal of the vertical tangent makes it co-orientable resulting in the data shown in Table~\ref{t:classification}
\item The ideal hyperbolic rhombus has its vertical tangency points at two vertically opposed cusps. 
By a Reidemeister move, two opposite branches are brought to cross, one on top of the other. 
This results in two swallow tails and  
vertical tangencies on the crossing branches. After replacing the vertical tangency points by cusps followed by a removal of the swallow tails yields the data shown in Table~\ref{t:classification}.  
\end{enumerate}

\begin{table}[t]
\begin{tabular}{|c||c|c||c|c|c|}\hline
\multicolumn{1}{|c||}{}&\multicolumn{2}{c||}{Arnold $\gamma_2$}&\multicolumn{3}{c|}{Thurston-Bennequin $\hat{\gamma}_2$}\\\hline\hline
 & winding number & Maslov $\mu=\mathrm{rot}$& writhe & $\mu$ &{\tt tb}\\\hline\hline
swallow-tailed triangle & -1 & 0 & -1 & 4 &-3 \\\hline
swallow-tailed quadrangle & -1 & 0 & -2 & 0&  -5\\\hline
ideal hyperbolic triangle & -1 & 0 & $0$ & 2 & -2\\\hline
ideal hyperbolic rhombus & -1 & 0 & 0 & 0& -1\\\hline
\end{tabular}
\caption{Classification  
of the most frequently 
occurring first excited level critical value (trigonometrically oriented) curves 
under the assumption that for the quadrangle the vertical tangency points occurs at opposite edges of swallow tails
and for the triangle only one tangency point occurs at the edge of a swallow tail: 
the swallow-tailed triangle~\cite[Fig. 4]{adiabatic}, 
the swallow-tailed quadrangle~\cite[Fig. 5]{adiabatic}, 
the hyperbolic triangle,
and the ideal hyperbolic \emph{rhombus}~\cite[Fig. 9]{adiabatic}.}  
\label{t:classification}
\end{table}

\subsubsection{``New" invariants: Chekanov-Pushkar}

Two Legendrian isotopic knots have the same classical invariants, but the converse is not, in general, true:  
Chekanov and Eliashberg found a pair of knots~\cite[Fig. 2(c)]{augmentations_and_rulings_of_Legendrian_knots} 
that are not equivalent but share the same ${\tt tb}=1$ and $\mathrm{rot}=0$. 
Among the ``new" invariants able to make a sharper discrimination, 
one will mention the Chekanov-Eliashberg 
differential algebra approach~\cite{Chekanov_original,Eliashberg_original,Chekanov_Eliashberg_invariants,
C-E_invariants_knots,computable_Legendrian_invariants}
and the Chekanov-Pushkar combinatorics of decomposition of fronts~\cite{fronts_of_Legendrian_links} 
along with Traynor's generating families~\cite{generating_function_polynomials}.  
The latter meshes better with the current development. 
These ``new" invariants will be examined in a further paper. 

\subsection{Knotting and linking}
\label{s:knotting_linking}

\begin{figure}[t]
	\centering
	\mbox{
\subfigure[case of one smooth section of curve intersected by two btranches of a cusp. In the contact space such curves are {\bf not linked.}]{\scalebox{0.5}{\includegraphics{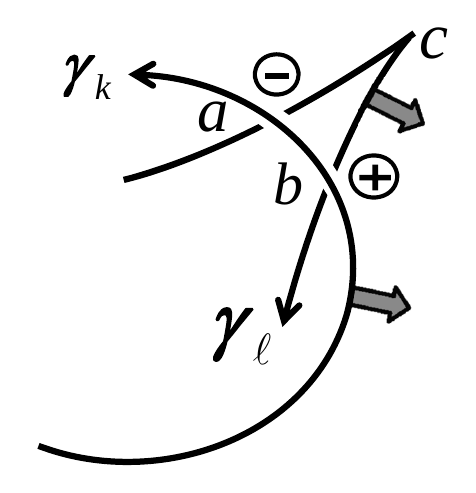}}}\quad		
\subfigure[Case of smooth sections of curves between intersection points. In the contact space such curves are {\bf linked.}]{\scalebox{0.5}{\includegraphics{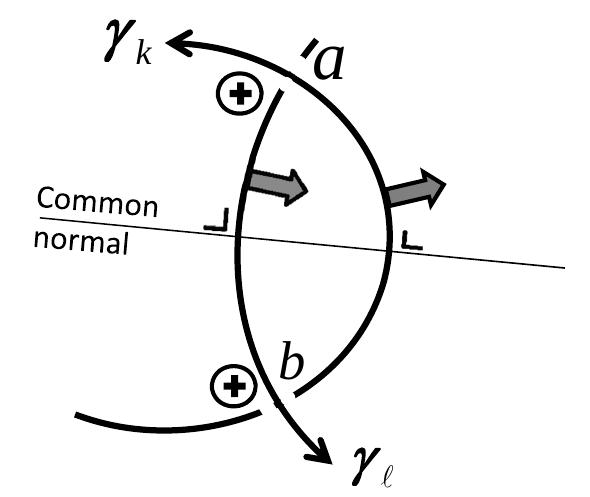}}}\quad
		\subfigure[Case of one smooth section of curve intersected by the two branches of a swallow tail. In the contact space such curves are {\bf linked.}]{\scalebox{0.5}{\includegraphics{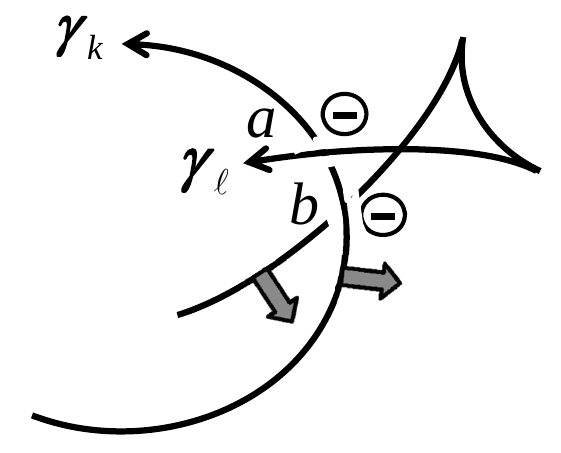}}}
	}
	\caption{Three possible cases of intersection of critical value curves in $\mathbb{C}$ crucially 
depending on the orientation (arrow along the curve) and the co-orientation (arrow normal to the curve)}
	\label{f:over_under}
\end{figure}

This subsection introduces the self-knotting of a single critical value curve and the linking number among many critical value curves in the contact space with contact structure $dx \cos \theta + dy \sin \theta$ 
that allows vertical tangents.  

\subsubsection{ground and first-excitation levels}

\begin{theorem}
The Legendrian ``knot" of the first excited critical value curve in the contact space is unknotted.
\end{theorem}

\begin{proof}

Consider a planar critical value curve made up with swallow tail singularities. 
Every swallow tail can be removed by a  Legendrian 
Reidemeister move~\cite[Fig.10]{contact_geometry_another_summary}, 
leaving a simple closed curve in the plane. 
Since the Reidemeister moves do not affect the knotting number~\cite[Chap. 3]{Lickorish1997}, 
the original critical value curve with swallow tails is unknotted. 
In case the critical value curve is an ideal quadrangle~\cite{adiabatic}, 
the result is obvious. 
\end{proof}

\begin{definition}[\cite{Lickorish1997}, Def. 1.4]
The \textbf{\textit{linking number}} $\mathrm{lk}(\bar{\gamma}_k,\bar{\gamma}_\ell)$ of two curves in the contact space is $\tfrac{1}{2}$ the sum of the signs of their crossings, where the ``sign" is defined as in Definition~\ref{d:writhe}.
\end{definition}

\begin{theorem}
Assume the gap never closes, that is, $\lambda_2(\theta)-\lambda_1(\theta) > 0$, $\forall \theta$. 
Then the ground state and the first excited state critical value curves in the contact space are not linked, 
that is $\mathrm{lk}(\bar{\gamma}_1,\bar{\gamma}_2)=0$. 
\end{theorem}

\begin{proof}
 Indeed, for the curves in the contact space to be linked, they would have to cross in the plane, 
that is, the first two energy levels would have to cross, which is not the case. 
\end{proof}

\subsubsection{Higher excitation levels}

The constraints imposed on the ground level and first excited level critical value curves become relaxed in regard to the first and higher excitation levels, 
where more complicated crossing phenomena can occur in the front~\cite{adiabatic}. 
A relevant such intersection is a smooth branch of a $\gamma_k$ curve intersecting both edges $\gamma_\ell$ of a cusp as shown in Fig.~\ref{f:over_under}(a). Such intersection appears in the adiabatic 
approach to the quantum hitting time of a Markov chain~\cite[Sec. 6]{adiabatic}. 

\begin{theorem}\label{t:cusp}
Consider two $\theta$-parameterized critical value curves $\gamma_k$, $\gamma_\ell$, 
restricted to $\Theta_k \subset [0,2\pi)$, $\Theta_\ell \subset [0,2\pi)$, 
where $\gamma_\ell(\Theta_\ell)$ contains a cusp $c$ but is otherwise smooth and $\gamma_k(\Theta_k)$ is smooth with its curvature bounded from above, with $\gamma_k$ crossing the two branches of the cusp at  $a=\gamma_k(\theta_{ka})=\gamma_\ell(\theta_{\ell,a})$ and 
$b=\gamma_k(\theta_{kb})=\gamma_\ell(\theta_{\ell,b})$. 
Then one curve, say $\gamma_k$, will be completely ``over" or completely ``under" $\gamma_\ell$ 
with opposite signs of crossings, as defined in Definition~\ref{d:writhe}. 
Hence, in the contact space, $\mathrm{lk}(\bar{\gamma}_k,\bar{\gamma}_\ell)=0$, 
that is, the two Legendrian curves are not linked.  
\end{theorem}

\begin{proof}
 Essentially, we have to show that 
\[\mbox{sign}(\theta_{k,a}-\theta_{\ell,a})=\mbox{sign}(\theta_{k,b}-\theta_{\ell,b})\]  
By the trivial move of bringing $\gamma_k$ arbitrarily close to, but not crossing, the cusp $c$, 
a move  that doesn't change the linking, we show that the difference
\[(\theta_{k,a}-\theta_{\ell,a})-(\theta_{k,b}-\theta_{\ell,b})=
(\theta_{k,a}-\theta_{k,b})-(\theta_{\ell,a}-\theta_{\ell,b})\]
can be made so small as to keep the signs equal. 
Moreover, by the same move, one can cancel $\theta_{k,a}-\theta_{k,b}$. 
It follows that 
\[|(\theta_{k,a}-\theta_{\ell,a})-(\theta_{k,b}-\theta_{\ell,b})|\leq |(\theta_{\ell,a}-\theta_{\ell,b})|\]
It is claimed that as $\gamma_k$ gets arbitrarily close to the cusp, 
one can achieve the inequality
\[|\theta_{\ell,a}-\theta_{\ell,b}| < \min \{|\theta_{k,a}-\theta_{\ell,a}|,|\theta_{k,b}-\theta_{\ell,b}|\}\]
Indeed, as $\gamma_k$ converges to the cusp, the LHS decreases to zero while the RHS increases to 90 deg. 
It follows that
\[|(\theta_{k,a}-\theta_{\ell,a})-(\theta_{k,b}-\theta_{\ell,b})|
< \min \{|\theta_{k,a}-\theta_{\ell,a}|,|\theta_{k,b}-\theta_{\ell,b}|\}\]
A simple contradictory argument shows that if the signs in the LHS are different, then the inequality is violated.  
This implies that $(\theta_{k,a}-\theta_{\ell,a})$ and $(\theta_{k,b}-\theta_{\ell,b})$ have the same sign.
\end{proof}

\begin{theorem}\label{t:smooth_curve}
Consider two $\theta$-parameterized critical value curves $\gamma_k$, $\gamma_\ell$, 
restricted to $\Theta_k \subset [0,2\pi)$, $\Theta_\ell \subset [0,2\pi)$, 
where $\gamma_k(\Theta_k)$ and $\gamma_\ell(\Theta_\ell)$ are smooth, with the two curves crossing at  $a=\gamma_k(\theta_{ka})=\gamma_\ell(\theta_{\ell,a})$ and 
$b=\gamma_k(\theta_{kb})=\gamma_\ell(\theta_{\ell,b})$. 
Then  
one curve, say $\gamma_k$, will intersect the other $\gamma_\ell$ ``over" (``under") at the point $a$ 
and then reintersect $\gamma_\ell$ ``under" (``over") at the point $b$ 
with consistent signs of crossings, as defined in Definition~\ref{d:writhe}. 
Hence, in the contact space, $\mathrm{lk}(\bar{\gamma}_k,\bar{\gamma}_\ell)=\pm 1$, 
that is, the two Legendrian curves are linked.  
\end{theorem}

\begin{proof}
The proof relies on the common normal to two smooth curves. 
(Existence of the common normal follows from existence of a smooth solution to 
$\max_{\theta_k,\theta_\ell} |\gamma_k(\theta_k)-\gamma_\ell(\theta_\ell)|.$) 
The ``over" versus ``under" intersection dichotomy 
depends essentially on the sign of the inequality between the arguments of the co-orientation vectors at the intersection point. 
Namely, $\gamma_k$ is ``over" (``under") $\gamma_\ell$ at the first intersection 
if $\theta_{k,a} > \theta_{\ell,a}$ ($\theta_{k,a} < \theta_{\ell,a}$), 
with the same criterion at the second intersection.  
Between the  
first and the second intersection, 
there exist points $\theta_k \in \Theta_k$, $\theta_\ell \in \Theta_\ell$ 
where the co-orientation vectors of $\gamma_k$ and $\gamma_\ell$,   
become aligned with the common normal, that is, $\theta_k=\theta_\ell$.  
At that point, the inequality is changed and the ``over" (``under") at the first intersection point 
becomes ``under" (``over") at the second intersection point. 
\end{proof}

\begin{corollary}\label{c:swallow_tail}
Consider two $\theta$-parameterized critical value curves $\gamma_k$, $\gamma_\ell$, 
restricted to $\Theta_k \subset [0,2\pi)$, $\Theta_\ell \subset [0,2\pi)$, 
where $\gamma_k(\Theta_k)$ is smooth and $\gamma_\ell(\Theta_\ell)$ 
contains a swallow tail (that is, a crossing and two cusps) outside the convex hull of $\gamma_k(\Theta_k)$, 
with the two curves crossing at  $a=\gamma_k(\theta_{k,a})=\gamma_\ell(\theta_{\ell,a})$ and 
$b=\gamma_k(\theta_{k,b})=\gamma_\ell(\theta_{\ell,b})$. 
Then  
one curve, say $\gamma_k$, will intersect the other $\gamma_\ell$ ``over" (``under") at the point $a$ 
and then reintersect $\gamma_\ell$ ``under" (``over") at the point $b$ 
with consistent signs of crossings, as defined in Definition~\ref{d:writhe}. 
Hence, in the contact space, $\mathrm{lk}(\bar{\gamma}_k,\bar{\gamma}_\ell)=\pm 1$, 
that is, the two Legendrian curves are linked.  
\end{corollary}

\begin{proof}
By a Legendrian Reidemeister  move, the swallow tail can be metamorphosed to a smooth curve. 
The result then follows from Th.~\ref{t:smooth_curve}. 
\end{proof}

\section{Hamming weight plus barrier: \\emergence of swallow tail}
\label{s:Hamming_plus_weight}

One of the landmark problems amenable to adiabatic quantum  computation is the Quadratic Binary Optimization (QUBO) problem 
$\min_{x\in\{0,1\}^n} f(x)$, where $f(x)=\sum_{i=1}^n h_ix_i+\sum_{i,j} J_{ij}x_ix_j$. 
The $J_{ij}$'s define a graph 
$(\mathcal{V},\mathcal{E})$ on $n$ vertices ($n$ spins) with the edges defined as $ij \in \mathcal{E}$ if and only if $J_{ij} \ne 0$, 
with the hope that this graph is (minor) embeddable in, say, the Chimera architecture~\cite{Chi_embedding}. 
The quantum adiabatic solution proceeds from the encoding of $x_i\in \{0,1\}$ in states of excitation, 
$|0_i\>=\left(\begin{array}{c}
1 \\
0
\end{array}\right)$ and 
$|1_i\>=\left(\begin{array}{c}
0 \\
1
\end{array}\right)$
with reference to the ``Pauli" operator $s^z=\left(\begin{array}{cc}
0 & 0 \\
0 & 1
\end{array}\right)=\frac{1}{2}(I-\sigma^z)$, where $\sigma^z$ is the usual Pauli operator. 
With the latter conventions, we construct a Hermitian matrix $H_f$, realizable as the Hamiltonian of an Ising network with couplings $J_{ij}$ and localized bias fields $h_i$. 
To set up an adiabatic process that reaches 
the ground state of $H_1:=H_f$, we introduce a transverse field Hamiltonian $H_0:=\sum_{i=1}^n S^x_i$, 
where $S^x_{i}=I^{\otimes(i-1)}\otimes s^x \otimes I^{\otimes (n-i)}$, 
where $s^x$ is the ``Pauli" operator $s^x=\tfrac{1}{2}\left(\begin{array}{cc}
1  & -1\\
-1 & 1
\end{array}\right)=\tfrac{1}{2}(I-\sigma^x)$, where $\sigma^x$ is the usual Pauli operator. 

The following result indicates the energy level of any combination of ``spins up" and ``spins down" is constant at given $n$:
\begin{proposition}
For any $n$-compounded tensor product combination of $\ket{0_i}$ and $\ket{1_i}$, the energy level, 
denoted $E(H_0)$ equals $n/2$. 
\end{proposition}
\begin{proof}
It is easy to prove that in the ``up"-``down" sequence a flip $\ket{\cdots \uparrow  \cdots }$ to 
$\ket{\cdots \downarrow  \cdots }$ does not change the energy. Therefore, with a finite number of such flips, 
we can go from any $n$-compounded tensor product to any other. To prove that $E(H_0)=n/2$, it suffices to check it 
for an arbitrary $n$-compounded tensor product.
\end{proof}

\subsection{Hamming weight plus barrier QUBO}

Specifically here, the objective function is defined as the ``Hamming weight plus barrier," $f(x)=w(x)+p(w(x))$, 
where $w(x)=\sum_{i=1}^n x_i$ is the Hamming weight and $p(w(x))$ is a rectangular-shaped barrier depending 
only on the Hamming weight:  
\begin{eqnarray*}
p(w(x)) &=& h, \quad \mbox{if    } \ell \leq w(x) \leq u, \\
&=&0, \quad \mbox{otherwise}.
\end{eqnarray*}
The $H_f$-Hamiltonian~\cite{reichardt-adiabatic} 
is defined as the Hermitian matrix having eigenvalue $f(x)$ associated with eigenvector $|x\>$. 
The plot of the function $f(x(w))$ versus the Hamming weight $w$ is shown in the black solid lines of  Fig.~\ref{f:Hamming_plus_barrier_geometry}. 
Since $H_{f}|x\rangle =f(x)|x\rangle$, this plot can be interpreted as the plot of eigenvalues 
(energy levels) of $H_f$ 
as a function of the Hamming weight. 
More physically speaking, when $w=0$, all spins are ``down," in the state  $|0_i\>$; 
as $w$ increases, more and more spins flip to the ``up" state $|1_i\>$.

\begin{figure}
\begin{center}
\scalebox{0.6}{\includegraphics{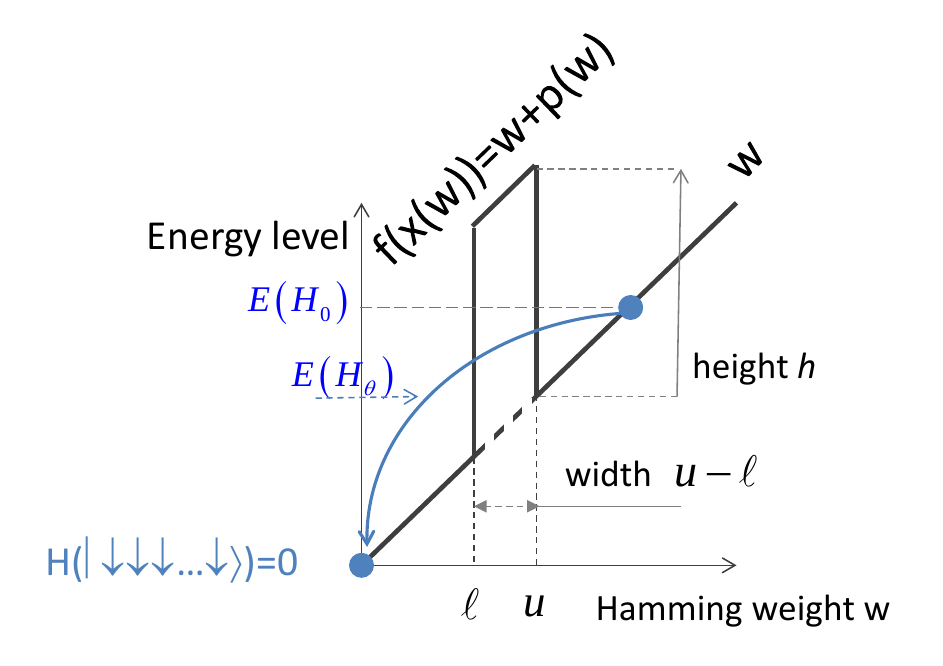}}
\end{center}
\caption{The ``Hamming weight plus barrier" energy level versus Hamming weight. 
The blue dotted cyclic arrow represents the ground state of an adiabatic process 
barely penetrating the barrier.}
\label{f:Hamming_plus_barrier_geometry}
\end{figure}
Our  objective is to show that by varying the width $u-\ell$ and the height $h$ of the barrier at fixed transversed field $\sum_i S_i^x$, 
we can develop a portfolio of illustrative examples from the constant gap to the steep gap case. 
Furthermore, it will serve to illustrate the connection with swallow tails and tunneling. 

We now proceed to construct the Hamiltonian $H_f$. 
The Hamiltonian associated with the Hamming weight is easily found to be 
\[ H_w=\sum_{i=1}^n S^z_i , \quad S^z_i=I^{\otimes(i-1)}\otimes s^z \otimes I^{\otimes(n-i)}. \] 
It is indeed easily see that $H_w|x_1,...,x_n\>=w(x)|x_1,...,x_n\>$. 

The Hamiltonian associated with $p(w)$ is formally defined as $p(H_w)$. 
Practically, 
we construct a polynomial $\tilde{p}$ that approximately agrees with $p$ on the spectrum of $H_w$, 
that is, $\sup_{i}|\tilde{p}(\lambda_i(H_w))- p(\lambda_i(H_w))|\leq \epsilon$, from which there holds 
the spectral norm bound $\|p(H_w) -\tilde{p}(H_w)\| \leq \epsilon$. 
For example, a simple approximation $\tilde{p}$ could be 
\[ \tilde{p}(w)= h\frac{w(w-1)(w-2)...(w-\ell+1)(w-u-1)...(w-n)}{m(m-1)(m-3)...(m-\ell+1)(m-u-1)...(m-n)}, \]
where $m=(\ell+u)/2$. Therefore, 
\begin{eqnarray*}
\lefteqn{H_{p(w)}\approx}\\
&&h\frac{H_w(H_w-I)(H_w-2I)...(H_w-(\ell-1)I)(H_w-(u+1)I)...(H_w-nI)}
{m(m-1)(m-3)...(m-\ell+1)(m-u-1)...(m-n)},   
\end{eqnarray*}
with spectral norm error bound $\epsilon=\sup_{\ell\leq w \leq u}|\tilde{p}(w)-h|$. 
Invoking the Lagrange interpolation, there exists a sequence of approximating polynomials 
such that $|\tilde{p} - p| \to 0$ 
uniformly on the compact set $[0,n]$; 
it then follows that $\|\tilde{p}(H_w) - p(H_w)\|\to 0$, uniformly over the compact set.  
 
Another method proceeds from the observation that, for $w=0,...,n$,  we have \linebreak
$p(w)=(1/2)\left(\mbox{sign}(w-(\ell-1/2))-\mbox{sign}(w-(u+1/2))\right)$, 
where $\mbox{sign}(0)=0$ and $\mbox{sign}(\mathbb{R}_\pm)=\pm 1$. 
Therefore, 
\[ p(H_w)= (1/2)\left(\mbox{sign}(H_w-(\ell-1/2)I)-\mbox{sign}(w-(u+1/2)I)\right),\]
where the sign of some Hermitian matrix $Z=V\Lambda V^*$ is defined 
as $\mbox{sign}(Z)=V\mbox{sign}(\Lambda)V^*$ with the convention that 
$\mbox{sign}(\Lambda)=\mbox{diag}\{\mbox{sign}(\lambda_i(Z))\}_{i=1}^n$.   
The sign of a general square matrix $Z$ can be efficiently computed via the matrix sign function iteration~\cite{KenneyLaubJonckheere}. 
Observe, however, that the sign iteration provides a rational, instead of a polynomial approximation. 

At this stage, we have 
\[ H_1=H_f=H_w+H_{p(w)}. \]
Regarding the transverse field Hamiltonian, we take
\[H_0=\sum_{i=1}^n S^x_i\]
noting that $[H_0,H_1]\ne 0$. 

\subsubsection{genericity}

Such a highly structured Hamiltonian as $H_w$ creates unstable singularities, 
that is, singularities that are non generic, that disappear under a vanishingly small perturbation.  
Such singularities exist in theory but are extremely difficult to exhibit in the realm of numerical computations, 
let alone in the realm of experiments. Precisely, 
\begin{theorem}
The field of values $\mathcal{F}(H_w)$ is the closed interval $[0,1]$, with the boundary points $0$ and $1$ eigenvalues 
of $H_w$. Consequently, $\mathcal{F}(H_w)$ is non generic. 
\end{theorem}
\begin{proof}
The field of values of a Hermitian matrix, $H_w$, is the convex hull of its eigenvalues and is hence embedded in $\mathbb{R}$. 
Moreover, it is easily seen that $0$ and $1$ are the extreme eigenvalues with eigenvectors  
$(|0\>)^{\otimes n}$ and  $(|1\>)^{\otimes n}$, resp. 
Such eigenvalues are rank $0$ critical values~\cite[Th. 2]{JonckheereAhmadGutkin} of the quadratic numerical range map 
and a matrix with such eigenvalues is nongeneric~\cite[Cor. 3]{JonckheereAhmadGutkin}. 
\end{proof}

In order to remove those unstable singularities, 
we introduce a perturbation relevant to an experimental environment, like rogue magnetic fields. 
Since the Pauli operators $S_i^z$ and $S_i^x$ are already utilized in the Ising and transverse fields, resp., 
we introduce the 
$y$-Pauli operators $S_i^y= I^{\otimes(i-1)}\otimes s^y \otimes I^{\otimes (n-i)}$, where 
$s^y=(1/2)(I-\sigma^y)=(1/2)\left(\begin{array}{cc}
1 & \imath\\
-\imath & 1
\end{array}\right)$, and $\sigma^y$ is the usual $y$-Pauli operator. 
The problem Hamiltonian is hence finalized with a small dis-singularizing $y$-field perturbation
\[  \tilde{H_1}=H_w+H_{p(w)} + \epsilon \sum_{i=1}^n r_i S^y_i,\]
where the $r_i$'s are uniformly distributed over $[-1,+1]$. 

\subsubsection{Tunneling}

The energy landscape of $H_1$ is shown by the (solid black) plot of Fig.~\ref{f:Hamming_plus_barrier_geometry}. 
The QUBO objective is to get to the minimum energy point $(w,E)=(0,0)$. 
The fact that the eigenenergy levels $E(H_0)$ do not depend on $w$ for fixed $n$  
justifies drawing the (dashed blue) horizontal line at that level (shown in blue).  
This line will intersect the $H_1$ energy plot at a certain point, $(w_0,E(H_0))$;   
$\ket{w_0}$ can be interpreted as the ``spin up, spin down" state the closest to the eigenenergy level of $H_0$. 
Therefore, the process (shown in blue) goes from $(w_0,E(H_0))$ to $(0,0)$ and will cross the barrier 
if $w_0> U$ and $h$ large enough. 

\subsection{Numerical results}
\label{s:constant_to_steep}

At constant transverse field Hamiltonian $H_0$, 
there are many parameters that can be manipulated  
in the Hamming weight plus barrier Hamiltonian $H_1$ 
in order to exhibit various phenomena.  
The parameters that will be manipulated are $\ell$, $u$, and $h$. 
While $\epsilon$ can be manipulated as well, 
contrary to other glaring cases like the Grover search, 
it does not incur much changes in the topology of this specific problem. 
The two main phenomena to be identified are (i) the emergence of swallow tails and (ii) tunneling. 
\begin{enumerate}
\item \textbf{Swallow tails:} it is shown that the higher the barrier $h$ the worse the gap, 
the closer the swallow tails of $\gamma_2$ are to the $\gamma_1$ curve. 
\item \textbf{Tunneling:} 
In Figs.~\ref{f:low_barrier_small_y}-\ref{f:barrier_low_Hamming}, 
$2.5=n/2>u=2$ so that the conditions for tunneling are met; 
however, in Fig.~\ref{f:low_barrier_small_y}, the barrier $h=0.95$ is a bit low to positively guarantee tunneling;
on the other hand, in Figs.~\ref{f:high_barrier_small_y}-\ref{f:barrier_low_Hamming}, the barrier is higher  
making the case for tunneling stronger. 
Finally, the opposite happens in Fig.~\ref{f:barrier_high_Hamming} where $2.5=n/2<u=2$, 
ruling out tunneling, with a symptomatic disappearance of the swallow tails. 

\end{enumerate}

\subsubsection{Constant gap case: $h=0$}
\label{s:constant_gap}

\begin{figure}[t]
\begin{center}
\scalebox{0.9}{\includegraphics{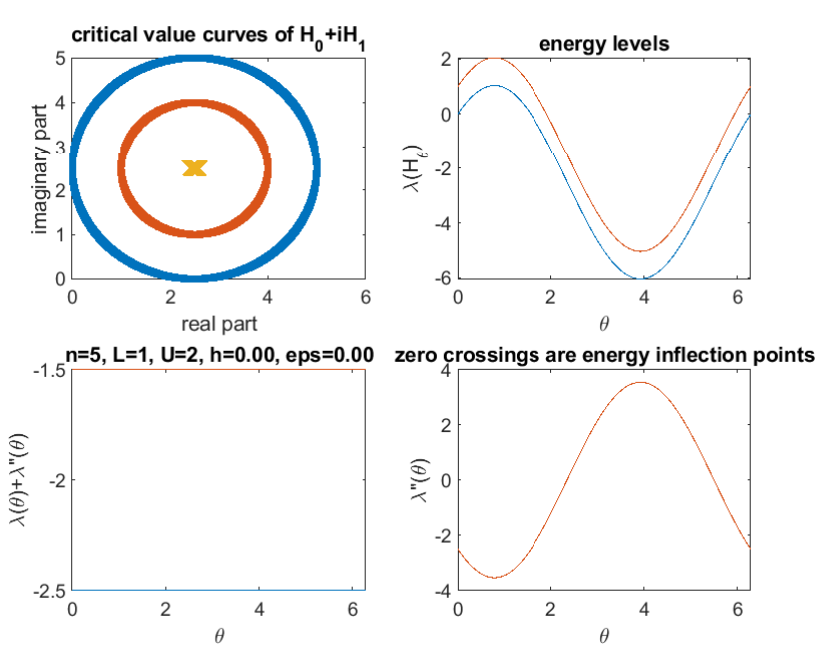}}
\end{center}
\vskip-0.9cm
\caption{\textbf{No barrier, no tunneling}: $n=5$, $\ell=1$, $u=2$, {\boldmath$h=0$}, $\delta \theta=0.001$, $\epsilon=0$ for $\theta \in [0,2\pi)$. 
Observe the big gap area between the boundary of the field of values 
and the degenerate first excitation level curve on the top left panel, 
consistently with the eigenvalue plots on the top right panel. 
The cross denotes the multiple eigenvalue of $H_0+\imath H_1$ at $0.25+\imath 0.25$.}
\label{f:first_case_no_barrier_no_y}
\end{figure}

The Hamming weight case---without barrier--- is the perfect example 
of a constant gap case, that is, (i) the energy difference between the ground state and first excited state remains constant along the homotopy, 
and (ii) the critical value curves of the ground state and the first excited state are concentric circles, with constant distance between them, 
and without swallow tails. Precisely, 
\begin{theorem}
\label{t:concentric_circles}
Consider the adiabatic evolution of $n$ spins under trigonometric scheduling evolving from the usual transverse field $H_0$ to 
the Hamming weight Hamiltonian $H_w=H_1$. Then (i) the energy levels curves are in phase cosines 
of constant energy differences and (ii) the critical value curves are concentric circles in the field of values with their center at the $2^n$-fold eigenvalue 
$n(1+\imath)/2$ of $H_0+\imath H_1$. 
\end{theorem}
\begin{proof}
Recall that in the ``no barrier" case, 
both $H_0$ and $H_1$ are of the form $\sum_{i=1}^n I^{\otimes(i-1)}\otimes A \otimes I^{\otimes (n-i)}$, 
where $A=s^x$ for $H_0$ and $A=s^z$ for $H_1$. Therefore, 
$H=\sum_{i=1}^n I^{\otimes(i-1)}\otimes (s^x+\imath s^z) \otimes I^{\otimes (n-i)}$ 
and $H_\theta=\sum_{i=1}^n I^{\otimes(i-1)}\otimes (s^x\cos \theta +s^z \sin \theta)\otimes I^{\otimes (n-i)}$. 
It is easily seen that the eigenvalues of $s^x \cos \theta + s^z \sin \theta$ are 
$\lambda_{\pm}(\theta)=(1/2)(\cos \theta + \sin \theta \pm 1)$ with eigenvectors 
$v_{\pm}(\theta)=\left(\begin{array}{c}
\mp \cos \theta \\
1 \pm \sin \theta 
\end{array}\right)\frac{1}{(2(1\pm\sin \theta))^{1/2}}$. It follows that 
the eigenvectors of $H_\theta$ are $\otimes^n v_{\pm}(\theta)$ and the eigenvalues of $H_\theta$, the energy levels,  
are $\sum^n \lambda_{\pm}(\theta)$. Therefore, the energy levels are 
\[ \frac{1}{2}\left( n(\cos \theta + \sin \theta) +(n_+-n_-) \right) \]
where $n_+$, $n_-$ are the number of $+$, $-$ signs, resp., picked up in the construction of the energy level. 
This proves (i). 
To prove (ii), it suffices to look at the critical values
\begin{align*}
&\left\<\otimes^n v_{\pm}(\theta)\left|
\sum_{i=1}^n I^{\otimes(i-1)}\otimes (s^x + \imath s^z)\otimes I^{\otimes (n-i)}
\right|\otimes^n v_{\pm}(\theta)\right\>\\
&=\sum_{i=1}^n \<v_{\pm}(\theta)|s^x  + \imath s^z  | v_{\pm}(\theta)\>\\
&=\sum_{i=1}^n \frac{1}{2} \left((1\pm \cos \theta) + \imath (1 \pm \sin \theta) \right) \\
&=\frac{1}{2}( (n +(n_+-n_-)\cos \theta) + \imath ((n_+-n_-)\sin \theta))
\end{align*}

If $x$, $y$ denotes the real and imaginary parts, resp., of the critical value, elimination of $\theta$ yields $(2x-n)^2+(2y-n)^2=(n_+-n_-)^2$, that is, a circle 
with its center at $(n/2)(1+\imath)$ and radius $(n_+-n_-)/2$.  Finally, we need to show that $(n/2)(1+\imath)$ is the $2^n$-fold 
eigenvalue of $H_0+\imath H_1$. To this end, observe that $s^x+\imath s^z$ has a double eigenvalue at $(1+\imath)/2$. Since 
$H_0+\imath H_1=\sum_{i=1}^n I^{\otimes (i-1)}\otimes (s^x+\imath s^z)\otimes I^{\otimes (n-i)}$ 
and since all terms of the sum have common eigenvectors, the result follows from properties of eigenvalues of tensor products.  
\end{proof}

The following corollary is easily proved:
\begin{corollary}
Under the same hypothesis as in preceding theorem, the degeneracy of the energy level $(n_+,n_-)$ is 
$\left(\begin{array}{c}
n\\
n_+
\end{array}\right)$. In particular, the ground level is nondegenerate. 
\end{corollary}

Note that the number of energy levels, including multiplicity, is 
$\sum_{n_+=0}^n \left(\begin{array}{c}
n\\
n_+
\end{array}\right)=2^n$, as expected.

The preceding theorem is illustrated in Fig.~\ref{f:first_case_no_barrier_no_y}. Observe that the center of the concentric circles is at $0.25+\imath 0.25$, 
consistently with the Theorem~\ref{t:concentric_circles} since $n=5$. 
A bit less trivially, the $\gamma_2$ circle has vertical tangency points, requiring their replacement by cusps, 
resulting in a $\hat{\gamma}_2$ of the same canonical structure 
as that of Fig.~\ref{f:breaking_vertical_tangents}(b) with classical invariants $(\mu,{\tt tb})=(0,-1)$.

\subsubsection{Towards steep gap and tunneling: $h$ increases}

Here, we set $\ell=1$, $u=2$, and progressively increase the barrier height $h$  
to expose the connection between the narrowing of the gap, the 
swallow tail, and tunneling.

In Figures~\ref{f:low_barrier_small_y}-\ref{f:barrier_low_Hamming}, 
the important swallow tail is the ``South-West" one, at a co-orientation $\theta\approx \pi/4$,  
since it is captured by the physically relevant path $[0,\pi/2]$ embedded in the complete $[0,2\pi]$ homotopy.   
The whole point is easily seen by comparing Figure~\ref{f:low_barrier_small_y} 
and Figure~\ref{f:high_barrier_small_y}. The deterioration of the gap together with a more pronounced swallow tail as $h$ increases is obvious.  
Note that the two $\hat{\gamma}_2$ curves share the same classical invariants $(\mu,{\tt tb})=(0,-1)$ 
as that of Fig.~\ref{f:first_case_no_barrier_no_y}. 
To see this, it suffices to replace the vertical tangency points by cusps, 
remove the swallow tails by Reidemeister moves, and observe that we are back to 
Fig.~\ref{f:breaking_vertical_tangents}(b).  

The top-left panel of Fig.~\ref{f:barrier_low_Hamming} shows an interruption of the plot 
for the reason that the plot is extremely sensitive around $\theta \approx 0$  
due to a pair of cusps points, and hence a swallow tail, 
revealed by the bottom-left panel. 
Closer inspection reveals that the vertical tangency point is not on the edge of the swallow tail, 
but on the branch departing from the swallow tail to $\theta \approx \pi/2$, 
resulting in the structure of Fig.~\ref{f:breaking_vertical_tangents}(b) 
and classical invariants $(\mu,{\tt tb})(\hat{\gamma}_2)=(0,-1)$.

{\it As shown in Figure~\ref{f:barrier_low_Hamming}, these two phenomena---steep gap and tunneling---are concomitant with a swallow tail. }

\begin{figure}[t]
\begin{center}
\scalebox{0.9}{\includegraphics{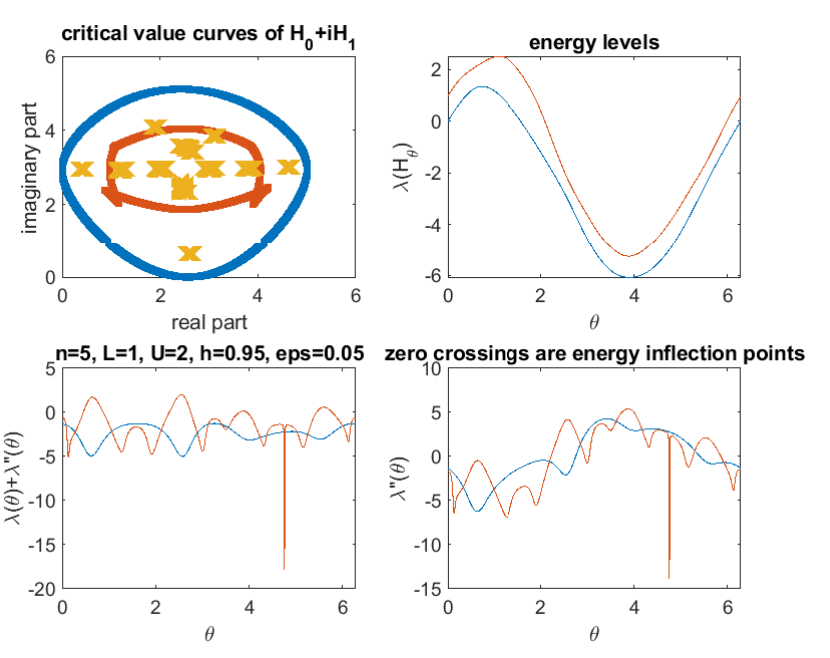}}
\end{center}
\vskip-0.9cm
\caption{{\bf Low barrier but emergence of swallow tails}: $n=5$, $\ell=1$, $u=2$, {\boldmath$h=0.95$}, $\delta \theta=0.001$, $\epsilon=0.05$. 
From the left panel, observe the emergence of the swallow tail at $\theta \approx \pi/4$, 
the ``South-West" swallow tail, with a pair of eigenvalues in the prolongation of the cusps.
From the right panel, 
the swallow tail gap  appears to be between mild and steep, 
as an inflection point in the first excited state just appeared. 
}
\label{f:low_barrier_small_y}
\end{figure}

\begin{figure}[t]
\begin{center}
\scalebox{0.9}{\includegraphics{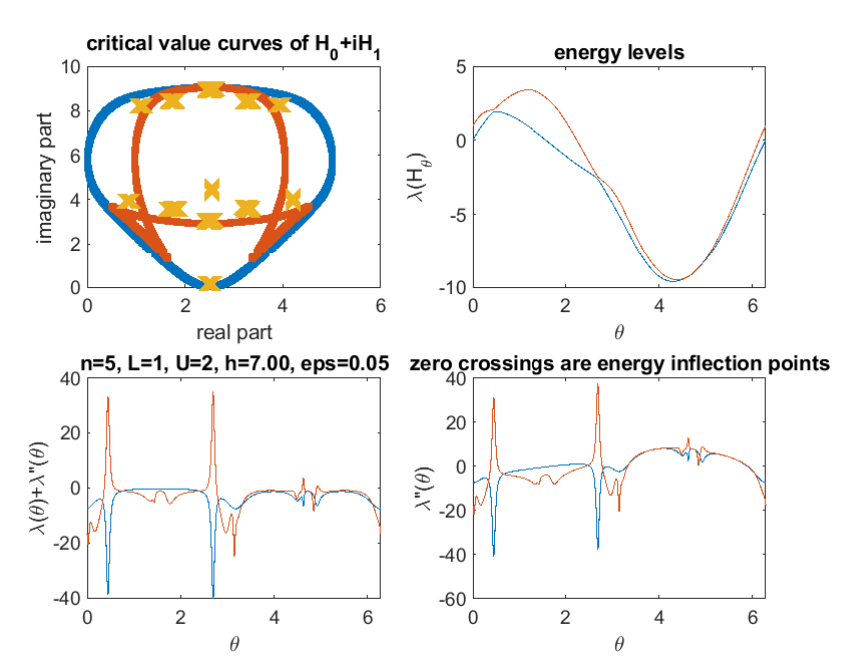}}
\end{center}
\vskip-0.9cm
\caption{{\bf Higher barrier, swallow tails approaching boundary, and tunneling}: $n=5$, $\ell=1$, 
{\boldmath $u=2$}, {\boldmath$h=7$}, $\delta \theta=0.001$, $\epsilon=0.05$, 
with $\theta\in [0,\pi/2]$. 
Observe, on the left panel, 
the ``South-West" swallow tail nearly collapsing on the ground state critical value curve in the reverse of an universal unfolding~\cite{CastrigianoHayes1993}. 
On the right panel, observe the steep gap at $\theta \approx \pi/4$, 
as the first excited state curve has a pair of inflection points, 
but the ground state does not have it.  
(The gap at $\theta \approx 3\pi/4$ is {\bf supersteep} as both levels have pairs of inflection points.)
 }
\label{f:high_barrier_small_y}
\end{figure}

\begin{figure}[t]
\begin{center}
\scalebox{0.9}{\includegraphics{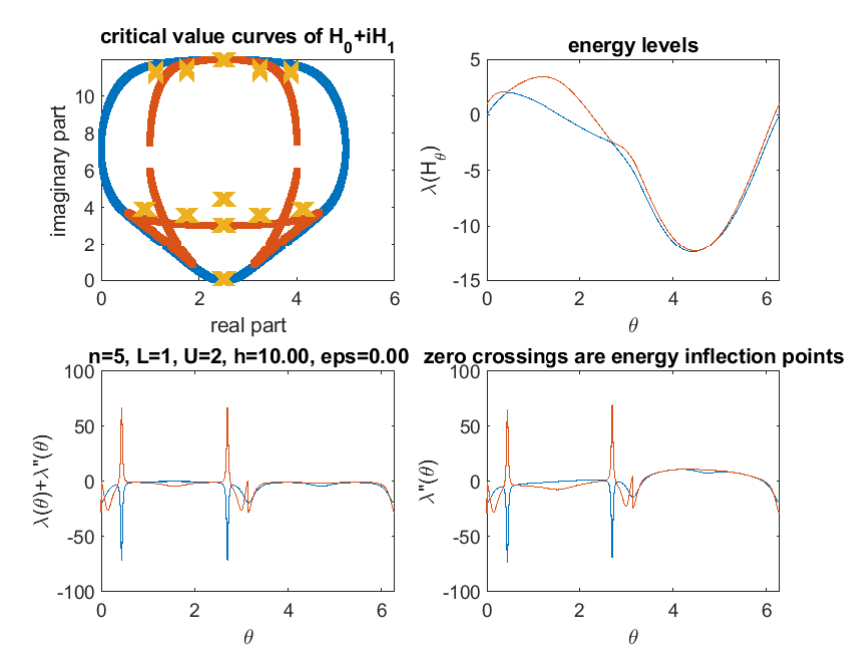}}
\end{center}
\vskip-0.9cm
\caption{{\bf Highest barrier, swallow tails closest to boundary, and tunneling:} 
$n=5$, $\ell=1$, {\boldmath$u=2$}, {\boldmath$h=10$}, $\delta \theta=0.001$, $\epsilon=0$, 
with $\theta\in [0,2\pi)$. 
On the top left panel, observe the ``South-West" swallow tail nearly collapsing on the ground state critical value curve, 
hence closing the gap.
Observe, on the right panel, the {\bf steep gap} at $\theta \approx \pi/4$, 
as the first excited state curve has a pair of inflection points, 
but the ground state does not quite have it.}
\label{f:barrier_low_Hamming}
\end{figure}

\subsubsection{No tunneling, no swallow tail}

If the barrier has its support at high Hamming weight ($p(w) \ne 0$ for $3 \leq w \leq 4$), 
the adiabatic path starts at energy level below the barrier at $3$ 
and, hence, no tunneling is required.  

The latter is concomitant with the absence of swallow tail, 
as seen from Figure~\ref{f:barrier_high_Hamming}, 
top left panel.  
It is also noted on the right panel that the gap is mild. 
After removing the vertical tangents, the classical invariants are $(\mu, {\tt tb})=(0,-1)$. 

\begin{figure}[t]
\begin{center}
\scalebox{0.9}{\includegraphics{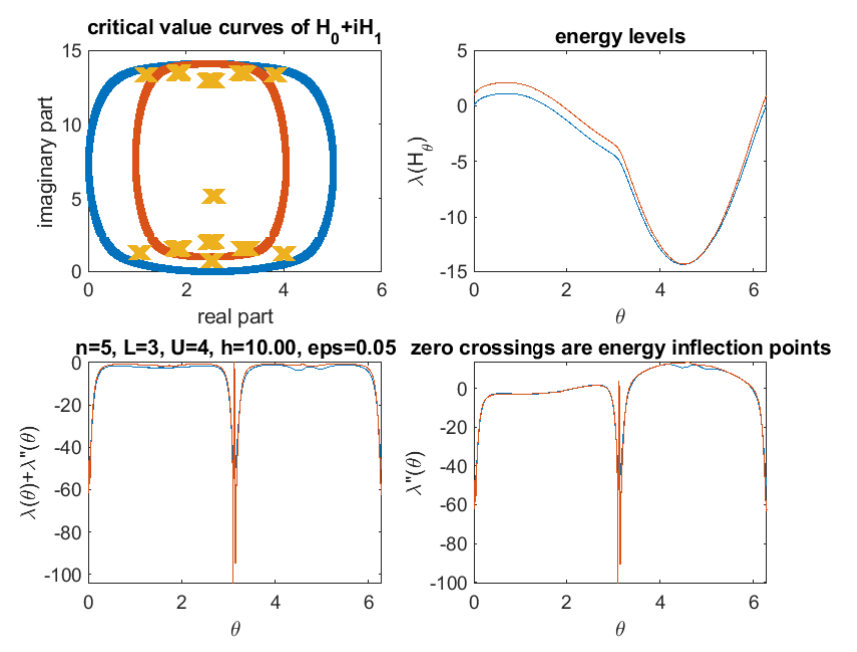}}
\end{center}
\vskip-0.9cm
\caption{{\bf Barrier at too high a Hamming weight, no swallow tails, no tunneling:} 
$n=5$, {\boldmath$\ell=3$, $u=4$}, $h=10$, $\delta \theta=0.001$, $\epsilon=0.05$, 
with $\theta\in [0,2\pi)$. 
On the left panel, observe that the swallow tail has disappeared, hence enlarging the gap.
Observe, on the right panel, the {\bf mild gap,} as neither the ground nor the first excited state has inflection points.}
\label{f:barrier_high_Hamming}
\end{figure}

\section{Conclusion}

We have shown that reformulating the adiabatic gap in the light of 
the critical value curves of the numerical range of the matrix 
made up with the initial and terminal Hamiltonians gives the gap a novel interpretation: 
the distance between the  boundary of the numerical range, $\gamma_1$, 
and its  first excited critical value curve, $\gamma_2$. 
The near coalescence of the two critical value curves, especially around a swallow tail, 
provides  a ``magnifying lens" on the subtle details of the 
anti-crossing phenomenon difficult to visualize on the classical energy level plots. 
 
The visual display of the $\gamma_1$, $\gamma_2$ critical value curves supplements, 
even overpasses in visual acuity, 
the  morphology of the classical eigenenergy $\lambda_1$, $\lambda_2$ plots. 
On the latter, subtle details related to the existence and positioning of the inflection points 
are easy to miss with possible confusion between mild and steep gaps, whereas 
on the $\gamma_1$, $\gamma_2$ curves the difference is striking---(super)steep gaps are flagged by swallow tails,  
in a number anticipated by the number of pairs of ``nearby" roots of topological invariant $\lambda_2+\lambda_2''$.

Beyond visualization, a first practical consequence is in the case of the Grover search, 
where the gap is usually computed around an unstable critical value curve that breaks up into several stable singularities under an arbitrarily small perturbation, 
with the unfortunate consequence that the gap changes abruptly and for the worst. 
It is anticipated that such phenomena will happen in case of highly degenerate 
ground state~\cite{Itay}, 
where an arbitrarily small perturbation will break it into many nondegenerate states near the ground level. 

Probably most importantly, we showed on the selected ``Hamming weight plus barrier" terminal Hamiltonian 
that the swallow tails reveal that the adiabatic process tunnels through the barrier. 
Experimental confirmation of this feature could give, next to entanglement, 
further evidence of quantum phenomena in adiabatic quantum computers. 

On the theoretical side, we have adapted the Legendrian theory to the classification of the critical value curves, 
with the major difficulty that the $\gamma_2$ curves have vertical tangents, 
which is disallowed in the Legendrian theory. 
We have circumvented this difficulty by the new technique of replacing the vertical tangency points by cusps and defining the Thurston-Bennequin invariant on the resulting curve. 
It was shown that this new invariant classifies curves in three groups while the Arnold invariants 
does not ``see" this distinction. Certainly, 
the former is  still a coarse classification---it does not accurately count the number of swallow tails, 
but this can be remedied with the $\lambda_2+\lambda''_2$ invariant.  

Generically $\lambda_1+\lambda_1''$ never vanishes, but $\lambda_i+\lambda''_i$, $i>1$, does vanish.  
If  $\lambda_{i}+\lambda_{i}''$ and $\lambda_{i+1}+\lambda_{i+1}''$, $i=2,3,...$, vanish 
for nearby $\theta$-parameter values, 
then the $\lambda_i$ curve will sharply increase to the $\lambda_{i+1}$ curve 
while the latter will sharply decrease to $\lambda_i$. 
Hence, ``super-steep" gaps are likely to occur between higher excitation levels, 
which is relevant to the diabatic continuation and shortcuts to adiabaticity. 
This is left for further study.

Among the notable omissions here are the ``new'' invariants of Chekanov-Pushkar and Chekanov-Eliashberg. 
It is hoped that, along with the invariants already developed here, they will provide more accurate classification, 
but this is left for further study. 
The polynomial criterion for stability of the singularities is also left for further research 
as a follow up of the Gr\"obner basis approach developed for singularities in robust control~\cite{Nanaz_Alaska}.

\section*{Acknowledgment} 
Helpful discussions on this subject with Prof. Francis Bonahon are gratefully acknowledged. 

\appendix

\section{Differential topology of numerical range}
\label{a:review}

Here we review the fundamental facts of~\cite{JonckheereAhmadGutkin} that are put to use in the main body. The notation is the same as in the main body of the paper. 

\begin{proposition}
\label{p:smoothness_boundary_review}
Let $\mathcal{F}(H)$ be the field of values of $H \in \mathbb{C}^{n \times n}$. 
The boundary $\partial \mathcal{F}(H)$ is generically smooth. 
The nongeneric boundary features are either sharp points where the boundary is not differentiable or line segments embedded in the boundary. 
For any $\zeta \in \partial \mathcal{F}$ 
and any $z \in f^{-1}(\zeta)$, the differential $d_{z}f$ is rank deficient. Moreover, as far as the numerical range is concerned, 
\begin{enumerate}
\item If $\zeta$ is a smooth boundary point, $\mbox{rank} \left( d_{z}f \right)=1$.
\item If $\zeta$ is a sharp point, $\mbox{rank} \left( d_{z}f \right)=0$. 
Moreover, any such sharp point is an eigenvalue of $H$ and $z$ is an eigenvector of both $H$ and $H^*$. 
Conversely, if all preimage points of $\zeta \in \partial \mathcal{F}$ are eigenvectors of both $H$ and $H^*$, 
then $\zeta$ is sharp.  
\end{enumerate}
As far as the eigenvalue plots are concerned,
\begin{enumerate}
\item $\rank (d_{z} f) =1$ if and only if $z$ is a critical point of $\varphi_\theta$ for a unique $\theta \in [0,\pi)$. 
\item $\rank (d_{z} f) =0$ if and only if $z$ is a critical point of $\varphi_\theta$, $\forall \theta \in [0,\pi)$. 
\end{enumerate}
\end{proposition}

\section{Differentiable class of eigenvalues}
\label{a:differentiability}

As in the main body, we consider  $\mathrm{Herm}(N)$, the set of $N \times N$ Hermitian matrices, 
but here, as in~\cite{GutkinJonckheereKarow}, we emphasize the (nongeneric) case  
of possibly multiple eigenvalues  
ordered as $\lambda_{[1]}<\lambda_{[2]}<...$, 
where $\lambda_{[k]}$ denotes the multiple eigenvalue
\[ \lambda_{[k]}:=
\lambda_{\mu_1+...+\mu_{k-1}+1}=\lambda_{\mu_1+...+\mu_{k-1}+2}=...=\lambda_{\mu_1+...+\mu_{k-1}+\mu_k}\]
with multiplicity $\mu_k$. $E_{[k]}$ denotes the eigenspace of $\lambda_{[k]}$. 
Here, as already alluded to in the von Neumann-Wigner adiabatic theorem (Th.~\ref{t:von_neumann-wigner}), 
we consider a multi-parameter path 
in a manifold $\Theta$.
Given a differentiable map $F:\Theta \to \mathcal{N}$ between differentiable manifolds $\Theta$ and $\mathcal{N}$, 
the linear map $d_\theta F: T_\Theta \Theta \to T_{F(\theta)}\mathcal{N}$ denotes the first order differential (Jacobian) evaluated at $\theta \in \Theta$. 
If $\Theta$ and $\mathcal{N}$ are strengthened to $\mathcal{C}^2$-manifolds and $F$ to a $\mathcal{C}^2$-map,  
the bilinear map $d^2_\theta F : T_\theta \Theta \times T_\theta \Theta \to T_{F(\theta)}\mathcal{N}$ 
denotes the second-order differential (Hessian).  
The main result of this Appendix is the following:

\begin{proposition}
Let $\Theta$ be a $\mathcal{C}^2$-manifold and let $H: \Theta \to \mathrm{Herm}(N)$ be a $\mathcal{C}^2$-map. 
Then, for any unit vector $|z_{[k]}(\theta)\> \in E_{[k]}(H(\mathbf{\theta}))$, 
and any $u,v \in T_\theta\Theta$, we have
\begin{eqnarray*}
\lefteqn{d^2_\theta (\lambda_{[k]} \circ H)(v,w)=}\\
&&2\<z_{[k]}(\theta)| d_\theta H(u)(\lambda_{[k]}(H(\theta))I-H(\theta))^\dagger d_\theta H(w)|z_{[k]}(\theta)\>\\
&&+\<z_{[k]}(\theta)|d^2_\theta H(v,w)|z_{[k]}(\theta) \>. 
\end{eqnarray*}
\end{proposition}

\begin{proof}
See~\cite[Th. 3.7(.3)]{GutkinJonckheereKarow}
\end{proof}

Now if we set $\Theta=\mathbb{S}^1$ as in the main body, $v$ and $w$ become $1$-dimensional, hence scalars.   
Obviously, for a generic $\mathcal{C}^2$-map $F$, 
\begin{eqnarray*}
d_\theta F(v)&=& F'(\theta)v,\\
d_\theta^2 F(v,w) &=&d_\theta \left(d_\theta F(v) \right)(w)=d_\theta (F'(\theta)v)(w)=F''(\theta)vw.
\end{eqnarray*}
After simplifying $v$ and $uv$ across the equalities, the Proposition yields
\begin{eqnarray*}
\lefteqn{\lambda''(H(\theta))=}\\
&&2\<z_{[k]}(\theta)| H'(\theta)(\lambda_{[k]}(H(\theta))I-H(\theta))^\dagger H'(\theta)|z_{[k]}(\theta)\>\\
&&+\<z_{[k]}(\theta)|H''(\theta)|z_{[k]}(\theta) \>. 
\end{eqnarray*}
Finally, if $H(\theta)=H_0\cos \theta + H_1 \sin \theta$ as in the main body, $H''(\theta)=-H(\theta)$ 
and the result of Lemma~\ref{l:karow} follows.

\bibliographystyle{plain}

\begin{thebibliography}{10}

\bibitem{Aicardi}
F.~Aicardi.
\newblock Discriminants and local invariants of planar fronts.
\newblock In V.~I Arnold, I.~M. Gelfand, M.~Smirnov, and V.~S. Retakh, editors,
  {\em Arnold-Gelfand Mathematical Seminars---Geometry and Singularity Theory},
  pages 1--76. Birkh\"auser, Boston, 1997.

\bibitem{Alexander1978}
J.C. Alexander.
\newblock The topological theory of an embedding method.
\newblock In H.~Wacker, editor, {\em Continuation Methods}, pages 37--68.
  Academic Press, 1978.

\bibitem{Arnold1994}
V.~Arnold.
\newblock {\em Topological Invariants of Plane Curves and Caustics}, volume~5
  of {\em University Lecture Series; Dean Jacqueline B. Lewis Memorial
  Lectures, Rutgers University}.
\newblock American Mathematical Society, Providence, RI, 1994.

\bibitem{Arnold1993}
V.~I. Arnold.
\newblock {\em The Theory of Singularities and its Applications}.
\newblock Accademia Nationale Dei Lincei; Secuola Normale Superiore Lezioni
  Fermiane. Press Syndicate of the University of Cambridge, Pisa, Italy, 1993.

\bibitem{Arnold_on_Tabachnikov}
V.~I. Arnold.
\newblock Topological problems in wave propagation theory and topological
  economy principle in algebraic geometry.
\newblock {\em Fields Institute Communications}, 24:39--54, 1999.

\bibitem{ArnoldGusein-ZadeVarchenko1985}
V.~I. Arnold, S.~M. Gusein-Zade, and A.~N. Varchenko.
\newblock {\em Singularities of Differentiable Maps---The Classification of
  Critical Points, Caustics and Wave Fronts}, volume~1.
\newblock Birkh\"auser, Boston, 1985.

\bibitem{ArnoldGusein-ZadeVarchenko1988}
V.~I. Arnold, S.~M. Gusein-Zade, and A.~N. Varchenko.
\newblock {\em Singularities of Differentiable Maps ---Monodromy and
  Asymptotics of Integrals}, volume~2.
\newblock Birkh\"auser, Boston, 1988.

\bibitem{Belevitch}
V.~Belevitch.
\newblock {\em Classical Network Theory}.
\newblock Holden-Day, San Francisco, 1968.

\bibitem{Applications_of_envelopes}
K.~Bickel, P.~Gorkin, and T.~Tran.
\newblock Applications of envelopes.
\newblock arXiv:1810.11678v1 [math.FA] 27 Oct 2018, 2028.

\bibitem{envelope_and_cusp}
G.~Capitano.
\newblock Singularities of the envelope of curves tangent to a semi-cubic cusp.
\newblock {\em Journal of Mathematical Sciences}, 126(4):1243--1250, 2005.

\bibitem{Lagrangian_skeleta}
R.~Casals.
\newblock Lagrangian skeleta and plane curve singularities.
\newblock {\em J. Fixed Point Theory Appl.}, 24(34), 2022.
\newblock Available at https://doi.org/10.1007/s11784-022-00939-8.

\bibitem{CastrigianoHayes1993}
D.~P.~L. Castrigiano and S.~A. Hayes.
\newblock {\em Catastrophe Theory}.
\newblock Addison-Wesley, Reading, MA, 1993.

\bibitem{Cerf1970}
J.~Cerf.
\newblock La stratification naturelle des espaces de fonctions
  diff\'erentiables r\'eelles et le th\'eor\`eme de la pseudo-isotopie.
\newblock {\em Publications Math\'ematiques, Institut des Hautes Etudes
  Scientifiques (I.H.E.S.)}, 39:5--173, 1970.

\bibitem{fronts_of_Legendrian_links}
Yu.~V. Chekanov and P.~E. Pushkar.
\newblock Combinatorics of fronts of {L}egendrian links and the {A}rnold's
  4-conjecture.
\newblock {\em Russian Math. Surveys}, 60(1):95--149, 2005.

\bibitem{Chekanov_original}
Yuri Chekanov.
\newblock Differential algebra of legendrian links.
\newblock {\em Inventiones Mathematicae}, 150:441--483, 12 2002.

\bibitem{arnold_birthday}
S.~Chmutov and V.~Goryunov.
\newblock Polynomial invariants of {L}egendrian links and plane fronts.
\newblock {\em Amer. Math. Soc. Transl.}, 180(2):25--43, 1997.
\newblock To V. I. Arnold on the occasion of his 60th birthday.

\bibitem{diabatic}
E.~J. Crosson and D.~A. Lidar.
\newblock Prospects for quantum enhancement with diabatic quantum annealing.
\newblock ar{X}iv:2008.09913v1 [quant-ph] 22 Aug 2022, 2020.

\bibitem{Leray_Maslov}
M.~de~Gossom.
\newblock On the usefulness of an index due to {L}eray for satudying the
  intersections of {L}agrangian and symplectic paths.
\newblock {\em J. Math. Pures Appl.}, 91:598--613, 2009.

\bibitem{Eliashberg_original}
Yakov Eliashberg.
\newblock Invariants in contact topology.
\newblock {\em Doc. Math.}, Extra Vol.:327--338, 1998.

\bibitem{analytic_spectral_factorization}
Lasha Ephremidze, Gigla Janashia, and Edem Lagvilava.
\newblock An analytic proof of the matrix spectral factorization theorem.
\newblock {\em Georgian Mathematical Journal}, 15(2):241--249, 2008.

\bibitem{Chekanov_Eliashberg_invariants}
J.~Epstein, D.~Fuchs, and M.~Meyer.
\newblock Chekanov-{E}liashberg invariants and transverse approximations of
  {L}egendrian knots.
\newblock {\em Pacific Journal of Mathematics}, 201(1):89--106, November 2001.

\bibitem{contact_geometry_another_summary}
J.~B. Etnyre.
\newblock Introductory lectures on contact geometry.
\newblock In G.~Mati\'c and C.~McCrory, editors, {\em Topology and Geometry of
  Manifolds}, volume~71 of {\em Proceedings of Symposia in Pure Mathematics},
  pages 70--81, Athens, GA, 2007. American Mathematical Society.
\newblock Available at {\tt arXiv:math/0111118v1 [math.SG] 9 Nov 2001}.

\bibitem{Brouwer_domain_invariance}
N.~Fathpour and E.~A. Jonckheere.
\newblock A {B}rouwer domain invariance approach to boundary behavior of
  nyquist maps.
\newblock {\em Mathematics of Controls, Signals and Systems (MCSS)},
  11(04):357--371, 1998.

\bibitem{Nanaz_SanDiego}
N.~Fathpour and E.~A. Jonckheere.
\newblock Structural stability and infinitesimal {V}-stability for the
  {R}iccati equation.
\newblock In {\em Proceedings of the American Control Conference (ACC'99)},
  pages 2340--2344, San Diego, CA, June 02-04 1999.
\newblock Session TM14-1.

\bibitem{C-E_invariants_knots}
Dmitry Fuchs.
\newblock {Chekanov-Eliashberg} invariant of {L}egendrian knots: existence of
  augmentations.
\newblock {\em Journal of Geometry and Physics}, 47(1):43--65, 2003.

\bibitem{Invariants_Legendrian_transverse_knots}
Dmitry Fuchs and Serge Tabachnikov.
\newblock Invariants of legendrian and transverse knots in the standard contact
  space.
\newblock {\em Topology}, 36(5):1025--1053, 1997.

\bibitem{exceptional_points_close}
A~Galda and V.~M. Vinokur.
\newblock Exceptional points in classical spin dynamics.
\newblock {\em Sci Rep}, 9:17484, 2019.
\newblock Available at https://doi.org/10.1038/s41598-019-53455-0.

\bibitem{GolubitskyGuillemin1973}
M.~Golubitsky and V.~Guillemin.
\newblock {\em Stable Mappings and Their Singularities}, volume~14 of {\em
  Graduate Texts in Mathematics}.
\newblock Springer-Verlag, New York, 1973.

\bibitem{wave_front_Legendrian_knots}
V.~V. Goryunov.
\newblock Plane curves, wave fronts and {L}egendrian knots.
\newblock {\em Philosophical Transactions: Mathematical, Physical and
  Engineering Sciences}, 359(1784):1497–510, 2001.
\newblock JSTOR, http://www.jstor.org/stable/3066350.

\bibitem{Dolezal_vector_bundle}
K.~A. Grasse.
\newblock A vector-bundle version of {D}olezal's theorem.
\newblock {\em Linear Algebra and its Applications}, 392:45--59, 2004.

\bibitem{shortcut_to_adiabaticity}
D.~Gu\'ery-Odelin, A.~Ruschhaupt, A.~Kiely, E.~Torrontegui,
  S.~Mart\'{\i}nez-Garaot, and J.~G. Muga.
\newblock Shortcuts to adiabaticity: Concepts, methods, and applications.
\newblock {\em Rev. Mod. Phys.}, 91:045001, Oct 2019.

\bibitem{Toeplitz_Hausdorff}
K.~Gustafson.
\newblock The {T}oeplitz-{H}ausdorff theorem for linear operators.
\newblock {\em Proc. Amer. Math. Soc.}, 25:203--204, 1970.

\bibitem{GutkinJonckheereKarow}
E.~Gutkin, E.~A. Jonckheere, and M.~Karow.
\newblock Convexity of the joint numerical range: Topological and differential
  geometric viewpoints.
\newblock {\em Linear Algebra and Its Applications (LAA)}, 376C:143--171,
  November 2003.

\bibitem{Hausdorff}
F.~Hausdorff.
\newblock Der wertvorrat einer bilinearform.
\newblock {\em Math. Zeitschrift}, 3:314--316, 1919.

\bibitem{adiabatic_integer_factorization}
Narendra~N. Hegade, Koushik Paul, F.~Albarr\'an-Arriagada, Xi~Chen, and Enrique
  Solano.
\newblock Digitized adiabatic quantum factorization, Nov 2021.

\bibitem{Itay_Hen}
Itay Hen.
\newblock How quantum is the speedup in adiabatic unstructured search.
\newblock arXiv:1811.08302v2 [quant-ph] 14 Apr 2019, 2019.

\bibitem{exceptional_Mexican05}
E.~Hern\'andez, A.~J\'auregui, and A.~Mondrag\'on.
\newblock Energy eigenvalue surfaces close to a degeneracy of unbound states:
  crossings and anticrossings of energies and widths.
\newblock {\em Phys Rev E Stat Nonlin Soft Matter Phys.}, 72, Aug 2005.

\bibitem{exceptional_Mexican11}
E.~Hern\'andez, A.~J\'auregui, and A.~Mondrag\'on.
\newblock Exceptional points and non-hermitian degeneracy of resonances in a
  two-channel model.
\newblock {\em Phys Rev E Stat Nonlin Soft Matter Phys.}, 84:046209, Oct 2011.

\bibitem{Hirsch1976}
M.~W. Hirsch.
\newblock {\em Differential Topology}, volume~33 of {\em Graduate Texts in
  Mathematics}.
\newblock Springer-Verlag, New York, 1976.

\bibitem{exceptional_enhanced}
H.~Hodaei, A.~Hassan, S.~Wittek, and et~al.
\newblock Enhanced sensitivity at higher-order exceptional points.
\newblock {\em Nature}, 548:187–191, 2017.

\bibitem{Understanding_Quantum_Tunneling_through_Quantum_Mo}
S.~Isakov, G.~Mazzola, V.~N. Smelyanskiy, Z.~Jiang, S.~Boxio, H.~Nevens, and
  M.~Troyer.
\newblock Understanding quantum tunneling through quantum monte carlo.
\newblock arXiv [quant-phy] 27 Oct 2015, 2015.

\bibitem{Legendrian_unfolding}
S.~Izumiya.
\newblock The theory of graph-like {L}egendrian unfoldings and its
  applications.
\newblock {\em Journal of Singularities}, 12:53--79, 2015.

\bibitem{adiabatic}
E.~Jonckheere, A.~T. Rezakhani, and F.~Ahmad.
\newblock Differential topology of adiabatically controlled quantum processes.
\newblock {\em Quantum Information Processing (QINP)}, 12(3):1515--1538, 2013.

\bibitem{Jonckheere1997}
E.~A. Jonckheere.
\newblock {\em Algebraic and Differential Topology of Robust Stability}.
\newblock Oxford, New York, 1997.

\bibitem{JonckheereAhmadGutkin}
E.~A. Jonckheere, F.~Ahmad, and E.~Gutkin.
\newblock Differential topology of numerical range.
\newblock {\em Linear Algebra and Its Applications}, 279/1-3:227--254, August
  1998.

\bibitem{Nanaz_Alaska}
E.~A. Jonckheere and N.~Fathpour.
\newblock Algebraic {R}iccati equation and infinitesimal {V}-stability---a
  {G}r\"obner basis approach.
\newblock In {\em American Control Conference (ACC2002)}, pages 5138--5143,
  Anchorage, Alaska, May 08-10 2002.
\newblock Session FP17.

\bibitem{Nanaz_NotreDame}
E.~A. Jonckheere and N.~Fathpour.
\newblock Hamiltonian structure of the algebraic {R}iccati equation and its
  infinitesimal {V}-stability.
\newblock In {\em Fifteenth International Symposium on the Mathematical Theory
  of Networks and Systems (MTNS 2002)}, University of Notre-Dame, IN, August
  12-16 2002.
\newblock Session WM5, ``Algebraic and Differential Geometry in System Theory".

\bibitem{Kato1995}
T.~Kato.
\newblock {\em Perturbation Theory for Linear Operators}.
\newblock Classics in Mathematics. Springer, 1995.

\bibitem{KenneyLaubJonckheere}
C.~Kenney, A.~Laub, and E.~Jonckheere.
\newblock Positive and negative solutions of dual riccati equations by matrix
  sign function iteration.
\newblock {\em Systems and Control Letters}, 13:109--116, 1989.

\bibitem{primer_analytic}
S.~G. Krantz and H.~R. Parks.
\newblock {\em A Primer of Real Analytic Functions (Second Edition)}.
\newblock Birkh\"auser Advanced Texts, Boston, Basel, Berlin, 2002.

\bibitem{exceptional_nonreciprocal}
HK~Lau and AA~Clerk.
\newblock Fundamental limits and non-reciprocal approaches in non-hermitian
  quantum sensing.
\newblock {\em Nat Commun}, 9:4320, 2018.

\bibitem{Lickorish1997}
W.~B.~Raymond Lickorish.
\newblock {\em An Introduction to Knot Theory}.
\newblock Number 175 in Graduate Text in Mathematics. Springer, New York, 1997.

\bibitem{exceptional_tomography}
M.~Naghiloo, M.~Abbasi, Y.~N. Joglekar, and et~al.
\newblock Quantum state tomography across the exceptional point in a single
  dissipative qubit.
\newblock {\em Nat. Phys.}, 15:1232–1236, 2019.
\newblock Available at https://doi.org/10.1038/s41567-019-0652-z.

\bibitem{computable_Legendrian_invariants}
Lenhard~L. Ng.
\newblock Computable legendrian invariants.
\newblock {\em Topology}, 42(1):55--82, 2003.

\bibitem{oneill}
B.~O'Neill.
\newblock {\em Elementary Differential Geometry}.
\newblock Academic Press, 1997.

\bibitem{Thurston_Bennequin}
O.~Plamenevskaya.
\newblock Bounds for the {T}hurston-{B}ennequin number from {F}loer homology.
\newblock {\em Algebraic \& Geometric Topology}, 4:399--406, June 2002.

\bibitem{reichardt-adiabatic}
B.~W. Reichardt.
\newblock The quantum adiabatic optimization algorithm and local minima.
\newblock In {\em STOC '04, Proceedings of the thirty-sixth annual ACM
  symposium on Theory of Computing}, pages 502--510, Chicago, IL, 2004.

\bibitem{Maslov_index_for_paths}
Joel Robbin and Dietmar Salamon.
\newblock The {M}aslov index for paths.
\newblock {\em Topology}, 32(4):827--844, 1993.

\bibitem{augmentations_and_rulings_of_Legendrian_knots}
J.~M. Sabloff.
\newblock What is a legendrian knot?
\newblock {\em Notices of the {AMS}}, 56(10):1282--1284, November 2009.

\bibitem{robust_performance_open}
Sophie~G. Schirmer, Frank~C. Langbein, Carrie~A. Weidner, and Edmond
  Jonckheere.
\newblock Robust control performance for open quantum systems.
\newblock {\em IEEE Transactions on Automatic Control}, 67(11):6012--6024,
  November 2022.
\newblock Available at arXiv:2008.13691 [math.OC] 31 Aug 2020.

\bibitem{SilvermanBucyDolezal}
L.~M. Silverman and R.~S. Bucy.
\newblock Generalizations of a theorem of dole\v{z}al.
\newblock {\em Mathematical Systems Theory}, 4:334–339, 1970.

\bibitem{newest_from_Weiss}
Ilya~M. Spitkovsky and Stephan Weiss.
\newblock Signatures of quantum phase transitions from the boundary of the
  numerical range.
\newblock {\em Journal of Mathematical Physics}, 59(12):121901--1--21, December
  2018.

\bibitem{Tabachnikov_original}
S.~L. Tabachnikov.
\newblock Around four vertices.
\newblock {\em Uspekhi Mat. Nauk}, 45(1 (271)):191--192, 1990.

\bibitem{4_vertex_revisited}
V.~Tabachnikov.
\newblock The four-vertex theorem revisited---two variations on the old theme.
\newblock {\em American Mathematical Monthly}, 102(10):912--916, December 1995.

\bibitem{Toeplitz}
O.~T\"oplitz.
\newblock Das algebraische analogon zu einem satze von {F}ej\'er.
\newblock {\em Math. Zeitschrift}, 2:187--197, 1918.

\bibitem{generating_function_polynomials}
L.~Traynor.
\newblock Generating function polynomials for {L}egendrian links.
\newblock {\em Geometry \& Topology}, 5:719--760, October 2001.

\bibitem{von_neumann_wigner}
J.~von Neumann and E.~Wigner.
\newblock {\"U}ber das {V}erhalten von {E}igenwerten bei {A}diabatischen
  {P}rozessen.
\newblock {\em Phys. Zschr.}, 30:467--470, 1929.

\bibitem{Wacker1978}
H.~Wacker, editor.
\newblock {\em Continuation Methods}.
\newblock Academic Press Rapid Manuscript Reproduction. Academic Press, New
  York, 1978.
\newblock Proceedings of a Symposium at the University of Linz, Austria,
  October 3-4, 1977.

\bibitem{adiabatic_nature}
Chi Wang, Huo Chen, and Edmond Jonckheere.
\newblock Quantum versus simulated annealing in wireless interference network
  optimization.
\newblock {\em Nature Scientific Reports}, 6:1--9, May 2016.

\bibitem{quantum_wireless_II}
Chi Wang and Edmond Jonckheere.
\newblock Simulated versus reduced noise quantum annealing in maximum
  independent set solution to wireless network scheduling.
\newblock {\em Quantum Information Processing (QINP)}, 18(6):1--25, January
  2019.
\newblock First Online: 17 November 2018.

\bibitem{Chi_embedding}
Chi Wang, Edmond Jonckheere, and Todd Brun.
\newblock Differential geometric treewidth estimation in adiabatic quantum
  computation.
\newblock {\em Quantum Information Processing}, July 19 2016.
\newblock doi:10.1007/s11128-016-1394-9.

\bibitem{exceptional_sensing}
J.~Wiersig.
\newblock Prospects and fundamental limits in exceptional point-based sensing.
\newblock {\em Nature Communications}, 11:2454, 2020.
\newblock Available at https://doi.org/10.1038/s41467-020-16373-8.

\bibitem{Thurston_Bennequin_Maslov}
Hao {Wu}.
\newblock {A New Way to Compute the {T}hurston-{B}ennequin Number}.
\newblock {\em arXiv Mathematics e-prints}, page math/0312224, December 2003.

\bibitem{Itay}
B.~H. Zhang, G.~Wagenbreth, V.~Martin-Mayor, and I.~Hen.
\newblock The fair and unfair quantum ground-state sampling.
\newblock arXiv:1701.01524v1 [quant-ph] 6 Jan 2017, 2017.

\bibitem{Zhou}
K.~Zhou and J.~C. Doyle.
\newblock {\em Essentials of robust control}.
\newblock Prentice Hall, Upper Saddle River, NJ, 1998.

\bibitem{exceptional_Czech}
Miloslav Znojil.
\newblock Perturbation theory near degenerate exceptional points.
\newblock {\em Symmetry}, 12(8), 2020.

\end{thebibliography}


\end{document}